\newif\ifnotes
\newif\ifcr
\notesfalse
\crfalse

\newcommand{\dnote}[1]{\textsf{\color{blue} Dakshita: { #1}}}

\newcommand{\anote}[1]{\textsf{\color{violet} Amit: { #1}}}

\documentclass[10pt]{article}
\usepackage{palatino}
\usepackage{amsthm}
\usepackage{fullpage}
\usepackage{amsmath,amssymb,enumerate,algorithm,algorithmic,latexsym}
\usepackage{tikz}
\usetikzlibrary{fpu}
\usepackage{footmisc}
\usepackage{breakcites}
\usepackage{float}
\newfloat{algorithm}{H}{lop}
\usepackage{color}
\usepackage{colortbl}
\usepackage{url}
\usepackage{graphicx}
\usepackage{xspace}
\usepackage{subfigure}
\usepackage{amsfonts}
\usepackage{caption}
\usepackage{enumitem}
\usepackage{soul}
\usepackage[normalem]{ulem}
\usepackage{url}
\usepackage{graphicx}
\usepackage{bookmark}
\usepackage{boxedminipage}
\usepackage{mathtools}
\usepackage[capitalize]{cleveref}
\numberwithin{algorithm}{section}

\renewcommand{\paragraph}[1]{\vspace{1.5mm}\noindent \textbf{#1}}

\newcommand{\hyb}{\boldsymbol{\mathcal{H}}}

\newcommand{\poly}{\mathsf{poly}}
\newcommand{\negl}{\mathsf{negl}}
\newcommand{\bN}{\mathbb{N}}
\newcommand{\bZ}{\mathbb{Z}}
\newcommand{\M}{\mathsf{M}}
\newcommand{\R}{\mathsf{R}}
\newcommand{\Q}{\mathsf{Q}}
\newcommand{\td}{\mathsf{td}}
\newcommand{\msg}{\mathsf{msg}}
\newcommand{\key}{\mathsf{key}}
\newcommand{\lock}{\mathsf{lk}}
\newcommand{\corrupt}{\mathbb{S}}
\newcommand{\honest}{\mathbb{H}}

\newcommand{\LL}{\mathcal{L}^{(\PECom)}}
\newcommand{\LLL}{\mathcal{L}^{(\mathsf{both})}}

\newcommand{\D}{\mathsf{D}}
\newcommand{\REAL}{\mathsf{REAL}}
\newcommand{\IDEAL}{\mathsf{IDEAL}}
\newcommand{\na}{\mathsf{No}\text{ }\mathsf{Abort}}
\newcommand{\sh}{\mathsf{sh}}
\newcommand{\bA}{\mathbf{A}}
\newcommand{\bB}{\mathbf{B}}
\newcommand{\bC}{\mathbf{C}}
\newcommand{\bM}{\mathbf{M}}
\newcommand{\bG}{\mathbf{G}}
\newcommand{\bS}{\mathbf{S}}
\newcommand{\bE}{\mathbf{E}}
\newcommand{\bs}{\mathbf{s}}
\newcommand{\be}{\mathbf{e}}
\newcommand{\bu}{\mathbf{u}}
\newcommand{\bv}{\mathbf{v}}
\newcommand{\bb}{\mathbf{b}}

\renewcommand{\O}{\mathcal{O}}
\newcommand{\obfC}{\widetilde{\mathbf{CC}}}
\newcommand{\CC}[3]{\ensuremath{\mathbf{CC}[#1,#2,#3]}}
\newcommand{\ccSim}{\mathsf{Sim}}
\newcommand{\obf}{\mathsf{Obf}}

\newcommand{\fhe}{\mathsf{FHE}}
\newcommand{\mfhe}{\mathsf{MFHE}}
\newcommand{\ahe}{\mathsf{AHE}}
\newcommand{\convert}{\mathsf{Convert}}

\newcommand{\fheKg}{\fhe.\mathsf{KeyGen}}
\newcommand{\fheE}{\fhe.\mathsf{Enc}}
\newcommand{\fheD}{\fhe.\mathsf{Dec}}
\newcommand{\fheQE}{\fhe.\mathsf{QEnc}}
\newcommand{\fheQD}{\fhe.\mathsf{QDec}}
\newcommand{\fheEv}{\fhe.\mathsf{Eval}}
\newcommand{\fheK}{\mathsf{sk}}
\newcommand{\fhePK}{\mathsf{pk}}

\newcommand{\ct}{\mathsf{ct}}
\newcommand{\qmfhe}{\mathsf{QMFHE}}
\newcommand{\qcmfhe}{\mathsf{QCMFHE}}
\newcommand{\keygen}{\mathsf{KeyGen}}
\newcommand{\enc}{\mathsf{Enc}}
\newcommand{\qenc}{\mathsf{QEnc}}
\newcommand{\dec}{\mathsf{Dec}}
\newcommand{\qdec}{\mathsf{QDec}}
\newcommand{\eval}{\mathsf{Eval}}
\newcommand{\pk}{\mathsf{pk}}
\newcommand{\sk}{\mathsf{sk}}
\newcommand{\setup}{\mathsf{Setup}}
\newcommand{\pp}{\mathsf{pp}}

\newcommand{\gentrap}{\mathsf{GenTrap}}
\newcommand{\invert}{\mathsf{Invert}}
\newcommand{\spooky}{\mathsf{Spooky}}
\newcommand{\qspooky}{\mathsf{QSpooky}}

\newcommand{\cmt}{\mathsf{cmt}}
\newcommand{\Com}{\mathsf{Com}}
\newcommand{\COM}{\mathsf{Com}}

\newcommand{\EC}{\mathsf{EC}}
\newcommand{\comS}{\EC.\mathsf{S}}
\newcommand{\commS}{\mathsf{S}^*}
\newcommand{\comR}{\EC.\mathsf{R}}
\newcommand{\commR}{\mathsf{R}^*}
\newcommand{\comV}{\EC.\mathsf{V}}
\newcommand{\cm}{\mathsf{c}}
\newcommand{\PECom}{\mathsf{eCom}}
\newcommand{\nmCom}{\mathsf{nmCom}}
\newcommand{\com}{\mathsf{c}}
\newcommand{\rec}{\mathsf{r}}

\newcommand{\zk}{\mathsf{ZK}}
\newcommand{\A}{\mathsf{A}}
\newcommand{\prot}[2]{\ve{#1,#2}}
\newcommand{\zkE}{\mathsf{E}}
\newcommand{\zkP}{\mathsf{P}}
\newcommand{\zkSim}{\mathsf{Sim}}
\newcommand{\zkSimAbort}{\mathsf{SimAbort}_\bot}
\newcommand{\zkSimNoAbort}{\mathsf{SimNoAbort}_\bot}
\newcommand{\cfSimAbort}{\mathsf{SimAbort}_\bot}
\newcommand{\cfSimNoAbort}{\mathsf{SimNoAbort}_\bot}

\newcommand{\eC}{\mathsf{C}}
\newcommand{\eD}{\mathsf{D}}
\newcommand{\eR}{\mathsf{R}}
\newcommand{\dk}{\mathsf{dk}}

\newcommand{\sfegen}{\mathsf{SFE.Gen}}
\newcommand{\sfeenc}{\mathsf{SFE.Enc}}
\newcommand{\sfeeval}{\mathsf{SFE.Eval}}
\newcommand{\sfedec}{\mathsf{SFE.Dec}}
\newcommand{\sfe}{\mathsf{SFE}}

\newcommand{\wiP}{\mathsf{WI.P}}

\newcommand{\zkR}{\mathsf{R}}
\newcommand{\zkS}{\mathsf{S}}

\newcommand{\zkmD}{\mathsf{D}}

\newcommand{\zkmP}{\zkP^*}

\newcommand{\zkV}{\mathsf{V}}
\newcommand{\zkmV}{\zkV^*}
\newcommand{\wiV}{\mathsf{WI.V}}

\newcommand{\viewval}{\mathsf{View}\text{-}\mathsf{Val}}
\newcommand{\hview}{\mathsf{HOUT}}
\newcommand{\wi}{\mathsf{wi}}

\newcommand{\NP}{\mathsf{NP}}

\newcommand{\cdsSch}{\mathsf{CDS}}

\newcommand{\lang}{\mathcal{L}}

\newcommand{\rel}{\mathcal{R}}
\newcommand{\ins}{x}

\newcommand{\wit}{w}

\newcommand{\RL}{\rel_{\lang}}

\newcommand\ket[1]{| #1 \rangle}

\newcommand{\textabbrevstyle}[1]{\mbox{#1}}
\newcommand{\textabbrevstylebol}[1]{\mbox{\textbf{#1}}}
\newcommand{\newtextabbrev}[1]{\expandafter\newcommand\csname #1\endcsname{\textabbrevstyle{#1}\xspace}}
\newcommand{\newtextabbrevbol}[1]{\expandafter\newcommand\csname #1\endcsname{\textabbrevstylebol{#1}\xspace}}
\newcommand{\renewtextabbrevbol}[1]{\expandafter\renewcommand\csname
#1\endcsname{\textabbrevstylebol{#1}\xspace}}

\newtextabbrevbol{EXP}
\newtextabbrevbol{Dtime}
\newtextabbrev{PRG}
\newtextabbrev{PRF}
\newtextabbrev{OWF}
\newtextabbrev{PPT}
\newtextabbrev{EOWF}
\newtextabbrev{GEOWF}
\newtextabbrev{ZAP}
\newtextabbrevbol{Ntime}
\newtextabbrevbol{coAM}
\newtextabbrevbol{SZK}
\newtextabbrevbol{BPP}
\newtextabbrevbol{coNP}
\newtextabbrevbol{coMA}
\newtextabbrevbol{TFNP}
\newtextabbrev{NIWI}
\newtextabbrevbol{AM}
\newtextabbrev{SNARG}
\newtextabbrev{SNARK}
\newtextabbrev{coRP}
\newtextabbrev{NIUA}
\newtextabbrev{UA}
\newtextabbrev{ZK}
\newtextabbrev{WZK}
\newtextabbrev{WH}
\newtextabbrev{AI}
\newtextabbrev{ECRH}
\newtextabbrev{PIR}
\newtextabbrev{WI}
\newtextabbrev{WIPOK}
\newtextabbrev{CDS}
\newtextabbrev{POK}
\newtextabbrev{FE}
\newtextabbrev{IO}
\newtextabbrev{XIO}
\newtextabbrev{SXIO}
\newtextabbrev{SKFE}
\newtextabbrev{PKFE}
\newtextabbrev{PPRF}
\newtextabbrev{LWE}
\newtextabbrev{PKE}

\newtheorem{definition}{Definition}[section]

\newtheorem{lemma}{Lemma}[section]

\newtheorem{theorem}{Theorem}[section]
\newtheorem{claim}{Claim}[section]

\theoremstyle{remark}

\newtheorem{remark}{Remark}[section]

\newcommand{\defref}[1]{Definition~\protect\ref{#1}}

\newcommand{\remref}[1]{Remark~\protect\ref{#1}}
\newcommand{\proref}[1]{Protocol~\protect\ref{#1}}

\newenvironment{boxfig}[2]{\begin{figure}[#1]\fbox{\begin{minipage}{\linewidth}
                        \vspace{0.2em}
                        \makebox[0.025\linewidth]{}
                        \begin{minipage}{0.95\linewidth}
            {{
                        #2 }}
                        \end{minipage}
                        \vspace{0.2em}
                        \end{minipage}}}{\end{figure}}

\newenvironment{boxedalgo}
  {\begin{center}\begin{boxedminipage}{0.95\textwidth}}
  {\end{boxedminipage}\end{center}}

\newcommand{\pprotocol}[4]{
\begin{boxfig}{h!}{
\begin{center}
\textbf{#1}
\end{center}
    #4
\vspace{0.2em} } \caption{\label{#3} #2}
\end{boxfig}
}

\newcommand{\protocol}[4]{
\pprotocol{#1}{#2}{#3}{#4} }

\newcommand{\ve}[1]{\langle #1 \rangle}

\newcommand{\set}[1]{\left\{#1\right\}}
\newcommand{\abs}[1]{\left|#1\right|}

\renewcommand{\]}{\right ]}

\renewcommand{\[}{\left [}
\newcommand{\pST}{\; \middle\vert \;}

\newcommand{\zo}{\{0,1\}}

\newcommand{\Nat}{\mathbb{N}}
\newcommand{\secp}{\lambda}
\newcommand{\tagg}{\ensuremath{\mathsf{tag}}\xspace}

\newcommand{\simulator}{\ensuremath{\mathsf{Sim}}\xspace}

\newcommand{\bbN}{\ensuremath{\mathbb{N}}\xspace}

\newcommand{\cC}{\ensuremath{\mathcal{C}}\xspace}
\newcommand{\cP}{\ensuremath{\mathcal{P}}\xspace}
\newcommand{\cA}{\ensuremath{\mathcal{A}}\xspace}
\newcommand{\cR}{\ensuremath{\mathcal{R}}\xspace}
\newcommand{\cV}{\ensuremath{\mathcal{V}}\xspace}
\newcommand{\cS}{\ensuremath{\mathcal{S}}\xspace}
\newcommand{\cL}{\ensuremath{\mathcal{L}}\xspace}

\newcommand{\cE}{\ensuremath{\mathsf{Ext}}\xspace}

\def\mim{\ensuremath{\mathsf{MIM}}\xspace}
\def\adv{\ensuremath{\mathsf{ADV}}\xspace}

\def\state{\ensuremath{\mathsf{st}}\xspace}
\def\view{\ensuremath{\mathsf{VIEW}}\xspace}
\def\output{\ensuremath{\mathsf{OUT}}\xspace}

\def\reject{\ensuremath{\mathsf{Reject}}\xspace}
\def\noreject{\ensuremath{\mathsf{NoReject}}\xspace}
\def\fail{\ensuremath{\mathsf{Fail}}\xspace}

\newcommand{\dist}[1]{\ensuremath{\langle{#1}\rangle}\xspace}

\newcommand{\cfail}{\ensuremath{\mathsf{check}\text{-}\mathsf{fail}}}

\newcommand{\nmcsmall}{\ensuremath{\mathsf{nmCom}}\xspace}

\newcommand{\piagk}{\proref{fig:tag_amplification_nmcom}\xspace}

\newcommand{\pibs}{\ensuremath{\Pi_{\mathsf{zk}}}\xspace}

\newcommand{\ie}{\text{i.e.}\xspace}
\newcommand{\suchthat}{\text{s.t.}\xspace}

\newcommand{\comb}{\mathsf{comb}\xspace}

\newcommand\TODO[1]{{\color{red}{#1}}\xspace}

\newcommand{\sender}{\ensuremath{\mathsf{S}}\xspace}
\newcommand{\receiver}{\ensuremath{\mathsf{R}}\xspace}
\newcommand{\extractor}{\ensuremath{\mathsf{Ext}}\xspace}

\title{Post-Quantum Multi-Party Computation}

\author{}
\date{}

\author{
\begin{tabular}{c@{\hskip 1in}c@{\hskip 1in}c} Amit Agarwal\thanks{UIUC. \url{amita2@illinois.edu}}  & James Bartusek\thanks{UC Berkeley. \url{bartusek.james@gmail.com}}
& Vipul Goyal\thanks{CMU. \url{vipul@cmu.edu} }\\ \end{tabular}\\ \\
\begin{tabular}{c@{\hskip 1in}c}
Dakshita Khurana\thanks{UIUC. \url{dakshita@illinois.com}}  & 
Giulio Malavolta\thanks{Max Planck Institute for Security and Privacy. \url{giulio.malavolta@hotmail.it}
}  \\ \end{tabular}
}

\begin{document}
\maketitle

\begin{abstract}
We initiate the study of multi-party computation for classical functionalities (in the plain model) with security against malicious polynomial-time quantum adversaries. We observe that existing techniques readily give a polynomial-round protocol, but our main result is a construction of \emph{constant-round} post-quantum multi-party computation. We assume mildly super-polynomial quantum hardness of learning with errors (LWE), and polynomial quantum hardness of an LWE-based circular security assumption. 
Along the way, we develop the following cryptographic primitives that may be of independent interest:
\begin{itemize}
    \item A spooky encryption scheme for relations computable by quantum circuits, from the quantum hardness of an LWE-based circular security assumption. This yields the first quantum multi-key fully-homomorphic encryption scheme with classical keys.
    \item Constant-round zero-knowledge secure against multiple parallel quantum verifiers from spooky encryption for relations computable by quantum circuits.

    To enable this, we develop a new straight-line non-black-box simulation technique against {\em parallel} verifiers that does not clone the adversary's state. This forms the heart of our technical contribution and may also be relevant to the classical setting. 
    \item A constant-round post-quantum non-malleable commitment scheme, from the mildly super-polynomial quantum hardness of LWE. 
\end{itemize}
\end{abstract}


\thispagestyle{empty}
\newpage
\tableofcontents
\thispagestyle{empty}
\newpage
\pagenumbering{arabic}

\section{Introduction}
Secure multi-party computation (MPC) allows a set of parties to compute a
joint function of their inputs, revealing only the output of the function while keeping their inputs private. 
General secure MPC, initiated in works such as~\cite{Yao86,STOC:GolMicWig87,STOC:BenGolWig88,C:ChaCreDam87}, has played a central role in modern theoretical cryptography.
The last few years have seen tremendous research optimizing MPC in various ways, enabling a plethora of practical applications that include joint computations on distributed medical data, privacy-preserving machine learning, e-voting, distributed key management, among others.
The looming threat of quantum computers naturally motivates the problem of constructing protocols with {\em provable security against quantum adversaries}.

After Watrous' breakthrough work on zero-knowledge against quantum adversaries \cite{10.1137/060670997}, the works of~\cite{AC:DamLun09,AFRICACRYPT:LunNie11,C:HalSmiSon11} considered variants of quantum-secure computation protocols, in the \emph{two}-party setting. Very recently, Bitansky and Shmueli~\cite{BS20} obtained the first \emph{constant-round} classical zero-knowledge arguments with security against quantum adversaries. 
Their techniques (and those of~\cite{EPRINT:ALP19} in a concurrent work) are based on the recent non-black-box simulation technique of~\cite{STOC:BKP19}, who constructed two-message {\em classically-secure} weak zero-knowledge in the plain model. Unfortunately, it is unclear whether these protocols compose under parallel repetition. As a result, they become largely inapplicable to the constant-round multi-party setting. 

There has also been substantial effort in constructing protocols for securely computing quantum circuits~\cite{C:DupNieSal10,C:DupNieSal12,DBLP:conf/eurocrypt/DulekGJMS20} (see Section \ref{sec:related} for further discussion). But to the best of our knowledge, generic multi-party computation protocols with classical communication and security against quantum adversaries have only been studied in models with \emph{trusted pre-processing or setup}. To make things even worse,~\cite{DBLP:conf/eurocrypt/DulekGJMS20} construct a maliciously-secure multi-party protocol for computing quantum ciruits, assuming the existence of maliciously-secure post-quantum classical MPC. This means that the only available implementations of such a building block require trusted pre-processing or a common reference string.


\paragraph{Post-Quantum MPC.} In this work we initiate the study of MPC protocols that allow classical parties to securely compute general classical functionalities, and where security is guaranteed against \emph{malicious quantum adversaries}. Our focus is on MPC in the \emph{plain model, with a dishonest majority}: Fully classical participants interact with each other with no access to trusted/pre-processed parameters or a common reference string. Multi-party protocols achieving security in this natural setting do not seem to have been previously analyzed in {\em any} number of rounds.
We stress that the challenges of proving post-quantum security of MPC protocols stretch far beyond the appropriate instantiations of the cryptographic building blocks (e.g.\ avoiding factoring or discrete logarithm-based cryptosystems):
in fact, it is possible to devise protocols~\cite{EPRINT:ALP19,FOCS:AmbRosUnr14} based on entirely post-quantum assumptions that are secure against classical adversaries but completely insecure against quantum adversaries.

Because quantum information behaves very differently from classical information, designing post-quantum protocols often requires new techniques to achieve provable security. As an example, a common strategy to prove classical security of MPC protocols is to define a simulator that can extract the inputs of the corrupted parties by ``rewinding'' them, i.e.\ taking a snapshot of the state of the adversary and splitting the protocol execution into multiple branches. However, when the adversary is a quantum machine, this technique becomes largely inapplicable since the no-cloning theorem (one of the fundamental principles of quantum mechanics) prevents us from creating two copies of an arbitrary quantum state. One of our key contributions is a new {\em parallel no-cloning non-black-box simulation technique} that extends the work of~\cite{BS20}, to achieve security against multiple parallel quantum verifiers.

\subsection{Our Results}
We begin by summarizing our main result: Classical multi-party computation with security against quantum circuits in the plain model.  
Here, parties communicate classically via authenticated point-to-point channels as well as broadcast channels, where everyone can send messages in the same round. 
In each round, all parties simultaneously exchange messages. 
The network is assumed to be synchronous with rushing adversaries, i.e. adversaries may generate their messages for any round after observing the messages of all honest parties in that round, but before observing the messages of honest parties in the next round.
The (quantum) adversary may corrupt upto all but one of the participants.
In this model, we obtain the following main result.


\begin{theorem}[Informal]
Assuming mildly super-polynomial quantum hardness of LWE and AFS-spooky encryption for relations computable by polynomial-size quantum circuits,
there exists a constant-round classical MPC protocol (in the plain model) maliciously secure against quantum polynomial-time adversaries. 
\end{theorem}

In more detail, our protocol is secure against any adversary $\A = \{\A_\secp,\rho_\secp\}_\secp$, where each $\A_\secp$ is the (classical) description of a polynomial-size quantum circuit and $\rho_\secp$ is some (possibly inefficiently computable) non-uniform quantum advice. Beyond being interesting in its own right, our plain-model protocol may serve as a useful stepping stone to obtaining interesting protocols for securely computing quantum circuits in the plain model, as evidenced by the work of~\cite{DBLP:conf/eurocrypt/DulekGJMS20}.
This protocol is constructed in Sections \ref{sec:coin-tossing} and \ref{sec:mpc}.

By ``mildly'' super-polynomial quantum hardness of LWE, we mean to assume that there exists a constant $c \in \mathbb{N}$, such that for large enough security parameter $\secp \in \mathbb{N}$, no quantum polynomial time algorithm can distinguish LWE samples from uniform with advantage better than $\negl(\secp^{\mathsf{ilog}(c,\secp)})$, where $\mathsf{ilog}(c,\secp)$ denotes the $c$-times iterated logarithm 
${\log \log \cdots_{{c}~\mathrm{times}}}(\secp)$.
We note that this is weaker than assuming the quasi-polynomial quantum hardness of LWE, i.e. the assumption that quantum polynomial-time adversaries cannot distinguish LWE samples from uniform with advantage better than $2^{-{(\log \secp)}^c}$ for some constant $c > 1$.


A key technical ingredient in our work is an additive function sharing (AFS) spooky encryption scheme \cite{C:DHRW16} for relations computable by quantum circuits. An AFS-spooky encryption scheme has a publicly-computable algorithm that, on input a set of ciphertexts $\enc(\pk_1, m_1),\allowbreak \dots,\allowbreak \enc(\pk_n, m_n)$ encrypted under \emph{independently sampled} public keys and a (possibly quantum) circuit $C$, computes a new set of ciphertexts
$$
\enc(\pk_1, y_1), \dots, \enc(\pk_n, y_n) \textit{ s.t. }\mathop{\bigoplus}\limits_{i=1}^n y_i = C(m_1,\dots,m_n).
$$
In Section \ref{sec:spooky} we show how to construct AFS-spooky encryption for relations computable by quantum circuits, under an LWE-based circular security assumption. We refer the reader to~\cref{subsec:keyswitch} for the exact circular security assumption we need, which is similar to the one used in~\cite{FOCS:Mahadev18b}.  As a corollary, this immediately yields the first multi-key fully-homomorphic encryption~\cite{STOC:LopTroVai12} for quantum circuits with classical key generation and classical encryption of classical messages.


\begin{theorem}[Informal]
Under an appropriate LWE-based circular security assumption, there exists an AFS-spooky encryption scheme for relations computable by polynomial-size quantum circuits with classical key generation and classical encryption of classical messages.
\end{theorem}
Our most important technical contribution is a construction of  constant-round zero-knowledge arguments against parallel quantum verifiers, and constant-round extractable commitments against parallel quantum committers. 
Here, we develop a novel {\em parallel no-cloning non-black-box simulation} technique. This uses as a starting point the recently introduced no-cloning technique of~\cite{BS20,EPRINT:ALP19}, which in turns builds on the classical non-black-box technique of Bitansky, Khurana and Paneth~\cite{STOC:BKP19}.

We point out that we do not obtain protocols that compose under \emph{unbounded} parallel repetition. Instead we build a bounded variant in the multi-party setting (that we also refer to as multi-verifier zero-knowledge and multi-committer extractable commitments) that suffices for our application to constant round MPC.
Our technique makes crucial use of AFS-spooky encryption for relations computable by classical circuits.
Parallel extractable commitments and zero-knowledge are formally constructed and analyzed in Sections \ref{sec:pecom} and \ref{sec:pzk}, respectively. 

\begin{theorem}[Informal]
Assuming the quantum polynomial hardness of LWE and the existence of AFS-spooky encryption for relations computable by polynomial-size quantum circuits, there exists:
\begin{itemize}
    \item A constant-round classical argument for NP that is computational-zero-knowledge against parallel quantum polynomial-size verifiers.
    \item A constant-round classical commitment that is extractable against parallel quantum polynomial-size committers.
\end{itemize}
\end{theorem}

In addition, we initiate the study of post-quantum non-malleable commitments. Specifically, we construct and rely on constant-round post-quantum non-malleable commitments based on the super-polynomial hardness assumption described above. The formal construction and analysis can be found in Section \ref{sec:nmc}.  

\begin{theorem}[Informal]
Assuming the mildly super-polynomial quantum hardness of LWE and the existence of fully-homomorphic encryption for quantum circuits, there exists a constant-round non-malleable commitment scheme secure against quantum polynomial-size adversaries.
\end{theorem}

We also obtain quantum-secure non-malleable commitments in $O(\mathsf{ilog}(c,\secp))$ rounds for any constant $c \in \mathbb{N}$ based on any (polynomially) quantum-secure extractable commitment. In particular, plugging in these commitments instead of our constant round non-malleable commitments gives an $O(\mathsf{ilog}(c,\secp))$ round quantum-secure MPC from any quantum AFS-spooky encryption scheme.




\section{Technical Overview}
\subsection{Background}\label{subsec:tech-background}
Our starting point is any constant-round post-quantum maliciously-secure MPC protocol in the (programmable) common random string (CRS) model. A long line of work has studied constant-round MPC in the CRS model~\cite{STOC:CLOS02,C:IshPraSah08,EC:AJLTVW12,EC:MukWic16,EC:BenLin18,EC:GarSri18a}, and these protocols can all be instantiated with primitives that are plausibly quantum-secure. One method of arguing that the resulting protocol is post-quantum secure is to demonstrate that (i) the simulator does not rewind or clone the adversary's state, and (ii) the reductions to the underlying quantum-secure primitives used to establish indistinguishability of the real and simulated world do not rewind or clone the adversary's state. As a concrete example, since the simulators and reductions in~\cite{EC:GarSri18a} are non-rewinding and non-cloning, implicit in their work is the theorem that any two-message post-quantum maliciously-secure oblivious transfer (OT) in the CRS model with a straight-line simulator implies a two-round post-quantum maliciously-secure MPC protocol in the CRS model. Such an OT is known for example from the quantum hardness of learning with errors~\cite{C:PeiVaiWat08}.


Thus, a natural approach to achieving post-quantum MPC in the plain model is to then securely implement a multi-party functionality that generates the aforementioned CRS. Specifically, we would like a set of $n$ parties to jointly execute a {\em coin-flipping protocol}. Such a protocol outputs a uniformly random string that may then be used to implement a post-quantum MPC protocol in the CRS model.
The programmability requirement on the CRS roughly translates to ensuring that for any quantum adversary, there exists a simulator that on input a random string $s$, can force the output of the coin-flipping protocol to be equal to $s$. A protocol satisfying this property is often referred to as a {\em fully-simulatable} multi-party coin-flipping protocol.



\paragraph{Post-Quantum Multi-Party Coin-Flipping.}
Existing constant-round protocols ~\cite{FOCS:Wee10,STOC:Goyal11}
for multi-party coin-flipping against classical adversaries make use of the following template.
Each participant first commits to a uniformly random string using an appropriate perfectly binding commitment.\footnote{We actually require this commitment to also satisfy a property called {\em non-malleability}, which we discuss later in this section.} In a later phase, all participants reveal the values they committed to, without actually revealing the randomness used for commitment.
Additionally, each participant proves (in zero-knowledge) to every other participant that they opened to the same value that they originally committed to.
If all zero-knowledge arguments verify, the protocol output is computed as the sum of the openings of all participants.

But not every classically secure zero-knowledge argument based on post-quantum assumptions is post-quantum secure. Building on prior work~\cite{EPRINT:ALP19}, in Appendix \ref{app:zkcounterexample}, we outline a ZK argument that is classically secure, and is based entirely on post-quantum assumptions (LWE), but is not post-quantum secure.

To highlight challenges in constructing constant-round protocols, we elaborate on the template discussed above and outline a simple polynomial-round coin tossing protocol. Readers familiar with this template for multi-party coin-tossing may skip a page.

\paragraph{A Simple Protocol in Polynomially Many Rounds.}
In order to motivate the challenges involved in constructing a post-quantum \textit{constant-round} multiparty coin tossing protocol, we first outline a simple protocol that requires \textit{polynomially many} rounds, and follows from ideas in existing work.
Our starting point is the polynomial-round post-quantum zero-knowledge protocol due to Watrous~\cite{10.1137/060670997}. 
Ideas developed in~\cite{BS20} almost immediately convert this to a post-quantum extractable commitment, assuming polynomial hardness of LWE (or, more generally, any post-quantum oblivious transfer). 
For completeness, we outline how this is done in Appendix~\ref{app:poly-round}.

Next, it is possible to use the resulting post-quantum secure extractable commitment to obtain post-quantum multi-party fully-simulatable coin flipping, that admits a straight-line simulator in the dishonest majority setting. The protocol requires rounds that grow linearly with the number of parties and polynomially with the security parameter, and is described in Figure \ref{fig:ct-intro}. 
At a very high level, the protocol requires each party to sample uniform randomness. Then each party sequentially commits (via an extractable commitment) to the randomness it sampled. In the next step, all parties broadcast their randomness in the clear, together with (sequential) zero-knowledge proofs by each party that the broadcasted randomness is consistent with the randomness that was previously committed.

\begin{figure}[ht!]
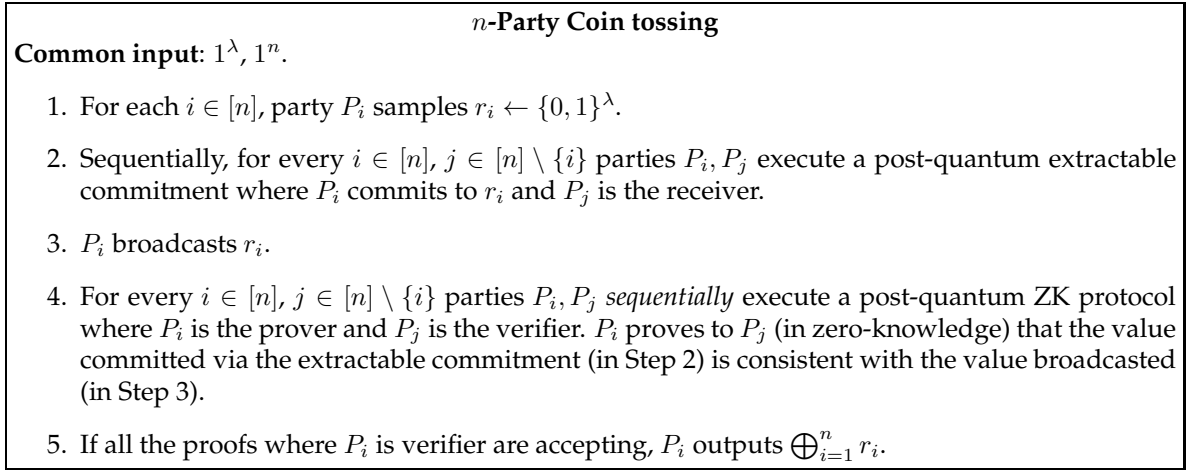

\begin{boxedalgo} 
\begin{center}
\textbf{$n$-Party Coin tossing}
\end{center}
\textbf{Common input}: $1^\secp$, $1^n$.
\begin{enumerate}
\item For each $i \in [n]$, party $P_i$ samples $r_i \gets \zo^\secp$.

\item Sequentially, for every $i \in [n]$, $j \in [n] \setminus \{i\}$ parties $P_i, P_j$ execute 
a post-quantum extractable commitment where $P_i$ commits to $r_i$ and $P_j$ is the receiver. 

\item $P_i$ broadcasts $r_i$.

\item For every $i \in [n]$, $j \in [n] \setminus \{i\}$ parties $P_i, P_j$ {\em sequentially} execute 
a post-quantum ZK protocol where $P_i$ is the prover and $P_j$ is the verifier. $P_i$  proves to $P_j$ (in zero-knowledge) that the value committed via the extractable commitment (in Step 2) is consistent with the value broadcasted (in Step 3).

\item If all the proofs where $P_i$ is verifier are accepting, $P_i$ outputs $\bigoplus_{i=1}^{n} r_i$.
\end{enumerate}
\end{boxedalgo}
\caption{Multiparty Coin Tossing}
\label{fig:ct-intro}
\end{figure}


Recall that the simulator $\simulator$ of any coin-flipping protocol obtains a uniformly random string $r^*$ from the ideal functionality, and must force this value as the output. 
We will briefly describe the construction of $\simulator$ for the case where $\cA$ controls $n-1$ parties and $\simulator$ plays the role of the only honest party $P_1$ (the same technique can be easily extended to the case where $\cA$ controls any arbitrary subset of parties). To do so, 
The $\simulator$ for the protocol in Figure \ref{fig:ct-intro} samples $r_i$ uniformly at random on behalf of each honest party $P_i$, and commits to $r_i$ in Step 2 following honest sender strategy. At the same time, $\simulator$ runs $\extractor$ to (sequentially) extract the value committed by every corrupted party in Step 2. This allows the simulator to compute $\bigoplus_{i \in \mathbb{M}} r_i$, where $\mathbb{M}$ denotes the set of corrupted parties. In Step 3, the simulator broadcasts values $r'_i$ on behalf of honest parties such that $\bigoplus_{i \in [n] \setminus \mathbb{M}} r'_i  = \bigoplus_{i \in \mathbb{M}} r_i \oplus r^*$.
Finally, it invokes the simulator of the ZK protocol to produce proofs on behalf of honest parties.
%
It is easy to see that the output would indeed end up being the intended output $r^*$.

Notice that replacing Watrous' polynomial-round ZK protocol with the constant-round ZK of~\cite{BS20,EPRINT:ALP19} only decreases the rounds to linear in the number of parties.
To decrease the number of rounds to constant, it is clear that one would need to find a way to execute the commitment sessions (Step 2) and ZK sessions (Step 4) in parallel.
%
While the recent work of Bitansky and Shmueli~\cite{BS20} builds constant-round post-quantum zero-knowledge, their protocol and its guarantees turn out to be insufficient for the parallel setting. In this setting, a single prover would typically need to interact in parallel with $(n-1)$ different verifiers, a subset or all of which may be adversarial. It should be possible for a simulator to {\em simultaneously} simulate the view of multiple parallel verifiers. In addition, the argument should continue to satisfy soundness, even if a subset of verifiers colludes with a (cheating) prover.

%

\paragraph{Post-Quantum Parallel Zero-Knowledge.} 
We overcome this barrier by building the first constant-round zero-knowledge argument secure against {\em parallel quantum verifiers} from quantum polynomial hardness of an LWE-based circular security assumption. 
This improves upon the work of~\cite{BS20,EPRINT:ALP19} who provided arguments with provable security only against a single quantum verifier.
Very roughly, the approach in~\cite{BS20,EPRINT:ALP19} relies on a modification of the~\cite{STOC:BKP19} homomorphic trapdoors paradigm. 
We do not assume familiarity with the details of this protocol or paradigm, and will in fact discuss a (variant of) this in the next subsection.
For now, we simply point out that in this paradigm, the verifier generates an initial FHE ciphertext and public key, as well as some additional information to enable simulation.
The simulator {\em homomorphically evaluates} the verifier's (quantum) circuit over the initial FHE ciphertext and then uses the result of this evaluation to recover secrets that will enable simulation.

However, when a prover interacts with several verifiers at once, each verifier will generate its own FHE ciphertexts. In a nutshell, in the parallel setting the simulator 
can no longer perform individual homomorphic evaluations corresponding to each verifier, due to no-cloning.
To address this issue, we develop a novel {\bf parallel no-cloning} simulation strategy. This is our key technical contribution: we develop a novel technique 
that enables the simulator to {\em peel away} secret keys of this FHE scheme layer-by-layer. 
An overview of this technique can be found in Section \ref{sec:over-pzk}.

Our technique also relies on a strong variant of quantum fully-homomorphic encryption that allows for homomorphic operations under multiple keys at once.
The encryption scheme that we use is a quantum generalization of the notion of additive function sharing (AFS) {\em spooky encryption}~\cite{C:DHRW16}. 
As a contribution of independent interest, we build the first AFS-spooky encryption (that also implies multi-key FHE) for quantum circuits from a circular variant of the LWE assumption. We give an overview of our construction in Section \ref{sec:over-spooky}.

\paragraph{Post-Quantum Non-malleable Commitments.}
Our construction of zero-knowledge against parallel quantum verifiers gives rise to a coin-flipping protocol that is secure as long as at least one participant is honest, and all committed strings are independent of each other. However, ensuring such independence is not straightforward, even in the classical setting. 
In fact, upon seeing an honest party's commitment string $c$, a malicious, rushing adversary may be able to produce a string $c'$ that commits to a related message. 
This is known as a malleability attack, and can be prevented by relying on {\em non-malleable commitments}.
In this work, we devise the first post-quantum non-malleable commitments based on slightly superpolynomial hardness of LWE. An overview of our construction can be found in Section \ref{sec:over-nm}.


Finally, we discuss how to combine all these primitives to build our desired coin-tossing protocol, and a few additional subtleties that come up in the process, in Section \ref{sec:over-putting}.

\subsection{A New Parallel No-Cloning Non-Black-Box Simulation Technique}
\label{sec:over-pzk}

In the following we give a high-level overview of our constant-round zero-knowledge protocol secure against parallel quantum verifiers. In favor of a simpler exposition, we first describe a \emph{parallel extractable commitment} protocol. A parallel extractable commitment is a commitment where a single receiver interacts in parallel with multiple committers, each committing to its own independent message. The main challenge in this setting is to simulate the view of an adversary corrupting several of these committers, while \emph{simultaneously} recovering all committed messages. Once we build a parallel extractable commiment, obtaining a parallel zero-knowledge protocol becomes a simple exercise (that we discuss towards the end of this overview).

Throughout the following overview we only consider adversaries that are (i) \emph{non-aborting}, i.e. they never interrupt the execution of the protocol, and (ii) \emph{explainable}, i.e. their messages always lie in the support of honestly generated messages, though they can select their random coins and inputs arbitrarily. We further simplify our overview by only considering (iii) \emph{classical} adversaries, while being mindful to avoid any kind of state cloning during extraction. In the end of this overview we discuss how to remove these simplications. 

\paragraph{Cryptographic Building Blocks.} Before delving into the description of our protocol, we introduce the technical tools needed for our construction. A fully-homomorphic encryption (FHE) scheme~\cite{STOC:Gentry09} allows one to compute any function (in its circuit representation) over some encrypted message $\enc(\pk, m)$, without the need to decrypt it first. We say that an FHE is multi-key~\cite{STOC:LopTroVai12} if it supports the homomorphic evaluation of circuits even over messages encrypted under \emph{independently sampled} public keys:
$$
\{\enc(\pk_i, m_i)\}_{i\in [n]} \xrightarrow{\eval((\pk_1, \dots, \pk_n),C, \cdot)} \enc((\pk_1, \dots, \pk_n), C(m_1, \dots, m_n)).
$$
Clearly, decrypting the resulting ciphertext should require the knowledge of all of the corresponding secret keys $(\sk_1, \dots, \sk_n)$. 
Other than semantic security, we require that the scheme is compact, in the sense that the size of the evaluated ciphertext is proportional to $|C(m_1, \dots, m_n)|$ (and possibly the number of parties $n$) but does not otherwise depend on the size of $C$.

The second tool that we use is compute and compare obfuscation~\cite{FOCS:WicZir17,FOCS:GoyKopWat17}. A compute and compare program $\CC{f}{u}{z}$ program is defined by a function $f$, a lock value $u$, and an output $z$. On input a string $x$, the program returns $z$ if and only if $f(x) = u$. The obfuscator $\obf$ is guaranteed to return an obfuscated program $\obfC$ that is indistinguishable from a program that rejects any input, as long as $u$ has sufficient entropy conditioned on $f$ and $z$. Finally, we use a conditional disclosure of secret (CDS)\footnote{In the body of the paper we actually resort to a slightly stronger tool, namely a secure function evaluation protocol with statistical circuit privacy.} scheme. Recall that this is an interactive protocol parametrized by an NP relation $\mathcal{R}$ where both the sender and the recevier share a statement $x$ and in addition, the sender has a secret message $m$. At the end of the interaction, the receiver obtains $m$ if and only if it knows a valid witness $w$ such that $\mathcal{R}(x,w)=1$.

\paragraph{A Strawman Solution.} We now describe a naive extension of the~\cite{BS20,EPRINT:ALP19} approach to the parallel setting (where a receiver interacts with multiple committers), and highlight its pitfalls. 
We do not assume familiarity with~\cite{BS20,EPRINT:ALP19}.

To commit to messages $(m_1, \dots, m_n)$, the 
committers and the receiver engage in the following protocol.
\begin{itemize}
    \item Each committer samples a key pair of a multi-key FHE scheme $(\pk_i, \sk_i)$, a uniform trapdoor $\td_i$, and a uniform lock value $\lock_i$, and sends to the receiver: 
    \begin{enumerate}
        \item A commitment $\cm_i = \COM(\td_i)$.
        \item An FHE encryption $\enc(\pk_i, \td_i)$.
        \item An obfuscation $\obfC_i$ of the program $\CC{\dec(\sk_i, \cdot)}{\lock_i}{(\sk_i, m_i)}$. 
    \end{enumerate}
    \item The receiver engages each committer in a (parallel) execution of a CDS protocol where the $i$'th committer sends $\lock_i$ if the receiver correctly guesses a valid pre-image of $\cm_i$.
\end{itemize}
At a high level, the fact that the protocol hides the message $m_i$ is ensured by the following argument. Since the receiver cannot invert $\cm_i$, it cannot guess $\td_i$ and therefore the CDS protocol will return $0$. This in turn means that the lock $\lock_i$ is hidden from the receiver, and consequently that the obfuscated program is indistinguishable from a null program. This is, of course, an informal explanation, and we refer the reader to~\cite{STOC:BKP19,BS20,EPRINT:ALP19} for a formal security analysis.

We now turn to the description of the extractor. The high-level strategy is the following: Upon receiving the first message from all committers, the extractor uses the FHE encryption $\enc(\pk_i, \td_i)$ and the code of the adversary to run the CDS protocol homomorphically (on input $\td_i$) to recover an FHE encryption of $\lock_i$. Then the extractor feeds it as an input to the obfuscated program $\obfC_i$, which returns $(\sk_i, m_i)$. 

Unfortunately this approach has a major limitation: It implicitly assumes that each corrupted party is a local algorithm. In other words, we are assuming that the adversary consists of individual subroutines (one per corrupted party), which may not necessarily be the case. As an example, if the adversary were to somehow implement a strategy where corrupted machines do not respond until \emph{all} receiver messages have been delivered, then the above homomorphic evaluation would get stuck and return no output. It is also worth mentioning that what makes the problem challenging is our inability to clone the state of the adversary. If we were allowed to clone its state, then we could extract messages one by one, by running a separate thread under each FHE key.


\paragraph{Multi-Key Evaluation.} A natural solution to circumvent the above issue is to rely on multi-key FHE evaluation. Using this additional property, the extractor can turn the ciphertexts $\enc(\pk_1, \td_1),\allowbreak \dots, \allowbreak \enc(\pk_n, \td_n)$ into a single encryption
$$
\enc((\pk_1, \dots, \pk_n), (\td_1, \dots, \td_n))
$$
under the hood of all public keys $(\pk_1, \dots, \pk_n)$. Given this information, the extractor
can homomorphically evaluate all instances of the CDS protocol at once, using the code of the adversary, no matter how intricate. This procedure allows the extractor to obtain the encryption of each lock value
$\enc((\pk_1, \dots, \pk_n), \lock_i)$. In the single committer setting, we could then feed this into the corresponding obfuscated program and call it a day. 

However, in the parallel setting, even given multi-key FHE, it is unclear how to proceed. If the compute and compare program $\obfC_i$ tried to decrypt such a ciphertext, it would obtain (at best) an encryption under the remaining public keys. Glossing over the fact that the structure of single-key and multi-key ciphertexts might be incompatible, it is unlikely that 
$$\dec(\sk_i, \enc((\pk_1, \dots, \pk_{n}), \lock_i)) = \lock_i$$
which is what we would need to trigger the compute and compare program. The general problem here is that each compute and compare program cannot encode information about other secret keys, thus making it infeasible to decrypt multi-key ciphertexts. One approach to resolve this issue would be to ask all committers to jointly obfuscate a compute and compare program that encodes all secret keys at once. However, this seems to require a general-purpose MPC protocol, which is what we are trying to build in the first place.
Therefore, we outline a different approah by imagining a special kind of multi-key fully homomorphic encryption scheme.

A spooky encryption\footnote{As a historical remark, while the name is inspired by Einstein's quote ``spooky action at a distance'' referring to entangled quantum states, the concept of spooky encryption (as defined in~\cite{C:DHRW16}) is entirely classical.} scheme~\cite{C:DHRW16} is an FHE scheme that supports a special \emph{spooky evaluation} algorithm, that generates no-signaling correlations among independently encrypted messages. We will restrict attention to a sub-class of no-signaling relations called \emph{additive function sharing} (AFS) relations, and we will call the scheme AFS-spooky. More concretely, on input a circuit $C$ and $n$ independently generated ciphertexts (under independently generated public keys), the algorithm $\spooky.\eval$ produces
$$
\{\enc(\pk_i, m_i)\}_{i\in [n]} \xrightarrow{\spooky.\eval((\pk_1, \dots, \pk_n),C, \cdot)}
\{\enc(\pk_i, y_i)\}_{i\in [n]} \textit{ s.t. }
\mathop{\bigoplus}\limits_{i=1}^n y_i = C(m_1,\dots,m_n).
$$
It is not hard to see that AFS-spooky encryption is a special case of multi-key FHE where multi-key ciphertexts have the following structure
$$
\enc((\pk_1, \dots, \pk_n), m) = \{\enc(\pk_i, y_i)\}_{i\in [n]} \textit{ s.t. }
\mathop{\bigoplus}\limits_{i=1}^n y_i = m.
$$
This additional structure is going to be our main leverage for constructing an efficient extractor. 

\paragraph{The Extractor.}
Going back to our extractor, our next technical insight is to look for a mechanism to \emph{peel away} encryption layers one by one from an AFS-spooky (multi-key) ciphertext. 
Our extractor will achieve this via careful \emph{homomorphic} evaluation of the independently generated programs $(\obfC_1, \dots, \obfC_n)$, as described below.
\begin{itemize}
    \item First, homomorphically execute the code of the adversary using the AFS-spooky scheme to obtain
    $$
    \ct_1 = \enc((\pk_1,\dots, \pk_n), \lock_1), \dots, \ct_n =  \enc((\pk_1,\dots, \pk_n), \lock_n),
    $$
as described above.
\item Parse $\ct_n$ as a collection of individual ciphertexts
$$
\enc((\pk_1,\dots, \pk_n), \lock_n) =  \{\enc(\pk_i, y_i)\}_{i\in [n]} = \{\enc(\pk_i, y_i)\}_{i\in [n-1]} ~\cup~ \underbrace{\{\enc(\pk_n, y_n)\}}_{\tilde{\ct}_n}.
$$
Note that we can interpret the first $n-1$ elements as an AFS-spooky ciphertext encrypted under $(\pk_1, \dots, \pk_{n-1}):$
$$
\tilde{\ct} = \{\enc(\pk_i, y_i)\}_{i\in [n-1]} = \enc\left(\left(\pk_1, \dots, \pk_{n-1}\right), \mathop{\bigoplus}\limits_{i=1}^{n-1} y_i\right) = \enc\left(\left(\pk_1, \dots, \pk_{n-1}\right), \tilde{y}\right)
$$
where $\tilde{y} = \mathop{\bigoplus}\limits_{i=1}^{n-1} y_i$.
\item Let $\Gamma$ be the following function
$$
\Gamma(\zeta): \spooky.\eval(\pk_n, \zeta \oplus \cdot, \tilde{\ct}_n)
$$
which homomorphically computes the XOR of $\zeta$ with the plaintext of $\tilde{\ct}_n$. 
Compute the following nested AFS-spooky correlation 
\begin{align}
\widehat{\ct} &= \spooky.\eval((\pk_1, \dots, \pk_{n-1}), \Gamma, \tilde{\ct}) \nonumber \\&=
\enc\left(\left(\pk_1, \dots, \pk_{n-1}\right), \spooky.\eval(\pk_n,\tilde{y} \oplus \cdot, \tilde{\ct}_n)\right)\\ &= \enc\left(\left(\pk_1, \dots, \pk_{n-1}\right),\enc\left(\pk_n,  \mathop{\bigoplus}\limits_{i=1}^{n} y_i\right)\right)\\ &= \enc\left(\left(\pk_1, \dots, \pk_{n-1}\right),\enc\left(\pk_n,  \lock_n \right)\right)
\end{align}
by interpreting $\tilde{\ct}_n$ as a single key ciphertext. 
Here (1) follows by substituting $\Gamma$, and (2) follows by correctness of the AFS-spooky evaluation.


\item Run the obfuscated compute and compare program homomorphically to obtain an encryption of $\sk_n$ and $m_n$ under $(\pk_1, \dots, \pk_{n-1})$
\begin{align*}
    \spooky.\eval\left((\pk_1, \dots, \pk_{n-1}), \obfC_n, \widehat{\ct}\right)
    &=\enc\left(\left(\pk_1, \dots, \pk_{n-1}\right),\obfC_n\left(\enc\left(\pk_n,  \lock_n \right)\right) \right)\\ &=
\enc\left(\left(\pk_1, \dots, \pk_{n-1}\right), (\sk_n, m_n) \right).
\end{align*}

\item Using the encryption of $\sk_n$ under $(\pk_1, \dots, \pk_{n-1})$, update the initial ciphertexts $(\ct_1, \dots, \ct_{n-1})$ by homomorphically decrypting their last component and adding the resulting string. This allows the extractor to obtain
$$
\enc((\pk_1, \dots, \pk_{n-1}), \lock_1), \dots, \enc((\pk_1, \dots, \pk_{n-1}), \lock_{n-1}).
$$
\item Recursively apply the procedure described above until $\enc(\pk_1, \lock_1)$ is recovered, then feed this ciphertext as an input to $\obfC_1$ to obtain $(\sk_1, m_1)$ in the clear. Iteratively recover $(\sk_2, \dots, \sk_n)$ by decrypting the corresponding ciphertexts. At this point the extractor knows all secret keys and can decrypt the transcript of the interaction together with the committed messages.
\end{itemize}
To summarize, this extractor will isolate single-key ciphertexts (albeit in a nested form) by relying on AFS-spooky encryption. These ciphertexts by design will be compatible with compute and compare programs. In turn, evaluating the program \emph{under the encryption} allows us to \emph{escape} from the newly introduced layer. Repeating this procedure recursively eventually leads to a complete recovery of the plaintexts.

We stress that, although the extraction algorithm repeats the nesting operation $n$ times, the additional encryption layer introduced in each iteration is immediately peeled off by executing the obfuscated compute and compare program. Thus the above procedure runs in (strict) polynomial time for \emph{any polynomial} number of parties $n$.

\paragraph{Parallel Zero Knowledge.} The above outline is deliberately simplified and ignores some subtle issues that arise during the analysis of the protocol. As an example, we need to ensure that the adversary is not able to \emph{maul} the commitment of the trapdoor into a CDS encryption to be used in the CDS protocol. This issue also arose in~\cite{BS20}, and we follow their approach of using non-uniformity in a reduction to the semantic security of the quantum FHE scheme.~\cite{BS20} also present the technical tools needed to lift the protocol to the setting of malicious and possibly aborting adversaries (as opposed to explainable), and we roughly follow their approach. However, it is worth pointing out that~\cite{BS20} directly construct a zero-knowledge argument, without first constructing and analyzing a stand-alone extractable commitment. Since we use a parallel extractable commitment as a building block in the our coin-flipping protocol, we analyze the above as a stand-alone commitment, which requires a few modifications to the protocol and proof techniques. More discussion about this can be found in~\cref{sec:pecom}.

Now, we describe how to obtain parallel zero-knowledge (i.e.~zero-knowledge against multiple verifiers) from parallel extractable commitments. This is accomplished in a routine manner by enhancing a standard $\Sigma$ protocol with a stage where each verifier commits to its $\Sigma$ protocol challenge using a parallel extractable commitment. Using the extractor, the simulator can obtain the challenges ahead of time and can therefore simulate the rest of the transcript, without the need to perform state cloning. 

It remains to argue that our extraction strategy does not break down in the presence of quantum adversaries. Observe that the only step that involves the execution of a quantum circuit is the AFS-spooky evaluation of the CDS protocol, under the hood of $(\pk_1, \dots, \pk_n)$. Assuming that we can construct AFS-spooky encryption for relations computable by quantum circuits (which we show in Section~\ref{sec:over-spooky}), the remainder of the extraction algorithm only depends on the encryptions of $(\lock_1, \dots, \lock_n)$, which are classical strings. Once the extractor recovers all the secret keys, it can decrypt the (possibly quantum) state of the adversary resulting from the homomorphic evaluation of the CDS, and resume the protocol execution, without the need to clone the adversary's state.




\subsection{Quantum AFS-Spooky Encryption}
\label{sec:over-spooky}



We now turn to the construction of AFS-spooky encryption for relations computable by quantum circuits. The main technical contribution of this section is a construction of multi-key fully-homomorphic encryption for quantum circuits with classical key generation and classical encryption of classical messages. Such schemes were already known in the \emph{single}-key setting, due to~\cite{FOCS:Mahadev18b,C:Brakerski18}. 

\paragraph{Background.} At a very high level, these single-key schemes follow a paradigm introduced by Broadbent and Jeffery~\cite{C:BroJef15}, which makes use of the quantum one-time pad (QOTP). The QOTP is a method of perfectly encrypting arbitrary quantum states with a key that consists of only classical bits.~\cite{C:BroJef15} suggest to encrypt a quantum state with a quantum one-time pad (QOTP), and then encrypt the classical bits that comprise the QOTP using a classical fully-homomorphic encryption scheme. One can then apply quantum gates to the encrypted quantum state, and update the classical encryption of the one-time pad appropriately. A key feature of this encryption procedure is that while an encryption of a quantum state necessarily must be a quantum state, an encryption of classical information does not necessarily have to include a quantum state. Indeed, one can simply give a classical one-time pad encryption of the data, along with a classical fully-homomorphic encryption of the pad.

However, the original schemes presented by Broadbent and Jeffery~\cite{C:BroJef15} and subsequent work~\cite{C:DulSchSpe16} based on their paradigm left much to be desired. In particular, they required even a classical encryptor to supply quantum ``gadgets'' encoding their secret key. These gadgets were then used to evaluate a particular non-Clifford gate over encrypted data.\footnote{We also remark here that~\cite{EPRINT:Goyal18} presented a \emph{multi}-key scheme based on this paradigm, but with the same drawbacks. Note that compactness and classical encryption are crucial in our setting, as per the discussion in the previous section.} The main innovation in the work of~\cite{FOCS:Mahadev18b} was to remove the need for quantum gadgets, instead showing how to evaluate an appropriate non-Clifford gate using just \emph{classical} information supplied by the encryptor.


\paragraph{Encrypted CNOT Operation.} In more detail, evaluating a non-Clifford gate on a ciphertext $(\ct,\ket{\phi})$, where $\ct$ is an FHE encryption of a QOTP key and $\ket{\phi}$ is a quantum state encrypted under the QOTP key, involves an operation (referred to as encrypted CNOT) that somehow must ``teleport'' the bits encrypted in $\ct$ into the state $\ket{\phi}$.~\cite{FOCS:Mahadev18b} gave a method for doing this, as long as the ciphertext $\ct$ is encrypted under a scheme with some particular properties. Roughly, the scheme must support a ``natural'' XOR homomorphic operation, it must be circuit private with respect to this homomorphism, and perhaps most stringently, there must exist some trapdoor that can be used to recover the message and the \emph{randomness} used to produce any ciphertext.

~\cite{FOCS:Mahadev18b} observed that the dual-Regev encryption scheme~\cite{STOC:GenPeiVai08} (with large enough modulus-to-noise ratio) does in fact satisfy these properties, as long as one generates the public key matrix $\bA$  along with a trapdoor. However, recall that
$\ct$ was supposed to be encrypted under a fully-homomorphic encryption scheme.~\cite{FOCS:Mahadev18b} resolves this by observing that ciphertexts encrypted under the dual variant of the~\cite{C:GenSahWat13} fully-homomorphic encryption scheme actually already contain a dual-Regev ciphertext. In particular, a dual-GSW ciphertext encrypting a bit $\mu$ is a matrix $\bM = \bA\bS + \bE + \mu \bG$, where $\bG$ is the gadget matrix. The final column of $\bM$ is $\bA\bs + \be + \mu [0,\dots,0,q/2]^\top$, which is exactly a dual-Regev ciphertext encrypting $\mu$ under public key $\bA$. Note that, crucially, if the dual-GSW public key $\bA$ is drawn with a trapdoor, then this trapdoor also functions as a trapdoor for the dual-Regev ciphertext. Thus, an evaulator can indeed perform the encrypted CNOT operation on any ciphertext $(\ct,\ket{\phi})$, by first extracting a dual-Regev ciphertext $\ct'$ from $\ct$ and then proceeding.

\paragraph{Challenges in the Multi-Key Setting.} Now, it is natural to ask whether this approach readily extends to the multi-key setting. Namely, does there exist a multi-key FHE scheme where any (multi-key) ciphertext contains within it a dual-Regev ciphertext \emph{with a corresponding trapdoor}? Unfortunately, this appears to be much less straightforward than in the single-key setting, for the following reason. Observe that (dual) GSW homomorphic operations over ciphertexts $\bM_i = \bA\bS_i + \bE_i + \mu_i \bG$ always maintain the same $\bA$ matrix, while updating $\bS_i$, $\bE_i$, and $\mu_i$. Thus, a trapdoor for $\bA$ naturally functions as a trapdoor for the dual-Regev ciphertext that consitutes the last column of $\bM_i$. However, LWE-based multi-key FHE schemes from the literature~\cite{C:CleMcG15,EC:MukWic16,TCC:PeiShi16,TCC:BraHalPol17} include a \emph{ciphertext expansion} procedure, which allows an evaluator, given public keys $\pk_1,\dots,\pk_n$, and a ciphertext $\ct$ encrypted under some $\pk_i$, to convert $\ct$ into a ciphertext $\hat{\ct}$ encrypted under all keys $\pk_1,\dots,\pk_n$. Now, even if these public keys are indeed matrices $\bA_1,\dots,\bA_n$ drawn with trapdoors $\tau_1,\dots,\tau_n$, it is unclear how to combine $\tau_1,\dots,\tau_n$ to produce a trapdoor $\hat{\tau}$ for the ``expanded'' ciphertext. Indeed, the expanded ciphertext generally can no longer be written as some $\bA\bS + \bE + \mu \bG$, since the expansion procedure constructs a highly structured matrix that includes components from the ciphertexts $\ct_1,\dots,\ct_n$, as well as auxiliary encryptions of the randomness used to produce the ciphertexts (see e.g.~\cite{EC:MukWic16}). 

\paragraph{A Solution Based on Key-Switching.} Thus, we take a different approach. Rather than attempting to tweak known ciphertext expansion procedures to also support ``trapdoor expansion'', we rely on the notion of key-switching, which is a method of taking a ciphertext encrypted under one scheme and converting it into a ciphertext encrypted under another scheme. The observation, roughly, is that we do not need to explicitly maintain a trapdoor for the multi-key FHE scheme, as long as it is possible to convert a multi-key FHE ciphertext into a dual-Regev ciphertext that \emph{does} explicitly have a trapdoor. In fact, we will consider a natural multi-key generalization of dual-Regev, as described below. Key switching is possible as long as the second scheme has sufficient homomorphic properties, namely, it can support homomorphic evaluation of the \emph{decryption circuit} of the first scheme.


Fortunately, the dual-Regev scheme is already linearly homomorphic, and many known classical multi-key FHE schemes~\cite{C:CleMcG15,EC:MukWic16,TCC:PeiShi16,TCC:BraHalPol17} support \emph{nearly linear decryption}, which means that decrypting a ciphertext simply consists of applying a linear function (derived from the secret key) and then rounding. Thus, as long as the evaluator has the secret key of the multi-key FHE ciphertext encrypted under a dual-Regev public key with a trapdoor, they can first key-switch the multi-key FHE ciphertext $\ct$ into a dual-Regev ciphertext $\ct'$, and then proceed with the encrypted CNOT operation. 

It remains to show how an evaluator may have access to such a dual-Regev encryption. Since we are still in the multi-key setting, we will need a ciphertext and corresponding trapdoor expansion procedure for dual-Regev. However, we show that such a procedure is much easier to come by when the scheme only needs to support \emph{linear} homomorphism (as is the case for the dual-Regev scheme) rather than \emph{full} homomorphism. Each party can draw its own dual-Regev public key $\bA_i$ along with a trapdoor $\tau_i$, and encrypt its multi-key FHE secret key under $\bA_i$ to produce a ciphertext $\ct_i$. The evaluator can then treat the block-diagonal matrix $\hat{\bA} = \mathsf{diag}(\bA_1,\dots,\bA_n)$ as an ``expanded'' public key.\footnote{Actually this expansion should be done slightly more carefully, see~\cref{subsec:keyswitch} for details.} Now, the message and randomness used to generate a ciphertext encrypted under $\hat{\bA}$ may be recovered by applying $\tau_1$ to the first set of entries of the ciphertext, applying $\tau_2$ to the second set of entries and so on. This observation, combined with an appropriate expansion procedure for the ciphertexts $\ct_i$, allows an evaluator to convert any multi-key FHE ciphertext into a multi-key dual-Regev ciphertext with trapdoor. Given a classical multi-key FHE scheme with nearly linear decryption, this suffices to build multi-key quantum FHE with classical key generation and encryption.

\paragraph{Distributed Setup.} We showed above how to convert any classical multi-key FHE scheme into a quantum multi-key FHE scheme, as long as the classical scheme has nearly linear decryption. However, most LWE-based classical multi-key FHE schemes operate in the common random string (CRS) model, which assumes that all parties have access to a common source of randomness, generated by a trusted party. Thinking back to our application to parallel extractable commitments, it is clear that this will not suffice, since we have no CRS a priori, and a receiver that generates a CRS maliciously may be able to break hiding of the scheme. Thus, we rely on the multi-key FHE scheme of~\cite{TCC:BraHalPol17}, where instead of assuming a CRS, the parties participate in a distributed setup procedure. In particular, each party (and in our application, each committer) generates some public parameters $\pp_i$, which are then combined publicly to produce a single set of public parameters $\pp$, which can be used by anyone to generate their own public key / secret key pair.

This form of distributed setup indeed suffices to prove the hiding of our parallel commitment, so it remains to show that our approach, combined with~\cite{TCC:BraHalPol17}, yields a quantum multi-key FHE scheme with distributed setup. First, the~\cite{TCC:BraHalPol17} scheme does indeed enjoy nearly linear decryption, so plugging it into our compiler described above gives a functional quantum multi-key FHE scheme. Next, we need to confirm that our compiler does not destroy the distributed setup property. This follows since each party draws its own dual-Regev public key with trapdoor without relying on any CRS, or even any public parameters.

\paragraph{Quantum AFS-Spooky Encryption.} Finally, we show, via another application of key-switching, how to construct a quantum AFS-spooky encryption scheme (with distributed setup). Recall that we only require ``spooky'' interactions to hold over classical ciphertexts. That is, for any \emph{quantum} circuit $C$ with classical outputs, given ciphertexts $\ct_1,\dots,\ct_n$ encrypting $\ket{\phi_1},\dots,\ket{\phi_n}$ respectively under public keys $\pk_1,\dots,\pk_n$, an evaluator can produce ciphertexts $\ct_1',\dots,\ct_n'$ where $\ct_i'$ encrypts $y_i$ under $\pk_i$, and such that $\mathop{\bigoplus}\limits_{i=1}^n y_i = C(\ket{\phi_1},\dots,\ket{\phi_n})$.

Now, using our quantum multi-key FHE scheme, it is possible to compute a single (multi-key) ciphertext $\hat{\ct}$ that encrypts $C(\ket{\phi_1},\dots,\ket{\phi_n})$ under all public keys $\pk_1,\dots,\pk_n$. Then, if each party additionally drew a key pair $(\pk_i', \sk_i')$ for a classical AFS-spooky encryption scheme, and released $\tilde{\ct}_1,\dots,\tilde{\ct}_n$, where $\tilde{\ct}_i = \enc(\pk_i', \sk_i)$ encrypts the $i$-th party's quantum multi-key FHE secret key under their AFS-spooky encryption public key, then the evaluator can homomorphically evaluate the quantum multi-key FHE decryption circuit (which is classical for classical ciphertexts) with $\hat{\ct}$ hardcoded, where $\hat{\ct}$ is the multi-key ciphertext defined at the beginning of this paragraph.
This circuit on input $\tilde{\ct}_1,\dots,\tilde{\ct}_n$  produces the desired output $\ct_1',\dots,\ct_n'$. Finally, note that the classical AFS-spooky encryption scheme must also have distributed setup, and we show (see~\cref{subsec:spooky}) that one can derive a distributed-setup AFS-spooky encryption scheme from~\cite{TCC:BraHalPol17} using standard techniques~\cite{C:DHRW16}.

\subsection{Post-Quantum Non-malleable Commitments}
\label{sec:over-nm}
In this section, we describe how to obtain constant-round post-quantum non-malleable commitments under the assumption that there exists a natural number $c>0$ such that quantum polynomial-time adversaries cannot distinguish LWE samples from uniform with advantage better than ${\secp^{-\mathsf{ilog}(c,\secp)}}$, where $\mathsf{ilog}(c,\secp) = {\log \log \cdots_{{c}~\mathrm{times}} \log}(\secp)$ and $\secp$ denotes the security parameter.

We will focus on perfectly binding and computationally hiding constant-round interactive commitments.
Loosely speaking, a commitment scheme is said to be non-malleable if no adversary (also called a man-in-the-middle), when participating as a receiver in an execution of an honest commitment $\COM(m)$, can at the same time generate a commitment $\COM(m')$, such that the message $m'$ is related to the original message $m$. This is equivalent (assuming the existence of one-way functions with security against quantum adversaries)
to a tag-based notion where the commit algorithm obtains as an additional input a tag in $\{0,1\}^\secp$, and the adversary is restricted to using a tag, or identity, that is different from the tag used to generate its input commitment. 
We will rely on tag-based definitions throughout this paper. 
We will also only focus on the {\em sychronous setting}, where the commitments proceed in rounds, and the man-in-the-middle sends its own message for a specific round before obtaining an honest party's message for the next round.

Before describing our ideas, we briefly discuss existing work on {\em classically-secure} non-malleable commitments. 
Unfortunately, existing constructions of constant-round non-malleable commitments against classical adversaries from standard polynomial hardness assumptions~\cite{Bar02,PR05b,PR08a,LPV08,PPV08,LP09,FOCS:Wee10,EC:PasWee10,STOC:LinPas11,STOC:Goyal11,GLOV12,GRRV14,GPR15,COSV16a,COSV16b,Khu17,GR19}
either rely on rewinding, or use Barak's non-black-box simulation technique, both of which require the reduction to perform state cloning. As such, known techniques fail to prove quantum security of these constructions.

We now discuss our techniques for constructing post-quantum non-malleable commitments. Just like several classical approaches, we will proceed in two steps.
\begin{itemize}
\item We will obtain simple “base” commitment schemes for very small tag/identity spaces from slightly superpolynomial hardness assumptions.
\item Then assuming polynomial hardness of LWE against quantum adversaries, and making use of constant-round post-quantum zero-knowledge arguments, we will convert non-malleable commitments for a small tag space into commitments for a larger tag space, while only incurring a constant round overhead.
\end{itemize}


For the base schemes, there are known classical constructions~\cite{EC:PasWee10} that assume hardness of LWE against $2^{\secp^\delta}$-size adversaries, where $\secp$ denotes the security parameter and $0<\delta<1$ is a constant. We observe that these constructions can be proven secure in the quantum setting, resulting in schemes that are suitable for tag spaces of $O(\log \log \secp)$ tags. 

\paragraph{Tag Amplification.}
Since an MPC protocol could be executed among up to $\poly(\secp)$ parties where $\poly(\cdot)$ is an arbitrary polynomial, 
we end up requiring non-malleable commitments suitable for tag spaces of $\poly(\secp)$.
This is obtained by combining classical tools for amplifying tag spaces~\cite{DDN91} with constant round post-quantum zero-knowledge protocols. Our tag amplification protocol, on input a scheme with tag space $2t$, outputs a scheme with tag space $2^{t}$, for any $t \leq \poly(\secp)$. This follows mostly along the lines of existing classical protocols, and as such we do not discuss the protocol in detail here. Our protocol can be found in Section~\ref{sec:tag-amp}.

\paragraph{Base Schemes from $\secp^{-\mathsf{ilog}(c,\secp)}$ Hardness.}
Returning to the question of constructing appropriate base schemes, we also improve the assumption from $2^{\secp^\delta}$-quantum hardness of LWE (that follows based on~\cite{EC:PasWee10}) to the mildly superpolynomial hardness assumption discussed at the beginning of this subsection.
Recall that we will only need to assume that there exists an (explicit) natural number $c>0$ such that quantum polynomial time adversaries cannot distinguish LWE samples from uniform with advantage better than $\negl(\secp^{\mathsf{ilog}(c,\secp)})$
where $\mathsf{ilog}(c,\secp) = {\log \log \cdots_{{c}~\mathrm{times}} \log}(\secp)$.
Our base scheme will only be suitable for identities in $\mathsf{ilog}(c+1, \secp)$, where $c >0$ is a natural number, independent of $\secp$. We will then repeatedly apply the tag amplification process referred to above to boost the tag space to $2^\secp$, by adding only a constant number of rounds.

To build our base scheme, we take inspiration from the classically secure non-malleable commitments of Khurana and Sahai~\cite{KS17}. However, beyond considering quantum as opposed to classical adversaries, our protocol and analysis will have the following notable differences from~\cite{KS17}:
\begin{itemize}
    \item The work of~\cite{KS17} relies on sub-exponential hardness (i.e. $2^{\secp^\delta}$ security), which is stronger than the type of superpolynomial hardness we assume. This is primarily because~\cite{KS17} were restricted to two rounds, but we can improve parameters while allowing for a larger constant number of rounds.
    \item \cite{KS17} build a reduction that rewinds an adversary to the beginning of the protocol, and executes the adversary several times, repeatedly sampling the adversary's initial state. This may be undesirable in the quantum setting.\footnote{In particular this state may not always be efficiently sampleable, in which case it would be difficult to build an efficient reduction.} On the other hand, we have a simpler fully straight-line reduction that only needs to run the adversary once. 
\end{itemize}


Specifically, following~\cite{KS17}, we will establish {\em an erasure channel} between the committer and receiver that transmits the committed message to the receiver with probability $\epsilon$. To ensure that the commitment satisifies hiding, $\epsilon$ is chosen to be a value that is negligible in $\secp$.
At the same time, the exact value of $\epsilon$ is determined by the identity ($\tagg$) of the committer. 
Recall that $\tagg \in [1,\mathsf{ilog}(c+1,\secp)]$.
We will set $\epsilon = \eta^{-\tagg}$ where $\eta =  \secp^{\mathsf{ilog}(c+1,\secp)}$ is a superpolynomial function of $\secp$.

Next, for simplicity, we restrict ourselves to a case where the adversary's tag (which we denote by $\tagg'$) is smaller than that of the honest party (which we denote by $\tagg$).
In this case, the adversary's committed message is transmitted with probability $\epsilon' = \eta^{-\tagg'}$, whereas the honest committer's message is transmitted with probability only $\epsilon = \eta^{-\tagg}$, which is smaller than $\epsilon'$.

We set this up so that the transcript of an execution transmits the adversary's message with probability $\epsilon'$ (over the randomness of the honest receiver), and on the other hand, an honestly committed message will remain hidden except with probability $\epsilon < \epsilon'$ (over the randomness of the honest committer).
This gap in the probability of extraction will help us argue non-malleability, using a proof strategy that bears resemblance to the proof technique in~\cite{BitanskyL18} (who relied on stronger assumptions to achieve such a gap in the non-interactive setting).

We point out one subtlety in our proof that does not appear in~\cite{BitanskyL18}. We must rule out a man-in-the-middle adversary that on the one hand, does not commit to a related message if its message was successfully transmitted, but on the other hand, can succesfully perform a mauling attack if its message was not transmitted. To rule out such an adversary, just like~\cite{KS17}, we will design our erasure channel so that the adversary cannot distinguish transcripts where his committed message was transmitted from those where it wasn't. 

Finally, our erasure channel can be cryptographically established in a manner similar to prior work~\cite{KS17,KKS18,DBLP:conf/eurocrypt/BadrinarayananF20} via an indistinguishability-based variant of two-party secure function evaluation, that can be based on quantum hardness of LWE.
Specifically, we would like to ensure that the SFE error is (significantly) smaller than the transmission probabilities of our erasure channels: therefore, we will set parameters so that SFE error is $\secp^{-\mathsf{ilog}(c,\secp)}$.
We refer the reader to Section~\ref{sec:nmc} for additional details about our construction.

\paragraph{On Super-Constant Rounds from Polynomial Hardness.} We also observe that for any $t(\secp) \leq \poly(\secp)$, non-malleable commitments for tag space of size $t(\secp)$ can be obtained in $O(t(\secp))$ rounds based on any extractable commitment using ideas from~\cite{DDN91,ChoRab}, where only one party speaks in every round. 
These admit a straight-line reduction, and can be observed to be quantum-secure.
As such, based on quantum polynomial hardness of LWE and quantum FHE, we can obtain a base protocol for $O(\log \log \ldots_{\text{c times}} \log \secp)$ tags requiring $O(\log \log \ldots_{\text{c times}} \log \secp)$ rounds, for any constant $c \in \mathbb{N}$.
Applying our tag-amplification compiler to this base protocol makes it possible to increase the tag space to $2^\secp$ while only adding a constant number of rounds. Therefore, this technique gives
$O(\log \log \ldots_{\text{c times}} \log \secp)$ round non-malleable commitments for exponentially large tags from quantum polynomial hardness.
It also yields constant round non-malleable commitments for a constant number of tags from polynomial hardness.

\subsection{Putting Things Together}
\label{sec:over-putting}

Finally, we show how to combine the primitives described above to obtain a constant-round coin-flipping protocol that supports straight-line simulation. 
As we saw above, in the setting of multi-verifier zero-knowledge, simultanesouly simulating the view of multiple parties without rewinding can be quite challenging, so a careful protocol and proof is needed. 

Recall the outline presented at the beginning of this section, where each party first commits to a uniformly random string, then broadcasts the committed message, and finally proves in ZK that the message broadcasted is equal to the previously committed message. If all proofs verify, then the common output is the XOR of all broadcasted strings. Recall also that the coin-tossing protocol should be \emph{fully-simulatable}. This means that a simulator should be able to force the common output to be a particular uniformly drawn string given to it as input.

It turns out that in order to somehow force a particular output, the simulator should be able to {\em simultaneously extract in advance} all the messages that adversarial parties committed to. 
In particular, we require commitments where a simulator can extract from multiple committers committing in parallel. 
Here, we will rely on  
our parallel extractable commitment described above. 
Note that we will also need to simulate the subsequent zero-knowledge arguments given by the malicious parties in parallel, and thus we instantiate these with our parallel zero-knowledge argument described above. 
However, an issue remains. What if an adversary could somehow \emph{maul} an honest party's commitment to a related message and then broadcast that commitment as their own? This could bias the final outcome away from uniformly random. 

Thus, we need to introduce some form of non-malleability into the protocol. 
Indeed, we will add another step at the beginning where each party commits to its message $c_i$ and some randomness $r_i$ using our post-quantum many-to-one non-malleable commitment.\footnote{Above we described a construction of one-to-one non-malleable commitment, though a hybrid argument~\cite{LPV08} shows that one-to-one implies many-to-one.} 
Each party will then commit to $c_i$ again with our extractable commitment, using randomness $r_i$. 
Finally, each party proves in zero-knowledge that the previous commitments were consistent.

This protocol can be proven to be fully simulatable.
Intuitively, 
even though the simulator changes the behavior of honest players in order to extract from the adversary's commitments and then later force the appropriate output, the initial non-malleable commitments given by the adversary must not change in a meaningful way, due the the guarantee of non-malleablity. However, additional subtleties arise in the proof of security. In particular, during the hybrids the simulator will first have to simulate the honest party zero-knowledge arguments, before changing the honest party commitments in earlier stages. However, when changing an honest party's commitment, we need to rely on non-malleability to ensure that the malicious party commitments will not also change in a non-trivial way. 
Here, 
we use a proof technique that essentially invokes soundess of the adversary's zero-knowledge arguments at an earlier hybrid but allows us to nevertheless rely on non-malleable commitments to enforce that the adversary behaves consistently in all future hybrids. 
More discussion and a formal analysis can be found in~\cref{sec:coin-tossing}.




\subsection{Related Work}
\label{sec:related}
Classical secure multi-party computation was introduced and shown to be achievable in the two-party setting by~\cite{FOCS:Yao82b} and in the multi-party setting by~\cite{STOC:GolMicWig87}. Since these seminal works, there has been considerable interest in reducing the round complexity of classical protocols. In the setting of malicious security against a disjonest majority,~\cite{JC:Lindell03} gave the first \emph{constant}-round protocol for two-party computation, and~\cite{EC:KatOstSmi03} gave the first constant-round protocol for multi-party computation. Since then, there has been a long line of work improving on the exact round complexity and assumptions necessary for classical multi-party computation (see e.g.~\cite{STOC:Pass04,EC:GMPP16}).

\paragraph{Post-quantum classical protocols}. The above works generally focus on security against \emph{classical} polynomial-time adversaries. Another line of work, most relevant to the present work, has considered the more general goal of proving the security of classical protocols against arbitrary \emph{quantum} polynomial-time adversaries. 

This study was initiated by van de Graaf~\cite{10.5555/928621}, who observed that the useful rewinding technique often used to prove zero-knowledge in the classical setting may be problematic in the quantum setting. In a breakthrough work, Watrous~\cite{10.1137/060670997} showed that several well-known classical zero-knowledge protocols are in fact zero-knowledge against quantum verifiers, via a careful rewinding argument. However, these protocols require a polynomial number of rounds to achieve negligible security against quantum attackers. Later, Unruh~\cite{EC:Unruh12} developed a more powerful rewinding technique that suffices to construct classical zero-knowledge \emph{proofs of knowledge} secure against quantum adversaries, though still in a polynomial number of rounds. In a recent work,~\cite{BS20} managed to construct a constant-round post-quantum zero-knowledge protocol, under assumptions similar to those required to obtain classical fully-homomorphic encryption. In another recent work,~\cite{EPRINT:ALP19} constructed a constant-round protocol that is zero-knowledge against quantum verifiers under the quantum LWE assumption, though soundness holds against only classical provers.

There has also been some work on the more general question of post-quantum secure computation. In particular,~\cite{AC:DamLun09} used the techniques developed in~\cite{10.1137/060670997} to build a two-party coin-flipping protocol, and~\cite{AFRICACRYPT:LunNie11,C:HalSmiSon11} constructed general two-party computation secure against quantum adversaries, in a polynomial number of rounds. More recently,~\cite{BS20} gave a \emph{constant}-round two-party coin-flipping protocol, with full simulation of one party. However, prior to this work, nothing was known in the most general setting of post-quantum multi-party computation (in the plain model).

Finally, as mentioned at the beginning of \cref{subsec:tech-background}, there exist post-quantum classical protocols in the literature, as long as some form of trusted setup is available. 


\paragraph{Quantum protocols}. Yet another line of work focuses on protocols for securely computing \emph{quantum} circuits. General multi-party quantum computation was shown to be achievable in the information-theoretic setting (with honest majority) in the works of~\cite{STOC:CreGotSmi02,FOCS:BCGHS06}. In the computational setting,~\cite{C:DupNieSal10} gave a two-party protocol secure against a quantum analogue of semi-honest adversaries, and~\cite{C:DupNieSal12} extended security of two-party quantum computation to the malicious setting. In a recent work~\cite{DBLP:conf/eurocrypt/DulekGJMS20} constructed a maliciously secure multi-party protocol for computing quantum ciruits, assuming the existence of a maliciously secure post-quantum classical MPC protocol. We remark that all of the above protocols operate in a polynomial number of rounds.
\section{Preliminaries}\label{sec:prel}

Various parts of this section are taken nearly verbatim from~\cite{BS20}. All algorithms of cryptographic functionalities in this work are implicitly efficient and classical (i.e.
require no quantum computation or a quantum communication channel), unless noted otherwise. We
rely on the standard notions of classical Turing machines and Boolean circuits:
\begin{itemize}
\item We say that a Turing machine (or algorithm) is PPT if it is probabilistic and runs in polynomial
time.
\item We sometimes think about PPT Turing machines as polynomial-size uniform families of circuits
(as these are equivalent models). A polynomial-size circuit family $C$ is a sequence of circuits
$C = \{C_\secp\}_{\secp \in \mathbb{N}}$, such that each circuit $C_\secp$ is of polynomial size $\secp^{O(1)}$ and has $\secp^{O(1)}$ input and output bits. We say that the family is uniform if there exists a polynomial-time deterministic Turing machine $M$ that on input $1^\secp$
outputs $C_\secp$.
\item For a PPT Turing machine (algorithm) $M$, we denote by $M(x; r)$ the output of $M$ on input $x$ and
random coins $r$. For such an algorithm, and any input $x$, we may write $m \in M(x)$ to denote the
fact that $m$ is in the support of $M(x;\cdot)$. 
\end{itemize}

\paragraph{Miscellaneous notation}.

\begin{itemize}
    \item For a distribution $\mathcal{D}$ that may explicitly take its random coins $r$ as input, we denote by $x \gets \mathcal{D}$ the process of sampling from $\mathcal{D}$, and denote by $x \coloneqq \mathcal{D}(r)$ the fixed outcome $x$ that results from sampling from $\mathcal{D}$ with random coins $r$.
    \item We denote by $U_\secp$ the uniform distribution over $\{0,1\}^\secp$. 
    \item Given an NP language $\cL$ with associated relation $\mathcal{R}_\cL$, and an instance $x$, we let $\mathcal{R}_\cL(x)$ denote the set $\{w : \mathcal{R}_\cL(x,w) = 1\}$. 
    \item For some natural number $c$ and security parameter $\secp$, we use $\mathsf{ilog}(c,\secp)$ to denote $\underbrace{\log \log \cdots \log}_{{c}~\mathrm{times}}(\secp)$. 
    \item We will use $\Delta(\mathcal{X}, \mathcal{Y})$ to denote the statistical distance between two distributions $\mathcal{X}$ and $\mathcal{Y}$.
\end{itemize}

\subsection{Quantum Computation}
We use standard notions from quantum computation.
\begin{itemize}
\item We say that a Turing machine (or algorithm) is QPT if it is quantum and runs in polynomial time.
\item We sometimes think about QPT Turing machines as polynomial-size uniform families of quantum
circuits (as these are equivalent models). We call a polynomial-size quantum circuit familiy
$C = \{C_\secp\}_{\secp \in \mathbb{N}}$ uniform if there exists a polynomial-time deterministic Turing machine $M$ that on input $1^\secp$ outputs $C_\secp$.
\item Classical communication channels in the quantum setting are identical to classical communication
channels in the classical setting, except that when a set of qubits is sent through a classical
communication channel, then the qubits are automatically measured in the standard basis, and the
measured (now classical-state) qubits are then sent through the channel.
\item A quantum interactive algorithm (in a 2-party setting) has input divided into two registers and
output divided into two registers. For the input qubits, one register is for an input message from
the other party, and a second register is for a potential inner state the machine holds. For the output,
one register is for the message to be sent to the other party, and another register is for a potential
inner state for the machine to keep to itself.
\end{itemize}

\paragraph{Quantum Adversarial Model}. We would like to consider
security definitions that not only achieve quantum security, but are also composable and can be used
modularly inside other protocols. For this we think by default of security against polynomial-size
quantum adversaries with non-uniform polynomial-size quantum advice (i.e. an arbitrary quantum mixed state
that is not necessarily efficiently generatable).

An adversary will be usually denoted
by $A^* = \{A^*_\secp, \rho_\secp\}_{\secp \in \mathbb{N}}$, where $\{A^*_\secp\}_{\secp \in \mathbb{N}}$
is a polynomial-size non-uniform sequence of quantum circuits,
and $\{\rho_\secp\}_{\secp \in \mathbb{N}}$ is some polynomial-size sequence of mixed quantum states. All adversaries are implicitly
unrestricted in their behaviour (i.e. they are fully malicious and can arbitrarily deviate from protocols).
We conclude with notions regarding indistinguishability in the quantum setting.

\begin{itemize}
\item A function $f : \mathbb{N} \rightarrow \[0, 1\]$ is:
\begin{itemize}
\item negligible if for every constant $c \in \mathbb{N}$ there exists $N \in \mathbb{N}$ such that for all $n > N$, $f(n) < n^{-c}$.
\item noticeable if there exists $c \in \mathbb{N}, N \in \mathbb{N}$ s.t. for every $n \geq N$, $f(n) \geq n^{-c}$.
\end{itemize}
\item A quantum random variable is simply a random variable that can have values that are quantum
states. That is, a quantum random variable induces a probability distribution over a (possibly infinite) set of
quantum states. Such quantum random variables can also be thought of as a mixed quantum state, which is
simply a distribution over quantum states.
\item For two quantum random variables $X$ and $Y$, quantum distinguisher $D$ with quantum mixed state $\rho$
as auxiliary input, and $\mu \in \[0, 1\]$, we write $X \approx_{D(\rho),\mu} Y$ if
$$|\Pr[D(X; \rho) = 1] - \Pr[D(Y; \rho) = 1]| \leq \mu.$$
\item Two ensembles of quantum random variables $X = \{X_\secp\}_{\secp \in \mathbb{N}}$ and $Y = \{Y_\secp\}_{\secp \in \mathbb{N}}$ are said to be
computationally indistinguishable, denoted by $X \approx_c Y$, if for every polynomial-size non-uniform
quantum distinguisher with quantum advice 
$D = \{D_\secp, \rho_\secp\}_{\secp \in \mathbb{N}}$, there exists a negligible function $\mu$ such that for all $\secp \in \mathbb{N}$,
$$X_\secp \approx_{D_\secp(\rho_\secp), \mu(\secp)} Y_\secp.$$
\item The trace distance between two quantum distributions $X, Y$ , denoted by $\mathsf{TD}(X, Y)$, 
is a generalization of statistical distance to the quantum setting and represents the maximal distinguishing advantage between two quantum distributions by an unbounded quantum algorithm. We thus say that $X =
\{X_\secp\}_{\secp \in \mathbb{N}}$ and $Y = \{Y_\secp\}_{\secp \in \mathbb{N}}$ are statistically indistinguishable (and write $X \approx_s Y$), if for every
unbounded non-uniform quantum distinguisher $D = \{D_\secp\}_{\secp \in \mathbb{N}}$, there exists a negligible function $\mu$ such that for all $\secp \in \mathbb{N}$,
$\mathsf{TD}(X_\secp, Y_\secp) \leq \mu(\secp)$.
\end{itemize}

\subsection{Notation for Interactive Protocols}

Throughout, we will be considering interactive protocols, generally defined by a set of classical interactive Turing machines $\{\M_i\}_{i \in [n]}$. We denote by $\tau \gets \dist{\{\M_i(y_i)\}_{i \in [n]}}(x)$ the public transcript of their interaction on common input $x$, where each $\M_i$ has private input $y_i$. More precisely, $\tau$ consists of the messages sent between the $\{\M_i\}_{i \in [n]}$, and is a random variable over the random coins of each $\M_i$. We let $\view_{\M_j}(\dist{\{\M_i(y_i)\}_{i \in [n]}}(x))$ denote the view of some party $\M_j$ that results from this interaction, which consists of the portion of the transcript $\tau$ that includes messages sent by or received by $\M_j$, along with $\M_j$'s private state $\state$ at the end of the interaction. If $\M_j$ is a quantum machine, then $\state$ may be a quantum state. If $\M_j$ is defined to have some specific output at the end of the interaction, we denote this by $\output_{\M_j}(\dist{\{\M_i(y_i)\}_{i \in [n]}}(x))$. 

\begin{definition}[Explainable Transcript]
Let $\{\M_i\}_{i \in [n]}$ be a (classical) interactive protocol, and consider some subset of participants $\{\M_i\}_{i \in I}$. We say that a transcript $\tau$ consisting of messages sent by and received by $\{\M_i\}_{i \in I}$ is \emph{explainable} with respect to set $I$ if there exists some $\{\M^*_i\}_{i \notin I}$, inputs $\{y_i\}_{i \in I}$ and random coins $\{r_i\}_{i \in I}$ such that $\tau$ is consistent with the transcript of an execution $\dist{\{\M^*_i\}_{i \notin I},\{\M_i(y_i;r_i)\}_{i \in I}}$.
\end{definition}

\paragraph{Handling Abort and Misbehaviour.}  We set a general convention to handle publicly checkable misbehavior by parties in any interactive protocol.

\begin{itemize}
    \item For security parameter $\lambda$, for each message in the protocol, it will be known (publicly) based on $\lambda$, what is the length of each message (or upper and lower bounds on that length). If a party sends a message in an incorrect length, the receiving party fixes it locally and trivially; if the message is too long, it cuts the message in a suitable place, and if it's too short then pads with zeros.
    
    \item Whenever a party aborts, all other parties ends communication and output $\bot$. 
\end{itemize}




\subsection{Witness Indistinguishability}
We use classical constant-round proof systems for NP (where both honest prover and verifier are classical
efficient algorithms) that are witness-indistinguishable against quantum verifiers. That is, transcripts generated by the prover for two witnesses to the same instance are indistinguishable to quantum attackers.



\begin{definition}[WI Proof System for NP]
\label{def:wi}
A witness-indistinguishable proof system for a language $\cL \in \mathsf{NP}$ is a pair $(\zkP,\zkV)$ of classical PPT interactive Turing machines. $\zkP$ and $\zkV$ interact on common input $1^\secp$ and $x$, and $\zkP$ additionally takes a private input $w$. At the end of the interaction, $\zkV$ outputs a bit indicating whether it accepts or rejects. The proof system should satisfy the following properties.
\begin{enumerate}
\item {\bf Perfect Completeness:}
For any $\secp \in \mathbb{N}$, $x \in \cL \cap \{0, 1\}^\secp$, $w \in \cR_\cL(x)$,
$$\Pr[\output_\zkV \langle \zkP(w), \zkV \rangle (1^\secp,x) = 1] = 1.$$
\item {\bf Statistical Soundness:} For any non-uniform unbounded prover $\zkP^* = \{\zkP^*_\secp\}_{\secp \in \mathbb{N}}$, there exists a negligible function $\mu(\cdot)$ such that for any security parameter $\secp \in \mathbb{N}$ and any $x \in \{0, 1\}^\secp \setminus \cL$,
$$\Pr [\output_{\zkV} \langle \zkP^*_\secp, \zkV \rangle (1^\secp,x) = 1] = \mu(\secp).$$
\item {\bf Witness Indistinguishability:} 
For every non-uniform quantum polynomial-size verifier 
$\zkV^* = \{\zkV^*_\secp, \rho_\secp\}_{\secp \in \mathbb{N}}$, for any two sequences of witnesses $\{w_\secp\}_{\secp\in\mathbb{N}}, \{v_\secp\}_{\secp \in \mathbb{N}}$ s.t. for every $\secp \in \mathbb{N}$, $w_\secp$ and
$v_\secp$ are both witnesses for the same 
$x_\secp \in \cL \cap \{0, 1\}^\secp$, we have,
$$\{\view_{\zkV^*_\secp}
\langle \zkP (w_\secp), V^*_\secp (\rho_\secp) \rangle(1^\secp,x)\}_{\secp \in \mathbb{N}}
\approx_c
\{\view_{\zkV^*_\secp}
\langle \zkP (v_\secp), V^*_\secp (\rho_\secp) \rangle(1^\secp,x)\}_{\secp \in \mathbb{N}}.$$
\end{enumerate}
\end{definition}

\subsection{Sigma Protocol for NP}

\begin{definition}[Sigma Protocol for $\NP$] A sigma protocol for an $\NP$ relation $\cR$ is a pair $(\zkP = (\zkP_1,\zkP_2),\zkV = (\zkV_1,\zkV_2))$ of classical PPT Turing machines with the following syntax. Given an instance $x$ and witness $w$, $\zkP_1(x,w)$ outputs a string $\alpha$ and a prover state $\state$. $\zkV_1(1^{|x|})$ is public-coin, and outputs a uniformly random string $\beta$. Next, $\zkP_2(x,w,\state,\alpha,\beta)$ outputs a string $\gamma$ and finally, $\zkV_2(x,\alpha,\beta,\gamma)$ either accepts or rejects. The proof system should satisfy the following properties.
\begin{enumerate}
    \item {\bf Completeness: } For any $\secp \in \mathbb{N}, x \in \cL \cap \{0,1\}^\secp,w \in \cR_\cL(x),$ $$\Pr[\mathsf{OUT}_\zkV\dist{\zkP(w),\zkV}(x) = 1] = 1.$$
    \item {\bf Statistical Soundness: } For any non-uniform unbounded prover $\zkP^* = \{\zkP^*_\secp\}_{\secp \in \mathbb{N}}$, there exists a negligible function $\mu(\cdot)$ such that for any $\secp \in \mathbb{N}$ and any $x \in \{0, 1\}^\secp \setminus \cL$, $$\Pr[\mathsf{OUT}_\zkV\dist{\zkP^*_\secp,\zkV}(x) = 1] \leq \mu(\secp).$$
    \item {\bf Special Zero-Knowledge: } There exists a PPT simulator $\zkSim$ such that for any $\{x_\secp,w_\secp\}_{\secp \in \mathbb{N}}$ where $|x_\secp| = \secp$ and $(x_\secp,w_\secp) \in \cR$, and $\{\beta_\secp\}_{\secp \in \mathbb{N}}$ where $|\beta_\secp| = \secp$, $$\{(\alpha,\gamma) \ | \ (\alpha,\state) \gets \zkP_1(x_\secp,w_\secp), \gamma \gets \zkP_2(x_\secp,w_\secp,\state,\alpha,\beta_\secp)\}_{\secp \in \mathbb{N}} \approx_c \{(\alpha,\gamma) \gets \zkSim(x_\secp,\beta_\secp)\}_{\secp \in \mathbb{N}}.$$
    Observe that due to the prover's first message being generated independently of the verifier's message, this implies that for any $\{x_\secp,w_\secp\}_{\secp \in \mathbb{N}}$ where $|x_\secp| = \secp$ and $(x_\secp,w_\secp) \in \cR$,
    $$\{\alpha \gets \zkP_1(x_\secp,w_\secp)\}_{\secp \in \mathbb{N}} \approx_c \{\alpha \ | \ (\alpha,\gamma) \gets \zkSim(x_\secp,0^\secp)\}_{\secp \in \mathbb{N}}.$$
    We refer to this as \bf{First-Message Indistinguishability}.
\end{enumerate}
\end{definition}

Sigma protocols are known to follow from classical zero-knowledge proof systems such as the (parallel repetition) of the 3-coloring protocol~\cite{GolMicWig91}, which is in turn based on non-interactive perfectly-binding and computationally hiding commitments.

\subsection{Non-Interactive Commitment}


\begin{definition}[Quantum-secure Non-interactive Commitment]\label{def:com} A non-interactive commitment is defined by a PPT algorithm $\COM$ that takes as input security parameter $1^\secp$ 
and $x \in \{0,1\}^*$, and outputs a commitment $c$. The commitment algorithm satisfies:
\begin{enumerate}
\item {\bf Perfect Binding:} 
For any $x, x' \in \{0, 1\}^*$ of the same length, 
if $c \in \COM(1^\secp, x), c \in \COM(1^\secp, x')$, then $x = x'$.
\item {\bf Quantum Computational Hiding:} 
For any pair of $\poly(\secp)$-length strings $x_0 = \{x_{0,\secp}\}_{\secp \in \mathbb{N}}, x_1 =
\{x_{1,\secp}\}_{\secp \in \mathbb{N}}$, we have,
$$ \{\COM(1^\secp,x_{0,\secp})\}_{\secp \in \mathbb{N}} \approx_c \{\COM(1^\secp,x_{1,\secp})\}_{\secp \in \mathbb{N}}.$$
\end{enumerate}
\end{definition}

\paragraph{Instantiations.} Non-interactive commitments with quantum hiding are known based on various standard assumptions, including LWE \cite{GoyalHKW17}. 

\subsection{Compute and Compare Obfuscation}\label{sec:CC}

We start by defining the class of {\em compute and compare circuits.}

\begin{definition}[Compute and compare]
Let $f:\zo^n\rightarrow\zo^\secp$ be a circuit, and let $u\in\zo^\secp$ and $z \in \{0,1\}^*$ be two strings. Then $\CC{f}{u}{z}(x)$ is a circuit that returns $z$ if $f(x)=u$, and $0$ otherwise.  
\end{definition}

We now define compute and compare (CC) obfuscators with perfect correctness. In what follows $\obf$ is a \PPT algorithm that takes as input a CC circuit $\CC{f}{u}{z}$ and outputs a new circuit $\obfC$. (We assume that the CC circuit $\CC{f}{u}{z}$ is given in some canonical description from which $f$, $u$, and $z$ can be read.) 

\begin{definition}[CC obfuscator]
An algorithm $\obf$ is a compute and compare obfuscator if it satisfies:
\begin{enumerate}
\item
{\bf Perfect correctness:} For any circuit $f:\zo^n\rightarrow\zo^\secp, u\in \zo^\secp, z \in \{0,1\}^*$,
$$
\Pr\[\forall x \in \{0,1\}^n: \obfC(x)= \CC{f}{u}{z}(x) \pST \obfC\gets \obf(\CC{f}{u}{z})\]=1\enspace.
$$

\item
{\bf Simulation:} There exists a $\PPT$ simulator $\ccSim$ such that for any polynomial-size quantum circuit family $f = \{f_\secp\}_{\secp \in \mathbb{N}}$ and polynomial-length output string $z = \{z_\secp\}_{\secp \in \mathbb{N}}$,
$$\{\obfC | u \leftarrow U_\secp, \obfC \leftarrow \obf(\CC{f_\secp}{u}{z_\secp})\}_{\secp \in \mathbb{N}} \approx_c \{\mathsf{Sim}(1^{|f_\secp|}, 1^{|z_\secp|}, 1^\secp\}_{\secp \in \mathbb{N}}.$$
\end{enumerate}
\end{definition}

\paragraph{Instantiations.} Compute and compare obfuscators with almost perfect correctness are constructed in \cite{FOCS:GoyKopWat17,FOCS:WicZir17} based on quantum LWE, and recently with perfect correctness in~\cite{EPRINT:GKVW19} based on quantum LWE.

\subsection{Function-Hiding Secure Function Evaluation}

We define secure function evaluation protocols with statistical circuit privacy and quantum input privacy.

\begin{definition}[2-Message Function Hiding SFE]
\label{def:sfe}
A two-message secure function evaluation protocol $(\sfegen, \sfeenc, \sfeeval, \sfedec)$ has the following syntax:
\begin{itemize}
    \item $\dk \leftarrow \sfegen(1^\secp):$ a probabilistic algorithm that takes a security parameter $1^\secp$ and outputs a secret key $\dk$.
    \item $\ct \leftarrow \sfeenc(\dk,x):$ a probabilistic algorithm that takes a string $x \in \{0,1\}^*$ and outputs a ciphertext $\ct$.
    \item $\widehat{\ct} \leftarrow \sfeeval(C, \ct):$ a probabilistic algorithm that takes a classical circuit $C$ and ciphertext $\ct$ and outputs an  evaluated ciphertext $\widehat{\ct}$.
    \item $\widehat{x} = \sfedec(\dk,\widehat{\ct}):$ a deterministic algorithm that takes a ciphertext $\widehat{\ct}$ and outputs a string $\widehat{x}$.
 \end{itemize}
For any polynomial-size family of classical circuits $\mathcal{C} = \{\mathcal{C}_\secp\}_{\secp \in \mathbb{N}}$ the scheme satisfies:
\begin{itemize}
    \item {\bf Perfect Correctness:}
    For any $\secp \in \mathbb{N}, x \in \{0,1\}^*$ and circuit $C \in \mathcal{C}_\secp$,
    $$\Pr[\sfedec_\dk(\widehat{\ct}) = C(x) | \dk \leftarrow \sfegen(1^\secp), \ct \leftarrow \sfeenc_{\dk}(x), \widehat{\ct} \leftarrow \sfeeval(C, \ct)] = 1 $$
    \item {\bf Quantum Input Privacy: } For polynomial $\ell(\secp)$ and polynomial-size quantum adversary $\A^* = \{\A^*_\secp, \rho_\secp\}_{\secp \in \mathbb{N}}$, there exists a negligible function $\mu (\cdot)$ such that for every two length $\ell(\secp)$ messages $\{x_{0,\secp}\}_{\secp \in \mathbb{N}}, \{x_{0,\secp}\}_{\secp \in \mathbb{N}}$ for every $\secp \in \mathbb{N}$:
    $$\Pr[\A_\secp^*(\ct) = b| \dk \leftarrow \sfegen(1^\secp), \ct \leftarrow \sfeenc_{\dk}(x)] \leq \frac{1}{2} + \mu(\secp) $$
    \item {\bf Statistical Circuit Privacy: } There exist unbounded algorithms, probabilistic $\ccSim$ and deterministic $\mathsf{Ext}$ such that for every $x \in \{0,1\}^*, \ct \in \sfeenc(x)$, the extractor outputs $\mathsf{Ext}(\ct) = x$ and:
    $$\{\sfeeval(C, \ct^*)\}_{\substack{\secp \in \mathbb{N}, C \in \mathcal{C}_\secp, \\\ct^* \in \{0,1\}^{\poly(\secp)}}} \approx_s \{\mathsf{Sim}(C (\mathsf{Ext}(\ct^*; 1^\secp)); 1^\secp)\}_{\substack{\secp \in \mathbb{N}, C \in \mathcal{C}_\secp, \\ \ct^* \in \{0,1\}^{\poly(\secp)}}}$$
    Specifically, there exists a constant $c>0$ such that for large enough $\secp$, the statistical distance between the two distributions is at most $2^{-\secp^c}$.
\end{itemize}
\end{definition}

We will use the following claim in our analysis. This follows directly from the statistical circuit privacy property.

\begin{claim}[Evaluations of Agreeing Circuits are Statistically Close]
Let $\ct^* = \{\ct^*_\secp\}_{\secp \in \mathbb{N}}$ be any (possibly non-ciphertext) $\poly(\secp)$ length string and let $C_0 = \{\mathcal{C}_{0,\secp}\}_{\secp \in \mathbb{N}}, C_1 = \{\mathcal{C}_{1,\secp}\}_{\secp \in \mathbb{N}}$ be two families of circuits such that for all $\secp \in \mathbb{N}$, $C_{0,\secp}$ and $C_{1,\secp}$ have identical truth tables.
Then
$$\{\sfeeval(C_0, \ct^*)\}_{\secp \in \mathbb{N}, C_0 \in \mathcal{C}_{0,\secp}} \approx_s 
\{\sfeeval(C_1, \ct^*)\}_{\secp \in \mathbb{N}, C_1 \in \mathcal{C}_{1,\secp}}$$
Specifically, there exists a constant $c>0$ such that for large enough $\secp$, the statistical distance between the two distributions is at most $2^{-\secp^c}$.
\end{claim}

Secure function evaluation schemes satisfying Definition~\ref{def:sfe} for functions in $\mathsf{NC1}$ are known based on quantum hardness of LWE~\cite{TCC:BraDot18}.

We also define a superpolynomially secure variant of $2$-message function hiding SFE where the quantum input privacy property restricts adversaries to having smaller than inverse superpolynomial advantage, for a small superpolynomial function. 

\begin{definition}[2-Message Function Hiding SFE]
\label{def:sfe2}
A two-message SFE protocol with superpolynomial security is identical to the definition in Definition \ref{def:sfe}, except that it modifies the quantum input privacy requirement as follows:
There exists a constant $c>0$ such that for polynomial $\ell(\secp)$ and polynomial-size quantum adversary $\A^* = \{\A^*_\secp, \rho_\secp\}_{\secp \in \mathbb{N}}$, there exists a negligible function $\mu (\cdot)$ such that for every two length $\ell(\secp)$ messages $\{x_{0,\secp}\}_{\secp \in \mathbb{N}}, \{x_{0,\secp}\}_{\secp \in \mathbb{N}}$ for every $\secp \in \mathbb{N}$:
    $$\Pr[\A_\secp^*(\ct) = b| \dk \leftarrow \sfegen(1^\secp), \ct \leftarrow \sfeenc_{\dk}(x)] \leq \frac{1}{2} + \mu(\secp^{\mathsf{ilog}(c,\secp)}) $$
\end{definition}

Secure function evaluation schemes satisfying Definition~\ref{def:sfe2} for functions in $\mathsf{NC1}$ can be based on quantum slightly superpolynomial hardness of LWE~\cite{TCC:BraDot18}. Specifically, we assume that QPT distinguishers have advantage at most $\negl(\secp^{\mathsf{ilog}(c,\secp)})$ in distinguishing LWE samples from uniformly random matrices.

\subsection{Quantum Rewinding Lemma}

We will make use of the following lemma from~\cite{10.1137/060670997} and re-worded in~\cite{BS20}. 

\begin{lemma}\label{lemma:rewinding} There is a quantum algorithm $\R$ that gets as input:
\begin{itemize}
    \item A general quantum circuit $\Q$ with $n$ input qubits that outputs a classical bit $b$ and an additional $m$ qubits.
    \item An $n$-qubit state $\ket{\psi}$.
    \item A number $t \in \mathbb{N}$.
\end{itemize}
$\R$ executes in time $t \cdot \poly(|\Q|)$ and outputs a distribution over $m$-qubit states $D_\psi \coloneqq \R(\Q,\ket{\psi},t)$ with the following guarantees.

For an $n$-qubit state $\ket{\psi}$, denote by $\Q_\psi$ the conditional distribution of the output distribution $\Q(\ket{\psi})$, conditioned on $b=0$, and denote by $p(\psi)$ the probability that $b=0$. If there exist $p_0,q \in (0,1), \epsilon \in (0,\frac{1}{2})$ such that:
\begin{itemize}
    \item Amplification executes for enough time: $t \geq \frac{\log(1/\epsilon)}{4 \cdot p_0(1-p_0)},$
    \item There is some minimal probability that $b=0:$ For every $n$-qubit state $\ket{\psi},p_0 \leq p(\psi)$,
    \item $p(\psi)$ is input-independent, up to $\epsilon$ distance: For every $n$-qubit state $\ket{\psi}$, $|p(\psi)-q| < \epsilon$, and 
    \item $q$ is closer to $\frac{1}{2}:$ $p_0(1-p_0) \leq q(1-q)$,
\end{itemize}
then for every $n$-qubit state $\ket{\psi},$

$$\mathsf{TD}\left(\Q_\psi,D_\psi\right) \leq 4\sqrt{\epsilon}\frac{\log(1/\epsilon)}{p_0(1-p_0)}.$$

\end{lemma}

\section{Quantum Multi-Key Fully-Homomorphic Encryption}
\label{sec:spooky}

\subsection{Learning with Errors and Lattice Trapdoors}
\label{subsec:trapdoors}

The (decisional) learning with errors problem (LWE), introduced by~\cite{STOC:Regev05}, is parameterized by a modulus $q$, positive integers $n,m$, and an error distribution $\chi$. It asks to distinguish between the distributions $(\bA,\bA\bs+\be \bmod q)$ and $(\bA,\bu)$, where $\bA$ is uniformly random in $\bZ_q^{m \times n}$, $\bs$ is uniformly random in $\bZ_q^n$, $\bu$ is uniformly random in $\bZ_q^m$, and $\be$ is chosen from $\chi^m$. As shown in~\cite{STOC:Regev05,STOC:PeiRegSte17}, for \emph{any} sufficiently large modulus $q$, the LWE problem where $\chi$ is a discrete Gaussian distribution with parameter $\sigma = \alpha q \geq 2\sqrt{n}$ (i.e. the distribution over $\bZ$ where the probability of $x$ is proportional to $e^{-\pi(|x|/\sigma)^2}$), is at least as hard as approximating the shortest independent vector problem (SIVP) to within a factor of $\gamma = \tilde{O}(n/\alpha)$ in \emph{worst case} dimension $n$ lattices. One can truncate the discrete Gaussian distribution to have support only over integers bounded in absolute value by $\sigma \cdot \omega(\sqrt{\log(\secp)})$ while only introducing a negligible difference. Thus, we will use the fact that $\chi$ may be a $B$-bounded distribution, for some value $B$.

We will make use of the notion of a lattice trapdoor, defined in the following theorem~\cite{EC:MicPei12}.

\begin{theorem}[\cite{ICALP:Ajtai99,EC:MicPei12}]
\label{thm:trapdoor}
There is an efficient randomized algorithm $\gentrap(1^n,1^m,q)$ that, given any integers $n \geq 1, q \geq 2$, and sufficiently large $m = O(n\log q)$, outputs a matrix $\bA \in \bZ_q^{m \times n}$ and a trapdoor $\tau_\bA$ such that the distribution of $\bA$ is negligibly (in $n$) far from the uniform distribution. Moreover, there is an efficient deterministic algorithm $\invert$ that on input $\bA,\tau_\bA,$ and $\bA\bs+\be$, where $\bs$ is arbitrary in $\bZ_q^n$ and $||\be|| \leq q/(O(n\log q))$, returns $\bs$ and $\be$ with overwhelming probability over $(\bA,\tau_\bA) \gets \gentrap(1^n,1^m,q)$.
\end{theorem}

In fact, we'll need a slightly stronger version of the above statement. In particular, we will actually need the correctness of $\invert$ to hold \emph{perfectly} rather than statistically over the randomness of $\gentrap$. This can be arranged by slightly tweaking the $\gentrap$ procedure.


\begin{lemma}
\label{lemma:perfecttd}
There exist algorithms $\mathsf{GenTrap}$ and $\mathsf{Invert}$ as described in~\cref{thm:trapdoor} where $\mathsf{Invert}$ returns $s,e$ with probability 1.
\end{lemma}

\begin{proof} (Sketch) Call a matrix-trapdoor pair $(\bA,\tau_\bA)$ ``functional'' if $\bA$ is full rank (rank $n$) mod $q$, $\tau_\bA$ is an $m \times m$ matrix such that $\tau_\bA \cdot \bA = 0 \bmod q$,\footnote{The trapdoor generation procedure presented in~\cite{EC:MicPei12} actually results in an ``inhomogeneous'' trapdoor, where it holds that $\tau_\bA \cdot \bA = \mathbf{G}$, for the gadget matrix $\mathbf{G}$. However, one can derive a trapdoor satisfying $\tau_\bA \cdot \bA = 0$ from an inhomogeneous trapdoor.} each entry of $\tau_\bA$ is ``small enough'', and $\tau_\bA$ is full rank over the rationals. Such a functional matrix-trapdoor pair may be used to invert \emph{any} vector $\bv = \bA\bs + \be$, for small enough $\be$, as follows. Left multiply $\bv$ by $\tau_\bA$ over $\bZ_q$, and then left multiply the result by $\tau_\bA^{-1}$ over the rationals, which recovers $\be$. Then subtract $\be$ from $\bv$ and recover $\bs$ by linear algebra. Now observe that the four conditions for $(\bA,\tau_\bA)$ to be functional are all efficiently checkable. Thus, the modified $\gentrap$ algorithm can operate as follows. Sample $(\bA,\tau_\bA)$ as before, then check if it is functional, and if not replace $(\bA,\tau_\bA)$ with some fixed functional pair. Since $\gentrap$ only outputs a non-functional pair with negligible probability, this modification maintains the requirement that the distribution of $\bA$ is negligibly close to uniform.
\end{proof}

\subsection{Definition}

\begin{definition}[Quantum Multi-Key Fully-Homomorphic Encryption (QMFHE)]
\label{def:qmfhe}
A quantum multi-key fully-homomorphic encryption scheme is given by six algorithms ($\qmfhe.\keygen$, $\qmfhe.\enc$, $\qmfhe.\qenc$, $\qmfhe.\eval$, $\qmfhe.\dec$, $\qmfhe.\qdec$) with the following syntax.
\begin{itemize}
    \item $(\pk,\sk) \leftarrow \qmfhe(1^\secp)$ : A PPT algorithm that given a security parameter, samples a classical public key and a classical secret key.
    \item $c \leftarrow \qmfhe.\enc(\pk,b)$ : A PPT algorithm that takes as input a bit $b$ and outputs a classical ciphertext.
    \item $|\phi\rangle \leftarrow \qmfhe.\qenc(\pk,|\psi\rangle)$ : A QPT algorithm that takes as input a qubit $|\psi\rangle$ and outputs a ciphertext represented in qubits.
    \item $\widehat{c},\ket{\widehat{\phi}} \gets \qmfhe.\eval((\pk_1,\dots,\pk_n),C,(\ket{\phi_1},\dots,\ket{\phi_n}))$: A QPT algorithm that takes as input 
    \begin{enumerate}
        \item A set of $n$ public keys.
        \item A general quantum circuit with $\ell_1 + \dots + \ell_n$ input qubits and $\ell'$ output qubits, out of which $m$ are measured.
        \item A set of $n$ ciphertexts where $\ket{\phi_i}$ encrypts an $\ell_i$-qubit state under public key $\pk_i$. Some of the $\ell_i$ ciphertexts are possibly classical ciphertexts (generated by the classical encryption algorithm) encrypting classical bits.
    \end{enumerate}
    The evaluation algorithm outputs a classical ciphertext $\widehat{c}$ encrypting $m$ bits (under keys $\pk_1,\dots,\pk_n$), plus a quantum ciphertext $\ket{\widehat{\phi}}$ encrypting an $(\ell'-m)$-qubit quantum state (under keys $\pk_1,\dots,\pk_n$).
    \item $b \gets \qmfhe.\dec((\sk_1,\dots,\sk_n),c)$: A PPT algorithm that takes as input a set of $n$ secret keys and a classical ciphertext $c$ and outputs a bit.
    \item $\ket{\psi} \gets \qmfhe.\qdec((\sk_1,\dots,\sk_n),\ket{\phi})$: A QPT algorithm that takes as input a set of $n$ secret keys and a quantum ciphertext $\ket{\phi}$ and outputs a qubit.
\end{itemize}


The scheme satisfies the following.
\begin{enumerate}
\item
{\bf Quantum Semantic Security:} The encryption algorithm maintains quantum semantic security.
\item
{\bf Compactness:} There exists a polynomial $\mathsf{poly}(\cdot)$ s.t. for every quantum circuit $C$ with $\ell'$ output qubits and an encryption of an input for $C$, the output size of the evaluation algorithm is $\mathsf{poly}(\secp, \ell')$,
where $\secp$ is the security parameter of the scheme.
    \item {\bf Classicality-Preserving Quantum Homomorphism: } Let $C = \{C_\secp\}_{\secp \in \mathbb{N}}$ be a polynomial-size quantum circuit, where $C_\secp$ has $\ell_1(\secp) + \dots + \ell_n(\secp)$ input qubits and $\ell'(\secp)$ output qubits, of which $m(\secp)$ are measured. Let $\ket{\phi_1},\dots,\ket{\phi_n} = \{\ket{\phi_1}_\secp,\dots,\ket{\phi_n}_\secp\}_{\secp \in \mathbb{N}}$ be an input state for $C$, let $(\pk_1,\sk_1),\dots,(\pk_n,\sk_n) = \{(\pk_1,\sk_1)_\secp,\dots,(\pk_n,\sk_n)_\secp\}_{\secp \in \mathbb{N}}$ be pairs of public and secret keys ($\forall i \in [n], \secp \in \mathbb{N}, (\pk_i,\sk_i)_\secp \in \qmfhe.\keygen(1^\secp)$) and let $r_1,\dots,r_n = \{(r_1)_\secp,\dots,(r_n)_\secp\}_{\secp \in \mathbb{N}}$ be $n$ random strings for the encryption algorithm. Then there exists a negligible function $\mu(\cdot)$ such that for all $\secp \in \mathbb{N}$, $$\mathsf{TD}(\rho_{0,\secp},\rho_{1,\secp}) \leq \mu(\secp),$$ where $\rho_0,\rho_1$ are quantum distributions defined as follows:
    \begin{itemize}
        \item $\rho_{0,\secp}$: For each $i \in [n]$, encrypt each classical bit of $\ket{\phi_i}$ with $\qmfhe.\enc(\pk_i,\cdot)$ and the rest with $\qmfhe.\qenc(\pk_i,\cdot)$ (using randomness $r_i$). Execute $\qmfhe.\eval((\pk_1,\dots,\pk_n),C,\cdot)$ on the $n$ encryptions to get $\widehat{c},\ket{\widehat{\phi}}$, where $\widehat{c}$ is a classical ciphertext encrypting $m(\secp)$ bits. Then output\\ $\qmfhe.\dec((\sk_1,\dots,\sk_n),\widehat{c}),\qmfhe.\qdec((\sk_1,\dots,\sk_n),\ket{\widehat{\phi}})$.
        \item $\rho_{1,\secp}$: Output $C(\ket{\phi_1,\dots,\phi_n})$.
    \end{itemize}
\end{enumerate}
\end{definition}

Known \emph{classical} LWE-based constructions of multi-key fully-homomorphic encryption~\cite{C:CleMcG15,EC:MukWic16,TCC:PeiShi16,TCC:BraHalPol17} do not quite satisfy the above syntax.\footnote{Though there are NTRU-based constructions that do~\cite{STOC:LopTroVai12,EPRINT:AJJM20}.} Instead, they relax the syntax to allow for some notion of setup. In this work, we will be interested in the notion of \emph{distributed setup} which was achieved in the classical setting by~\cite{TCC:BraHalPol17}.

\begin{definition}[QMFHE with Distributed Setup]
\label{def:distributed-setup}
A QMFHE scheme $\qmfhe$ has distributed setup if it includes the following algorithm.
\begin{itemize}
    \item $\qmfhe.\setup(1^\secp,1^n,i)$: A PPT algorithm that takes as input the security parameter, a number of parties, and an index $i \in [n]$, and outputs a string $\pp_i$.
\end{itemize}
We then define the public parameters of the scheme $\pp = (\pp_1,\dots,\pp_n)$ and assume that all other algorithms take $\pp$ as input.
\end{definition}

\begin{remark}
This notion of distributed setup gives rise to a stronger notion of semantic security, which considers \emph{rushing} adveraries that may generate $\{\pp_j\}_{j \neq i}$ maliciously, possibly depending on $\pp_i$. More formally, in the security game the adversary first picks $n$ and an $i \in [n]$ and sends these to its challenger. The challenger then runs $\pp_i \gets \qmfhe.\setup(1^\secp,1^n,i)$, and returns $\pp_i$ to the adversary. Then, the adversary generates $\{\pp_j\}_{j \in [n] \setminus \{i\}}$ arbitrarily and sends these to its challenger. Finally, the challenger draws a public key secret key pair based on $\{\pp_i\}_{i \in [n]}$, and the semantic security game continues are usual. This notion of semantic security was achieved in the classical setting by~\cite{TCC:BraHalPol17}.
\end{remark}

In this work, we also consider a more stringent requirement on the operation of the $\qmfhe.\dec$ algorithm, which we call \emph{nearly linear decryption of classical ciphertexts}. Essentially, this states that decrypting a classical ciphertext $c$ encrypted under keys $\sk_1,\dots,\sk_n$ amounts to computing a linear function $\cL_c$ (defined by $c$) on the concatenated secret keys $[\sk_1 \ | \ \dots \ | \ \sk_n]$ modulo some integer $q$, and then rounding.

\begin{definition}[QMFHE with Nearly Linear Decryption of Classical Ciphertexts]
\label{def:lindec}
A QMFHE scheme $\qmfhe$ has nearly linear decryption of classical ciphertexts if the $\qmfhe.\dec$ algorithm operates as follows.
\begin{itemize}
    \item $\qmfhe.\dec((\sk_1,\dots,\sk_n),c)$: There is an efficiently computable linear function $\cL_c$ (determined by $c$) and an (even) integer $q$ such that the decryption prodecure computes $$\cL_c(\sk_1,\dots,\sk_n) = b \cdot q/2 + e \bmod q$$ (where $e < q/4$) and returns $b \in \{0,1\}$. Equivalenty, we can define linear functions $\cL_c^{(1)},\dots,\cL_c^{(n)}$ such that the decryption procedure computes $$\sum_{i \in [n]}\cL_c^{(i)}(\sk_i) = b \cdot q/2 + e \bmod q.$$
\end{itemize}
\end{definition}

Finally, we remark that we do not consider an additional security property often found in classical constructions of multi-key FHE, which roughly stipulates that partial decryptions of other parties may be simulated. This property is most relevant when considering the direct application of multi-key FHE to MPC, but we will not need it in this work.

\subsection{Background} 
\label{subsec:background}

We follow the template given by~\cite{FOCS:Mahadev18b} for constructing a quantum fully-homormorphic encryption scheme.~\cite{FOCS:Mahadev18b} essentially showed that a classical fully-homomorphic scheme $\fhe$ can be converted into a quantum fully-homomorphic scheme, as long as it has a few additional properties, necessary for computing the so-called \emph{encrypted CNOT} operation (properties 2-4 in~\cref{def:quantum-capable}). 

Unfortunately, no known fully-homomorphic schemes immediatedly satisfy these properties. However,~\cite{FOCS:Mahadev18b} observed that the dual Regev (non-fully-homomorphic) encryption scheme of~\cite{STOC:GenPeiVai08} does satisfy these properties, and moreover, and that there exists a fully-homomorphic encryption scheme $\fhe$ (the dual version of~\cite{C:GenSahWat13}) with an efficient procedure for converting an $\fhe$ ciphertext into a dual Regev ciphertext encrypting the same message. In fact, the conversion procedure presented in~\cite{FOCS:Mahadev18b} simply consists of taking the last column of the $\fhe$ ciphertext. This suffices to give a quantum fully-homomorphic encryption scheme, since before every encrypted CNOT operation, the evaluator can convert any $\fhe$ ciphertext needed during the operation into a dual Regev ciphertext, and then proceed. The fourth property below is needed in order to convert an evaluated dual Regev ciphertext back into an $\fhe$ ciphertext upon completion of the encrypted CNOT operation.\footnote{For technical reasons, the \emph{randomness} in the dual Regev ciphertext must also be recovered, which motivates the need for a trapdoor rather than merely a secret key.}

This motivated the definition of a quantum-capable fully-homomorphic encryption scheme given in~\cite{FOCS:Mahadev18b}. Such a scheme admits an efficient procedure that converts ciphertexts into ciphertexts of an alternate encryption scheme $\ahe$ that satisfies the properties necessary to carry out the encrypted CNOT operation.~\cite{FOCS:Mahadev18b} showed that any such FHE scheme gives rise to an FHE scheme that can additionally encrypt quantum states and evaluate quantum circuits. Below we give the analogous definition for \emph{multi-key} fully-homomorphic encryption, and follow with a sketch of the analogous conversion from quantum-capable multi-key fully-homomorphic encryption to quantum multi-key fully-homomorphic encryption.


\begin{definition}[Quantum-Capable Multi-Key Fully-Homomorphic Encryption Scheme]
\label{def:quantum-capable}
Let $\mfhe$ be a classical multi-key fully-homomorphic encryption scheme. $\mfhe$ is quantum-capable if i) its $\keygen$ procedure outputs a public key $\pk$, secret key $\sk$, and ``trapdoor'' $\tau$, and ii) there exists an alternate encryption scheme $\ahe$ such that the following properties holds.

\begin{enumerate}
    \item There exists an algorithm $\mfhe.\convert$ that takes as input a set of public keys $\pk_1,\dots,\pk_t$ and a ciphertext $c$ encrypted under $\pk_1,\dots,\pk_t$, and outputs an encryption $\widehat{c}$ under $\ahe$ with public keys $\pk_1,\dots,\pk_t$, where $c$ and $\widehat{c}$ encrypt the same value.
    \item There exists an invertible operation $\oplus_H$ (which may depend on $\pk_1,\dots,\pk_t$) on $\ahe$ ciphertexts such that, for all $x_0,x_1,a \in \{0,1\}$, $\ahe.\enc((\pk_1,\dots,\pk_t),x_0) \oplus_H a\cdot\ahe.\enc((\pk_1,\dots,\pk_t),x_1)$ is an $\ahe$ encryption of $x_0 \oplus a \cdot x_1$ under $\pk_1,\dots,\pk_t$.
    \item There exists a distribution $\mathcal{D}$ (which may depend on $\pk_1,\dots,\pk_t$) such that for all ciphertexts $c$ that can arise during homomorphic evaluation,\footnote{This set will consist of all ciphertexts with noise below some fixed bound.} \begin{align*}&\{\ahe.\enc((\pk_1,\dots,\pk_t),x;r) \ | \ (x,r) \leftarrow \mathcal{D}\}\\ \approx_s &\{\ahe.\enc((\pk_1,\dots,\pk_t),x;r) \oplus_H c \ | \ (x,r) \leftarrow \mathcal{D}\},\end{align*} and there is an efficient procedure for generating the superposition $\sum_{x,r}\sqrt{\mathcal{D}(x,r)}\ket{x,r}.$
    \item There exists an efficient function $f$ such that for any $c = \ahe.\enc((\pk_1,\dots,\pk_t),x;r)$, $f((\tau_1,\dots,\tau_t),c) = (x,r)$.
\end{enumerate}
\end{definition}

\paragraph{From quantum-capability to multi-key quantum FHE}. We now sketch, following~\cite{FOCS:Mahadev18b}'s approach in the single-key setting, how a quantum-capable multi-key FHE scheme gives rise to a full-fledged quantum multi-key FHE scheme. The following description makes use of the quantum one-time pad (QOTP), which is a method of \emph{perfectly} encrypting arbitrary quantum states $\ket{\phi}$ using classical bits $k$. We refer the reader to~\cite{FOCS:Mahadev18b} for details about how the QOTP is constructed and proven secure. We do not present details about how individual quantum gates are evaluated, or the inner workings of the encrypted CNOT operation, electing instead to present a high-level picture. Again we refer the reader to~\cite{FOCS:Mahadev18b} for all of these details. The following assumes a quantum-capable multi-key FHE scheme $\qcmfhe$, and describes a quantum multi-key FHE scheme $\qmfhe$.

\begin{itemize}
    \item \textbf{Key generation}. This procedure generates $(\pk',\sk',\tau) \gets \qcmfhe.\keygen(1^\secp)$, computes a ciphertext $\ct^{(\tau)} \gets \qcmfhe.\enc(\pk',\tau)$, and sets the public key of $\qmfhe$ to $\pk \coloneqq (\pk',\ct^{(\tau)})$ and the secret key to $\sk \coloneqq \sk'$.
    \item \textbf{Encryption}. To encrypt a quantum state $\ket{\phi}$, sample a random QOTP key $k$ and release the ciphertext $\ct = (\qcmfhe.\enc(\pk',k),\text{QOTP}(k,\ket{\phi}))$. To encrypt a classical string $m$, sample a classical one-time pad key $k$ and release $\ct = (\qcmfhe.\enc(\pk',k),k \oplus m)$.
    \item \textbf{Homomorphic evaluation}. This operation takes as input $t$ public keys $\{\pk_i = (\pk_i',\ct_i^{(\tau)})\}_{i \in [t]}$ and $t$ ciphertexts $\{\ct_i = (\cm_i,\ket{\cm}_i)\}_{i \in [t]}$. It first expands each $\cm_i$ into a multi-key ciphertext $\cm'_i$ encrypted under all public keys $\pk'_1,\dots,\pk'_t$, and gathers all components into a quantum multi-key ciphertext $(\hat{\cm},\ket{\hat{\cm}}) = ((\cm'_1,\dots,\cm'_t),(\ket{\cm}_1,\dots,\ket{\cm}_t))$. It also expands and concatenates the ciphertexts $\{\ct_i^{(\tau)}\}_{i \in [t]}$ to produce an evaluation key $\hat{\ct}^{(\tau)}$ that encrypts the trapdoors $(\tau_1,\dots,\tau_t)$ under all public keys $\pk'_1,\dots,\pk'_t$. Next, it applies a quantum circuit gate by gate on the ciphertext, as follows.
    \begin{itemize}
        \item If the gate is of a particular type, namely, it is a Clifford operator, then homomorphically evaluating the gate can be done via a \emph{parallel} procedure, where a classical circuit is applied homomorphically over $\hat{\cm}$ to produce $\hat{\cm}'$, and a quantum circuit is applied \emph{directly} to $\ket{\cm}$ to produce $\ket{\hat{\cm}'}$. 
        \item Any universal gate set for quantum computation must contain at least one non-Clifford operator, and~\cite{FOCS:Mahadev18b} includes the Toffoli gate.~\cite{FOCS:Mahadev18b} gives a procedure for homomorphically applying the Toffoli gate that involves parallel operations as above along with an encrypted CNOT operation, which requires the following manipulation. First, $\qcmfhe.\convert$ is run on $\hat{\cm}$ to produce an $\ahe$ ciphertext $\hat{\mathsf{d}}$. This ciphertext is used to define a quantum circuit that is applied to $\ket{\cm}$ to produce $\ket{\cm'}$ along with a (measured) $\ahe$ ciphertext $\hat{\mathsf{d}}'$. Finally, the function $f$ (defined in property 4 of~\cref{def:quantum-capable}), with ciphertext $\hat{\mathsf{d}}'$ hard-coded, is applied homomorphically over $\hat{\ct}^{(\tau)}$ to produce a classical ciphertext $\cm'$ encrypting the message and randomness from $\hat{\mathsf{d}}'$. 
    \end{itemize}
    \item \textbf{Decryption}. Given a ciphertext $(\hat{\cm},\ket{\hat{\cm}})$ encrypting a qubit under $\pk_1,\dots,\pk_t$, and corresponding $\qcmfhe$ secret keys $\sk'_1,\dots,\sk'_t$, this operation runs $\qcmfhe.\dec.((\sk'_1,\dots,\sk'_t),\hat{\cm})$ to produce a key $k$, and then uses $k$ to decrypt the one-time padded state $\ket{\hat{\cm}}$. The same procedure works if $\ket{\hat{\cm}}$ was instead a classical string $k \oplus m$.
    
    
\end{itemize}

\subsection{Construction}
\label{subsec:keyswitch}

Existing classical multi-key fully-homorphic encryption schemes~\cite{STOC:LopTroVai12,C:CleMcG15,EC:MukWic16,TCC:PeiShi16,TCC:BraHalPol17,EPRINT:AJJM20} in the literature do not appear to admit a simple conversion procedure necessary for quantum-capability, such as the one enjoyed by dual-GSW in the single-key setting. However, we show that indeed there exists a general conversion procedure that works for \emph{any} multi-key fully-homomorphic encryption scheme with \emph{nearly linear decryption} (see~\cref{def:lindec}). This method is essentially key-switching (see for example~\cite{FOCS:BraVai11,TCC:BDGM19}), and relies on the existence of a multi-key \emph{linearly}-homomorphic encryption scheme. This multi-key linearly homomorphic scheme is an extension of dual Regev encryption, and is implicit in the construction that follows.

Let $\mfhe$ be a classical multi-key fully-homomorphic encryption scheme with nearly linear decryption. Consider the following scheme $\qcmfhe$, which is identical to $\mfhe$ except that it has a different $\keygen$ algorithm and it additionally supports a $\convert$ algorithm. Let $q$ be an even $k$-bit modulus, let $\mathbf{g} = (1,2,\dots,2^k)$, and for $y \in \mathbb{Z}_q$, let $\mathbf{g}^{-1}(y) \in \{0,1\}^k$ be the binary expansion of $y$, i.e. it holds that $\mathbf{g}^\top \cdot \mathbf{g}^{-1}(y) = y$.

\begin{itemize}
    \item $\qcmfhe.\keygen(1^\secp)$: 
    \begin{enumerate}
        \item Compute $(\mfhe.\pk,\mfhe.\sk) \gets \mfhe.\keygen(1^\secp)$, where $\mfhe.\sk \in \bZ_q^\ell$ and  $\ell = \poly(\secp)$.
        \item Let $n,m$ be positive integers and $\chi$ be a $B$-bounded error distribution, where $n,m,B = \poly(\secp)$.
        \item Draw $(\bB,\tau) \leftarrow \gentrap(1^n,1^m,q)$, $\mathbf{b} \gets \bZ_q^n$, and set $\bA = \begin{pmatrix} \bB \\ \mathbf{b}^\top \end{pmatrix}$.
        \item Parse $\mfhe.\sk \in \bZ_q^\ell$ as $\mu_1,\dots,\mu_\ell \in \bZ_q$, and for each $i \in [\ell]$, compute the following.
        \begin{enumerate}
            \item Draw $\bS \leftarrow \mathbb{Z}_q^{n \times k}$ and $\bE \leftarrow \chi^{(m+1) \times k}$.
            \item Set $\bC_i \coloneqq \bA \cdot \bS + \bE + \mu_i \cdot \mathbf{g}^\top \cdot \bu_{m+1}$, where $\bu_{m+1}$ is the $(m+1)$-dimensional vector with all 0s except the final coordinate is 1.
        \end{enumerate}
        \item Output $\pk \coloneqq (\mfhe.\pk,\bC_1,\dots,\bC_\ell)$, $\sk \coloneqq \mfhe.\sk$, and $\tau$.
    \end{enumerate}
    \item $\qcmfhe.\convert((\pk_1,\dots,\pk_t),c)$: Let the linear function $\cL_c$ determined by $c$ consist of coefficients $a_{1,1},\dots,a_{1,\ell},\dots,a_{t,1},\dots,a_{t,\ell}$. Parse each $\pk_i$ to obtain $\bC_{i,1},\dots,\bC_{i,\ell}$ and define $\widehat{\bC}_{i,j}$ as follows. Let $\bar{\bC}_{i,j}$ be the first $m$ rows of $\bC_{i,j}$ and $\mathbf{c}_{i,j}$ be the last row. Then $\hat{\bC}_{i,j} \in \mathbb{Z}_q^{(m\ell+1) \times k}$ is the matrix with 0s everywhere except that the $(i-1)m+1,\dots,im$ rows are set to $\bar{\bC}_{i,j}$ and the last row is set to $\mathbf{c}_{i,j}$. Output $$\sum_{i \in [t],j \in [\ell]}\hat{\bC}_{i,j} \cdot \mathbf{g}^{-1}(a_{i,j}).$$
\end{itemize}

We assume that the parameters of $\mfhe$ are instantiated in a particular way, namely, the modulus $q$ is set such that for any well-formed ciphertext $c$ (encrypted under a set of $t$ public keys) that may arise during homomorphic evaluation, $\cL_c(\sk_1,\dots,\sk_t) = q/2 + e \bmod q$, where $q \geq \omega(\poly(\secp)) \cdot |e|$. Recall that the linear function $\cL_c$ is guaranteed to exist by the nearly linear decryption property.

\begin{theorem}
\label{thm:qmfhe}
Assuming the existence of a multi-key fully-homomorphic encryption scheme $\mfhe$ with nearly linear decryption and a particular circular security property, there exists a quantum multi-key fully-homomorphic encryption scheme $\qmfhe$. Moreover, $\qmfhe$ satisfies the following properties.

\begin{enumerate}
    \item The $\qmfhe.\setup$ algorithm is equivalent to $\mfhe.\setup$.
    \item If $\mfhe$ is perfectly correct, then $\qmfhe$ satisfies Classicality-Preserving Quantum Homomorphism.
\end{enumerate}
\end{theorem}

\begin{proof}

$\qmfhe$ is obtained by first applying the construction described in~\cref{subsec:keyswitch} to $\mfhe$ to obtain $\qcmfhe$, followed by the construction sketched in~\cref{subsec:background}. 


First, we argue that $\qcmfhe$ is indeed quantum-capable. Consider the output of the $\qcmfhe.\convert$ algorithm. It is straightforward to verify that if $c$ is a well-formed encryption under public keys $\pk_1,\dots,\pk_t$ of the bit $\mu$, where each $\pk_i$ may be parsed as $\begin{pmatrix} \bB_i \\ \bb_i^\top \end{pmatrix}$, then the resulting vector may be written as $$\begin{pmatrix} \bB_1 & & \\ & \ddots & \\ & & \bB_t \\ \bb_1^\top & \dots & \bb^\top_t \end{pmatrix} \cdot \bs^* + \be^* + \frac{q}{2}\begin{pmatrix} 0 \\ \vdots \\ 0 \\ \mu\end{pmatrix},$$ for some $\bs^*\in \bZ_q^{tn},\be^* \in \bZ^{(m+1)\ell}$. This is exactly an encryption of $\mu$ under the dual Regev scheme with public key $$\begin{pmatrix} \bB_1 & & \\ & \ddots & \\ & & \bB_t \\ \bb_1^\top & \dots & \bb^\top_t \end{pmatrix}.$$ Thus the $\ahe$ scheme we use in~\cref{def:quantum-capable} is identical to the scheme used in~\cite{FOCS:Mahadev18b}. This shows that $\qcmfhe$ satisfies the first requirement in~\cref{def:quantum-capable}, and the fact that it satisfies also the second requirement is immediate.

To confirm that $\qcmfhe$ satisfies the third requirement, we take a closer look at $\be^*$. The distribution $\mathcal{D}$ used by~\cite{FOCS:Mahadev18b} samples $\mu$ and $\bs$ uniformly at random, and $\be$ from a discrete Guassian distribution with ``large enough'' parameter $B'$. This requirement will hold if $B'$ is super-polynomially larger than the entries of $\be^*$ (see Lemma 3.3 and Section 5.3 of~\cite{FOCS:Mahadev18b} for more details). Note that the modulus $q$ is super-polynomially larger than each entry of $\be^*$, by the assumption on parameters of $\mfhe$. Indeed, all but the last entry are bounded by $\ell \cdot k \cdot B = \poly(\secp)$, and the last entry is bounded by $t \cdot \ell \cdot k \cdot B$ plus the error that results from the nearly linear decryption, which is super-polynomially smaller than $q$. This allows us to define $B'$ to be large enough such that the third requirement will hold.

To confirm that $\qcmfhe$ satisfies the fouth requirement, note that $\bs^*$ may be written as a concatenation of $t$ $n$-dimensional vectors, and that the $i$'th such vector may be recovered by using $\tau_{i}$, by~\cref{lemma:perfecttd}. This process also recovers all but the last entry of $\be^*$. The last entry of $\be^*$ may then be recovered by subtracting the public key times $\bs^*$ and rounding the last element of the resulting vector. 

The above shows that $\qcmfhe$ is quantum-capable according to~\cref{def:quantum-capable}. Next, we discuss security of the scheme $\qmfhe$ obtained by applying the construction sketched in~\cref{subsec:background}. Observe that $\qmfhe.\keygen$ outputs a public key that contains a $\mfhe$ public key, a dual Regev public key, an encryption of the $\mfhe$ secret key under the dual Regev public key, and an encryption of the dual Regev \emph{trapdoor} under the $\mfhe$ secret key. Since $\qmfhe$ encryption involves encrypting a QOTP key under $\mfhe$ and using that key to \emph{perfectly} hide the message, it follows that security of $\qmfhe$ reduces to the security of $\mfhe$ in the presence of the particular two-cycle of keys described above (which at the very least relies on LWE to ensure security of dual Regev). Thus, as stated in the theorem, security follows from a particular circular security property of $\mfhe$.\footnote{This property is similar to the one needed by~\cite{FOCS:Mahadev18b} in the single-key setting, in the sense that encryption of a dual Regev trapdoor is part of the circular security requirement.}




It remains to argue that the two extra properties promised by the theorem statement hold. First, note that the constructions given in~\cref{subsec:keyswitch} and~\cref{subsec:background} do not alter any $\mfhe.\setup$ algorithm that may exist. Next, the second property boils down to showing that for \emph{every} choice of random coins used in $\qmfhe.\keygen$, homomorphic evaluation of quantum (or classical) circuits will be statistically correct. Perfect correctness of any Clifford operation follows directly from perfect correctness of $\mfhe$. Statistical correctness of the encrypted CNOT operation follows from properties 2 and 3 of~\cref{def:quantum-capable} (this analysis can be found in~\cite{FOCS:Mahadev18b}), plus the \emph{perfect} correctness of property 4, which is ensured by using the variant of the $\gentrap$ algorithm promised by~\cref{lemma:perfecttd}. 
\end{proof}



\subsection{Quantum Spooky Encryption}
\label{subsec:spooky}

We define the notion of spooky encryption for (classical) relations computable by quantum circuits, generalizing the purely classical notion from~\cite{C:DHRW16}. In favor of a simpler exposition we present the additive function sharing (AFS) variant of the notion, but we note that considering more general relations
(in the same spirit as~\cite{C:DHRW16}) is also possible.

\begin{definition}[Quantum AFS-Spooky Encryption]
A quantum AFS-spooky encryption scheme is given by six algorithms ($\spooky.\keygen$, $\spooky.\enc$, $\spooky.\qenc$, $\spooky.\eval$, $\spooky.\dec$, $\spooky.\qdec$) with the same syntax as the corresponding $\qmfhe$ algorithms defined in~\cref{def:qmfhe}, except for the following differences.


\begin{itemize}
    \item $b \gets \qmfhe.\dec(\sk,c)$: A PPT algorithm that takes as input a secret key and a classical ciphertext $c$ and outputs a bit. (This algorithm takes only one secret key, as opposed to $n$ secret keys in $\qmfhe$.)
    \item $\widehat{c}_1, \dots, \widehat{c}_n,\ket{\widehat{\phi}} \gets \spooky.\eval((\pk_1,\dots,\pk_n),C,(\ket{\phi_1},\dots,\ket{\phi_n}))$: A QPT algorithm that takes as input 
    \begin{enumerate}
        \item A set of $n$ public keys.
        \item A general quantum circuit with $\ell_1 + \dots + \ell_n$ input qubits and $\ell'$ output qubits, out of which $m$ are measured.
        \item A set of $n$ ciphertexts where $\ket{\phi_i}$ encrypts an $\ell_i$-qubit state under $\pk_i$. Some of the $\ell_i$ ciphertexts are possibly classical ciphertexts (generated by the classical encryption algorithm) encrypting classical bits.
    \end{enumerate}
    The evaluation algorithm outputs $n$ classical ciphertexts $(\widehat{c}_1, \dots, \widehat{c}_n)$ each encrypting $m$ bits under the corresponding $\pk_i$, plus a quantum ciphertext $\ket{\widehat{\phi}}$ encrypting an $(\ell'-m)$-qubit quantum state (under keys $\pk_1,\dots,\pk_n$). (This algorithm outputs $n$ classical ciphertexts, as opposed to one in $\qmfhe$.)
\end{itemize}
The scheme satisfies the same properties of quantum semantic security and compactness as defined in~\cref{def:qmfhe}. In the following we present the notion of correctness for quantum AFS-spooky encryption.

\begin{itemize}
    \item {\bf Correctness of Spooky Evaluation: } Let $C = \{C_\secp\}_{\secp \in \mathbb{N}}$ be a polynomial-size quantum circuit, where $C_\secp$ has $\ell_1(\secp) + \dots + \ell_n(\secp)$ input qubits and $\ell'(\secp)$ output qubits, of which $m(\secp)$ are measured. Let $\ket{\phi_1},\dots,\ket{\phi_n} = \{\ket{\phi_1}_\secp,\dots,\ket{\phi_n}_\secp\}_{\secp \in \mathbb{N}}$ be an input state for $C$, let $(\pk_1,\sk_1),\dots,(\pk_n,\sk_n) = \{(\pk_1,\sk_1)_\secp,\dots,(\pk_n,\sk_n)_\secp\}_{\secp \in \mathbb{N}}$ be pairs of public and secret keys ($\forall i \in [n], \secp \in \mathbb{N}, (\pk_i,\sk_i)_\secp \in \spooky.\keygen(1^\secp)$) and let $r_1,\dots,r_n = \{(r_1)_\secp,\dots,(r_n)_\secp\}_{\secp \in \mathbb{N}}$ be $n$ random strings for the encryption algorithm. Then there exists a negligible function $\mu(\cdot)$ such that for all $\secp \in \mathbb{N}$, $$\mathsf{TD}(\rho_{0,\secp},\rho_{1,\secp}) \leq \mu(\secp),$$ where $\rho_0,\rho_1$ are quantum distributions defined as follows:
    \begin{itemize}
        \item $\rho_{0,\secp}$: For each $i \in [n]$, encrypt each classical bit of $\ket{\phi_i}$ with $\spooky.\enc(\pk_i,\cdot)$ and the rest with $\spooky.\qenc(\pk_i,\cdot)$ (using randomness $r_i$). Execute $\spooky.\eval((\pk_1,\dots,\pk_n),C,\cdot)$ on the $n$ encryptions to get $(\widehat{c}_1, \dots, \widehat{c}_n),\ket{\widehat{\phi}}$, where $(\widehat{c}_1, \dots, \widehat{c}_n)$ are classical ciphertexts each encrypting $m(\secp)$ bits. Then output $$\mathop{\bigoplus}\limits_{i=1}^n\spooky.\dec(\sk_i,\widehat{c}_i),\spooky.\qdec((\sk_1,\dots,\sk_n),\ket{\widehat{\phi}}).$$
        \item $\rho_{1,\secp}$: Output $C(\ket{\phi_1,\dots,\phi_n})$.
    \end{itemize}
\end{itemize}
\end{definition}

\paragraph{(Classical) AFS-Spooky Encryption with Distributed Setup.} As a stepping stone towards the main result of this section, we show how to construct spooky encryption for \emph{classical} relations with a distributed setup. More precisely, assuming the hardness of the LWE problem, we show an instantiation of spooky encryption for any polynomial-size (classical) circuit where the parties jointly compute the public parameters of the system via a local algorithm $\spooky.\setup$ (with the same syntax as \cref{def:distributed-setup}). This stands in contrast with the scheme of~\cite{C:DHRW16}, where the common reference string is assumed to be sampled by a trusted party. Before describing the construction, we recall a useful lemma (rephrased) from~\cite{C:DHRW16}.

\begin{lemma}[\cite{C:DHRW16}]
Let $\spooky$ be an AFS-spooky encryption scheme that supports (i) single key additive homomorphism and (ii) two-key spooky multiplication. Then the same scheme supports the AFS-spooky evaluation of all polynomial-size (classical) circuits.
\end{lemma}

It follows that it suffices to construct a spooky encryption that supports a single multiplication over an arbitrary pair of keys. We do this by showing that the scheme from~\cite{TCC:BraHalPol17} supports two-key spooky multiplication. This follows from the fact that the decryption circuit is identical to that of \cite{C:CleMcG15,EC:MukWic16}, which was shown to support two-key spooky multiplication in~\cite{C:DHRW16}. For completness, we recall the modified algorithms in the following.

\begin{itemize}
    \item $\spooky.\enc(\pk, m)$: Same as $\mfhe.\enc$ but append an extra $\delta = 0$ to the resulting ciphertext.
    \item $\spooky.\dec(\sk, c)$: Let $\cL_c$ be the linear function defined by $c$, compute $$v = \cL_c(\sk) +\delta \mod{q}$$ and return $0$ if $|v|< q/4$ and $1$ otherwise.
    \item $ \spooky.\eval((\pk_1,\pk_2),C,(c_1, c_2))$: Compute $$\widehat{c} \gets \mfhe.\eval\left((\pk_1, \pk_2), \prod, (c_1, c_2)\right)$$ and let $\cL^{(1)}_{\widehat{c}}$ and $\cL^{(2)}_{\widehat{c}}$ be the linear functions defined by the resulting $\widehat{c}$. Sample a uniform $\delta$ from $\mathbb{Z}_q$ and return $(\cL^{(1)}_{\widehat{c}}, \delta)$ and $(\cL^{(2)}_{\widehat{c}}, -\delta)$.
\end{itemize}
As discussed above, the scheme is quantum semantically secure assuming the hardness of the LWE problem. Correctness follows, for the same choice of parameters of~\cite{TCC:BraHalPol17}, by an invocation of the following lemma.
\begin{lemma}[\cite{C:DHRW16}]
Fix a modulus $q \in \mathbb{Z}$, a bit $b\in \{0,1\}$ and a value $v \in \mathbb{Z}_q$ such that $v = q/2 \cdot b + e \mod{q}$, for some $|e| < q/4$. Sample $v_1$ and $v_2$ uniformly at random from $\mathbb{Z}_q$ constrained on the fact that $v_1 + v_2 = v \mod{q}$, and let $b_i = 0$ if $|v_i| < q/4$ and $b_i = 1$ otherwise. Then
$$\Pr[b_1 \oplus b_2 = b] > 1 - 2(|e|+1)/q$$
over the random choice of $v_1$ and $v_2$.
\end{lemma}

\paragraph{Quantum AFS-Spooky Encryption with Distributed Setup.} Finally, we show how to combine a classical AFS-spooky encryption scheme $\spooky$ (with distributed setup) with a quantum multi-key fully-homomorphic encryption $\qmfhe$ (with distributed setup) to obtain a quantum AFS-spooky encryption scheme $\qspooky$. Since both of the building blocks have a distributed setup, then so does the resulting encryption scheme. The scheme is described below.
\begin{itemize}
    \item $\qspooky.\setup(1^\secp)$: Compute $\pp \gets \qmfhe.\setup(1^\secp)$ and $\tilde{\pp} \gets \spooky.\setup(1^\secp)$ and return $(\pp, \tilde{\pp})$.
    \item $\qspooky.\keygen(1^\secp,\pp)$: Sample $$(\pk', \sk') \gets \qmfhe.\keygen(1^\secp, \pp)\text{ and }(\tilde{\pk}, \tilde{\sk}) \gets \spooky.\keygen(1^\secp, \tilde{\pp})$$ and compute $\tilde{c} \gets \spooky.\enc(\tilde{\pk}, \sk')$. Return $\pk \coloneqq (\pk', \tilde{\pk}, \tilde{c})$ as the public key and $\sk \coloneqq (\sk', \tilde{\sk})$ as the secret key.  
    \item $\qspooky.\enc(\pk, m)$: Return $\qmfhe.\enc(\pk', m)$.
    \item $\qspooky.\qenc(\pk, \ket{\psi})$: Return $\qmfhe.\qenc(\pk', \ket{\psi})$.
    \item $ \qspooky.\eval((\pk_1,\dots, \pk_n),C,(\ket{\phi_1}, \dots, \ket{\phi_n}))$: Compute $$(\widehat{c},\ket{\widehat{\phi}}) \gets \qmfhe.\eval((\pk'_1,\dots,\pk'_n),C,(\ket{\phi_1},\dots,\ket{\phi_n}))$$ and let $(\tilde{c}_1, \dots, \tilde{c}_n)$ be the corresponding element of each public key. Compute $$(\widehat{c}_1, \dots, \widehat{c}_n)\gets \spooky.\eval((\tilde{\pk}_1,\dots, \tilde{\pk}_n),\qmfhe.\dec(\cdot, \widehat{c}),(\tilde{c}_1, \dots, \tilde{c}_n))$$ and return $(\widehat{c}_1, \dots, \widehat{c}_n, \ket{\phi})$.
    \item $\qspooky.\dec(\sk, c)$: Return $\spooky.\dec(\tilde{\sk}, c)$.
    \item $\qspooky.\qdec((\sk_1,\dots,\sk_n),\ket{\phi})$: Return $\qmfhe.\qdec((\sk'_1,\dots,\sk'_n),\ket{\phi})$.
\end{itemize}

The following theorem establishes our claim.

\begin{theorem}
\label{thm:qspooky}
Assuming that $\qmfhe$ is a quantum multi-key fully-homomorhic encryption scheme and that $\spooky$ is a classical AFS-spooky encryption, $\qspooky$ is a quantum AFS-spooky encryption scheme.\footnote{In fact, we also need the quantum spooky encryption scheme to be \emph{multi-hop}, which follows if the classical AFS-spooky scheme is multi-hop (which is satisfied by~\cite{C:DHRW16}).}
\end{theorem}

\begin{proof}
Assuming quantum semantic security of $\spooky$, the changes in the key generation algorithm do not affect the security of the scheme. Then quantum semantic security follows from an invocation of the quantum semantic security of $\qmfhe$. Correctness of spooky evaluation follows from the classicality-preserving homomorphism of $\qmfhe$ and from the correctness of classical spooky evaluation of $\spooky$.
\end{proof}
%
\section{Quantum-Secure Multi-Committer Extractable Commitment}
\label{sec:pecom}

In this section, we follow the outline presented in~\cref{sec:over-pzk} to construct a commitment scheme that allows for simultaneous extraction from multiple parallel committers. The protocol is somewhat more involved than the high-level description given earlier, so we briefly highlight the differences. 

First, the committer is instructed to (non-interactively) commit to its message and trapdoor at the very beginning of the protocol. We use these commitments to take advantage of non-uniformity in the reductions betwen hybrids in the extractability proof. In particular, hybrids that come before the step where the simulator goes ``under the hood'' of the FHE may still need access to the trapdoor and commitment, and this can be given to any reduction via non-uniform advice consisting of each committer's first message and corresponding openings.

Next, the CDS described earlier is replaced with a function-hiding secure function evaluation (SFE) protocol. In order to rule out the malleability attack mentioned in~\cref{sec:over-pzk}, where a malicious receiver mauls the AFS-spooky encryption of the committer's trapdoor into an SFE encryption of the trapdoor, we do the following. The first message sent by the receiver to each committer $\eC_i$ will actually be a commitment to some key $k_i$ of a generic secret-key encryption scheme. After $\eC_i$ sends its AFS-spooky encryption ciphertext and compute and compare obfuscation, the receiver prepares and sends a secret-key encryption of an arbitrary message. Then, the receiver's input to the SFE consists of the opening to its earlier commitment $k_i$, and the SFE checks if the secret-key encryption sent by the receiver is actually an encryption of the committer's trapdoor under secret key $k_i$. If so, it returns the lock and otherwise it returns $\bot$. This setup ensures that a malicious receiver cannot maul the AFS-spooky encryption of the committer's trapdoor, for the following reason. If it could, then a non-uniform reduction to the semantic security of AFS-spooky encryption may obtain the receiver's committed $k_i$ as advice and decrypt the receiver's secret-key encryption to obtain the trapdoor. Of course, this assumes the receiver actually acted explainably in sending a valid commitment at the beginning of the protocol, and this is ensured by the opening check performed under the SFE. We note that this mechanism is somewhat different than what was presented in~\cite{BS20}, as they directly build a zero-knowledge argument (i.e. without first constructing a stand-alone extractable commitment) and are able to take advantage of witness indistinguishability to enforce explainable behavior.

\paragraph{Compliant Distinguishers.} Finally, we discuss the issue of \emph{committer} explainability. Recall from the high-level overview that a simulator is able to extract from a committer by homomorphically evaluating its code on an AFS-spooky encryption ciphertext \emph{generated by the committer}. Thus, if the committer acts arbitrarily maliciously and does not return a well-formed ciphertext, the extraction may completely fail. Again,~\cite{BS20} address this issue by only analyzing their commitment within the context of a larger zero-knowledge argument protocol, and having the verifier prove to the prover using a witness indistinguishable proof that it performed the commitment explainably. 

Thus, without adding zero-knowledge and performing~\cite{JC:GolKah96}-style analysis to handle non-explainable and aborting committers, we will only obtain extractability against explainable committers. However, since we will be using this protocol inside larger protocols where participants are not assumed to be acting explainably, restricting the class of committers we consider in our definition is problematic. We instead consider arbitrary committers but restrict the class of \emph{distinguishers} (who are supposed to decide whether they received the view of a committer interacting in the real protocol or the view of a committer interacting with the extractor) to those that always output 0 on input a non-explainable transcript. In other words, any advantage these distinguishers may have must be coming from their behavior on input explainable views. Even though checking whether a particular view is explainable or not is not efficient, it turns out that this definition lends itself quite nicely to composition, since one can use witness indistinguishability/zero-knowledge to construct provably compliant distinguishers between hybrids for the larger protocols.

For completeness, and because post-quantum multi-committer extractable commitments may be of independent interest, we also show in~\cref{sec:fullecom} how to add zero-knowledge within the extractable commitment protocol itself to obtain security against arbitrary committers.


\subsection{Definition}
\begin{definition}[Quantum-Secure Multi-Committer Extractable Commitment] \label{defn:qspec}
A quantum-secure multi-committer extractable commitment scheme is a pair $(\eC,\eR)$ of classical PPT interactive Turing machines. In the commit phase, $\eR$ interacts with $n$ copies $\{\eC_i\}_{i \in [n]}$ of $\eC$ (who do not interact with each other) on common input $1^\secp$ and $1^n$, with each $\eC_i$ additionally taking a private input $m_i \in \{0,1\}^*$. This produces a transcript $\tau$, which may be parsed as a set of $n$ transcripts $\{\tau_i\}_{i \in [n]}$, one for each set of messages exchanged between $\eR$ and $\eC_i$. In the decommitment phase, each $\eC_i$ outputs $m_i$ along with its random coins $r_i$, and $\eR$ on input $(1^\secp,\tau_i,m_i,r_i)$ either accepts or rejects. The scheme should satisfy the following properties.
\begin{itemize}
    \item {\bf Perfect Correctness: } For any $\secp,n \in \mathbb{N}, i \in [n]$, $$\Pr[\eR(1^\secp,\tau_i,m_i,r_i) = 1 \ | \ \{\tau_i\}_{i \in [n]} \gets \dist{\eR,\eC_1(m_1;r_1),\dots,\eC_n(m_n;r_n)}(1^\secp,1^n)] = 1.$$
    \item {\bf Perfect Binding: } For any $\secp \in \mathbb{N}$ and string $\tau \in \{0,1\}^*$, there does not exist $(m,r)$ and $(m',r')$ with $m \neq m'$ such that $\eR(1^\secp,\tau,m,r) = \eR(1^\secp,\tau,m',r')=1$.
    \item {\bf Quantum Computational Hiding: } For any non-uniform quantum polynomial-size receiver $\eR^* = \{\eR^*_\secp,\rho_\secp\}_{\secp \in \mathbb{N}}$, any polynomial $\ell(\cdot)$, and any sequence of sets of strings $\{m^{(0)}_{\secp,1},\dots,m^{(0)}_{\secp,n}\}_{\secp,n \in \mathbb{N}}$, $\{m^{(1)}_{\secp,1},\dots,m^{(1)}_{\secp,n}\}_{\secp,n \in \mathbb{N}}$ where each $|m^{(b)}_{\secp,i}| = \ell(\secp)$,
    \begin{align*}
        &\{\view_{\eR^*_\secp}(\dist{\eR^*_\secp(\rho_\secp),\eC_1(m^{(0)}_{\secp,1}),\dots,\eC_{n}(m^{(0)}_{\secp,n})}(1^\secp,1^n))\}_{\secp,n \in \mathbb{N}} \\ \approx_c  &\{\view_{\eR^*_\secp}(\dist{\eR^*_\secp(\rho_\secp),\eC_1(m^{(1)}_{\secp,1}),\dots,\eC_{n}(m^{(1)}_{\secp,n})}(1^\secp,1^n))\}_{\secp,n \in \mathbb{N}}.
    \end{align*}
    \end{itemize}
   The extractability property will require the following two definitions. First, for any adversary $\eC^* = \{\eC^*_\secp,\rho_\secp\}_{\secp \in \mathbb{N}}$ representing a subset $I \subseteq [n]$ of $n$ committers, any honest party messages $\{m_i\}_{i \notin I}$, and any security parameter $\secp \in \mathbb{N}$, define $\mathsf{VIEW}^{\mathsf{msg}}_{\eC^*_{\secp}}(\dist{\eR,\eC^*_\secp(\rho_\secp),\{\eC_i(m_i)\}_{i \notin I}}(1^\secp,1^n))$ to consist of the following.
    \begin{enumerate}
        \item The view of $\eC^*_\secp$ on interaction with the honest receiver $\eR$ and set $\{\eC_i(m_i)\}_{i \notin I}$ of honest parties; this view includes a set of transcripts $\{\tau_i\}_{i \in I}$ and a state $\state$.
        \item A set of strings $\{m_i\}_{i \in I}$, where each $m_i$ is defined relative to $\tau_i$ as follows. If there exists $m'_i,r_i$ such that $\eR(1^\secp,\tau_i,m'_i,r_i) = 1$, then $m_i = m'_i$, otherwise, $m_i = \bot$.
    \end{enumerate}
    Next, we consider distinguishers $\D = \{\D_\secp,\sigma_\secp\}_{\secp \in \bN}$ that take as input a sample $(\{\tau_i\}_{i \in I},\state,\{m_i\}_{i \in I})$ from the distribution just described. We say that $\D$ is \emph{compliant} if whenever $\{\tau_i\}_{i \in I}$ is not an explainable transcript with respect to the set $I$, $\D$ outputs 0 with overwhelming probability (over the randomness of $\D$).
    \begin{itemize}
    \item {\bf Multi-Committer Extractability:} There exists a quantum expected-polynomial-time extractor $\cE$ such that for any \emph{compliant} non-uniform polynomial-size quantum distinguisher $\D = \{\D_\secp,\sigma_\secp\}_{\secp \in \bN}$, there exists a negligible function $\mu(\cdot)$, such that for all adversaries $\eC^* = \{\eC^*_\secp,\rho_\secp\}_{\secp \in \mathbb{N}}$ representing a subset of $n$ committers, namely, $\{\eC_i\}_{i \in I}$ for some set $I \subseteq [n]$, the following holds for all polynomial-size sequences of inputs $\{\{m_{i,\secp}\}_{i \notin I}\}_{\secp \in \mathbb{N}}$ and $\secp \in \bN$.
    \begin{align*}&\big|\Pr[\D_\secp(\mathsf{VIEW}^{\mathsf{msg}}_{\eC^*_\secp}(\dist{\eR,\eC^*_\secp(\rho_\secp),\{\eC_i(m_{i,\secp})\}_{i \notin I}}(1^\secp,1^n)),\sigma_\secp) = 1]\\ &- \Pr[\D_\secp( \cE(1^\secp,1^n,I,\eC^*_\secp,\rho_\secp),\sigma_\secp)=1]\big| \leq \mu(\secp).
    \end{align*}

\end{itemize}
\end{definition}

\begin{remark}
Observe that the above definition of quantum computational hiding does not consider potentially malicious committers that interact in the protocol to try to gain information about commitments made by other committers. This is without loss of generality, since all communication occurs between $\eR$ and some $\eC_i$. In particular, no messages are sent between any $\eC_i$ and $\eC_j$. 
\end{remark}

\subsection{Construction}

\paragraph{Ingredients:} All of the following are assumed to be quantum-secure.
\begin{itemize}
    \item A non-interactive perfectly-binding commitment $\COM$.
    \item A secret-key encryption scheme $(\enc,\dec)$.\footnote{We use the syntax that for key $k$, a ciphertext of message $m$ is computed as $\ct \gets \enc(k,m)$ and decrypted as $m \coloneqq \dec(k,\ct)$.}
    \item A compute-and-compare obfuscator $\obf$.
    \item A quantum AFS-spooky encryption scheme with distributed setup ($\spooky.\setup$,$\spooky.\keygen$, $\spooky.\enc$, $\spooky.\qenc$, $\spooky.\eval$, $\spooky.\dec$, $\spooky.\qdec$).
    \item A two-message function-hiding secure function evaluation scheme $(\sfegen, \sfeenc, \sfeeval, \sfedec)$.
\end{itemize}


\protocol
{\proref{fig:ext_com}}
{A constant-round quantum-secure multi-committer extractable commitment.}
{fig:ext_com}
{
\begin{description}
    \item[Common input:] $1^\secp,1^n$. 
    \item[$\eC_i$'s additional input:] 
    A string $m_i$.
\end{description}
\begin{enumerate}
    \item Each $\eC_i$ computes $\td_i \gets U_\secp$ and sends $\cm_i^{(\msg)} \gets \COM(1^\secp,m_i)$, $\cm_i^{(\td)} \gets \COM(1^\secp,\td_i)$ to $\eR$. 
    \item For each $i \in [n]$, $\eR$ computes $k_i,r_i \gets U_\secp$ and sends $\cm_i^{(\key)} \coloneqq \COM(1^\secp,k_i;r_i)$ to $\eC_i$.
    \item Each $\eC_i$ computes and sends $\pp_i \gets \spooky.\setup(1^\secp)$ to $\eR$.
    \item $\eR$ defines $\pp \coloneqq \{\pp_i\}_{i \in [n]}$, and sends $\pp$ to each $\eC_i$. Each $\eC_i$ checks that the $\pp_i$ it received matches the $\pp_i$ it sent in Step 3, and if not, it aborts.
    \item Each $\eC_i$ computes 
    \begin{itemize}
        \item $\lock_i \gets U_\secp$,
        \item $(\pk_i,\sk_i) \gets \spooky.\keygen(1^\secp,\pp)$,
        \item $\ct_i \gets \spooky.\enc(\pk_i,\td_i)$,
        \item and $\obfC_i \gets \obf\left(\CC{\spooky.\dec(\sk_i, \cdot)}{\lock_i}{(\sk_i,m_i)}\right)$,
    \end{itemize}
    and sends $(\pk_i,\ct_i,\obfC_i)$ to $\eR$.
    
    \item For each $i \in [n]$, $\eR$ computes $\ct_i^{(\td)} \gets \enc(k_i,0^\secp)$, $\dk_i \gets \sfegen(1^\lambda)$, and $\ct_i^{(\sfe)} \gets \sfeenc(\dk_i, (k_i,r_i))$ and sends $(\ct_i^{(\td)},\ct_i^{(\sfe)})$ to $\eC_i$.

    \item Define the circuit $\mathsf{C}[\cm_i^{(\key)},\ct_i^{(\td)},\td_i,\lock_i](\cdot)$ to take as input $(k_i,r_i)$, check if $\cm_i^{(\key)}$ opens to $k_i$ with opening $r_i$ and if $\td_i = \dec(k_i,\cm_i^{(\td)})$, and if so output $\lock_i$, and otherwise output $\bot$. Each $\eC_i$ computes and sends $\widehat{\ct}_i^{(\sfe)} \gets \sfe.\eval(\mathsf{C}[\cm_i^{(\key)},\ct_i^{(\td)},\td_i,\lock_i],\ct_i^{(\sfe)})$.
\end{enumerate}
}

\subsection{Hiding}

    
Perfect correctness and perfect binding are immediate, so we move to quantum computational hiding. 

\begin{lemma}
\proref{fig:ext_com} is quantum computational hiding.
\end{lemma}

\begin{proof}

Fix any non-uniform quantum polynomial-size receiver $\eR^* = \{\eR^*_\secp,\rho_\secp\}_{\secp \in \mathbb{N}}$, a polynomial $\ell(\cdot)$, and two sequences of sets $\{m^{(0)}_{\secp,1},\dots,m^{(0)}_{\secp,n}\}_{\secp,n \in \mathbb{N}}$, $\{m^{(1)}_{\secp,1},\dots,m^{(1)}_{\secp,n}\}_{\secp,n \in \mathbb{N}}$ where each $|m^{(b)}_{\secp,i}| = \ell(\secp)$. Consider the following sequence of hybrids for each $i \in [n]$, where each alters the view of $\eR^*$ in its interaction with $\eC_i$. The lemma follows immediately once we show that for all $i \in [n]$, $\hyb_{i,0} \approx_c \hyb_{i,6}$.

\begin{itemize}
    \item $\hyb_{i,0}$: $\{\view_{\cR^*_\secp}(\dist{\eR^*_\secp,\eC_i(m^{(0)}_{\secp,i})})(1^\secp,1^n)\}_{\secp \in \mathbb{N}}$.
    \item $\hyb_{i,1}$: Same as $\hyb_{i,0}$ except that in Step 1, $\cm^{(\msg)}_i$ and $\cm^{(\td)}_i$ are commitments to $0$.
    \item $\hyb_{i,2}$: Same as $\hyb_{i,1}$ except that in Step 7, $\eC_i$ computes $\sfe.\eval$ on the circuit $\mathsf{C}_\bot$ that always outputs $\bot$. 
    \item $\hyb_{i,3}$: Same as $\hyb_{i,2}$ except that in Step 5, the compute-and-compare obfuscation is simulated: $\obfC\gets \mathsf{Sim}^{\mathsf{CC}}(1^{|\spooky.\dec(\sk_i, \cdot)|},1^{|\sk_i| + |\ell(\secp)|},1^\secp)$.
    \item $\hyb_{i,4}$: Same as $\hyb_{i,3}$ except that in Step 5, the compute-and-compare obfuscation is performed honestly with respect to message $m^{(1)}_{\secp,i}$.
    \item $\hyb_{i,5}$: Same as $\hyb_{i,4}$ except that in Step 7, the $\sfe.\eval$ is performed honestly.
    \item $\hyb_{i,6}$: Same as $\hyb_{i,5}$ except that in Step 1, $\cm_i^{(\msg)}$ is a commitment to $m^{(1)}_{\secp,i}$ and $\cm_i^{(\td)}$ is a commitment to $\td_i$. Note that is this exactly $\{\view_{\cR^*_\secp}(\dist{\eR^*_\secp,\eC_i(m^{(1)}_{\secp,i})})(1^\secp,1^n)\}_{\secp \in \mathbb{N}}$.
\end{itemize}

Now we argue indistinguishability between each hybrid.

\begin{itemize}
    \item $\hyb_{i,0} \approx_c \hyb_{i,1}$: This follows directly from the quantum computational hiding of $\COM$.
    \item $\hyb_{i,1} \approx_s \hyb_{i,2}$: We consider two cases. First, conditioned on $\eC_i$ aborting in Step 4, the hybrids are trivially indistinguishable. Next, conditioned on $\eC_i$ not aborting in Step 4, we show below that with overwhelming probability (over the randomness of $\eC_i$ and $\eR$), the circuit $\mathsf{C}[\cm_i^{(\key)},\ct_i^{(\td)},\td_i,\lock_i]$ is functionally equivalent to $\mathsf{C}_\bot$. Given this, the indistinguishability of hybrids $\hyb_{i,1}$ and $\hyb_{i,2}$ follows directly from the circuit privacy of $\sfe$.
    
    Assuming that the circuits are not functionally equivalent with noticeable probability, we construct a non-uniform $\A = \{\A_\secp,\rho^{\A}_\secp\}_{\secp \in \mathbb{N}}$ that breaks the distributed-setup quantum semantic security of $\spooky$ (see~\cref{def:distributed-setup}). In the security game, $\A$ interacts with a challenger to generate $\pp \coloneqq \{\pp_i\}_{i \in [n]}$ for $n$ parties. Then, the challenger draws a public key $\pk_i$ based on $\pp$, a random $\td_i \gets U_\secp$, and outputs an encryption $\ct_i$ of $\td_i$ under $\pk_i$. $\A$ wins if it returns $\td_i$, which would clearly break semantic security.
    
    Now, we describe the distribution $\rho^{\A}_\secp$ that $\A$ receives as non-uniform advice (this distribution will ultimately be fixed to the advice state that gives $\A$ the best advantage). It will be generated as follows. 
    \begin{enumerate}
        \item Run $\eR^*_\secp$ on $\rho_\secp$, and feed to $\eR^*_\secp$ the first messages $\{\cm_i^{(\msg)},\cm_i^{(\td)}\}_{i \in [n]}$ it expects from $\{\eC_i\}_{i \in [n]}$ (which are commitments to 0).
        \item Continue running $\eR^*_\secp$ until it outputs its set of messages $\{\cm_i^{(\key)}\}_{i \in [n]}$. 
        \item Output the inner state of $\eR^*_\secp$, the messages exchanged so far, and the following. For $\cm_i^{(\key)}$, check (inefficiently) if it is a commitment to some $k_i$ and if so, output $k_i$.
    \end{enumerate}
    
    Finally, we describe $\A$. $\A_\secp$ receives from its challenger the $i$'th public parameters $\pp_i$. It then runs $\eR^*_\secp$ on the state it received as advice and $\pp_i$. $\eR^*_\secp$ returns $\pp = \{\pp_j\}_{j \in [n]}$, where by assumption $\pp$ includes the same $\pp_i$ that it took as input. $\A_\secp$ then forwards $\{\pp_j\}_{j \in [n] \setminus \{i\}}$ to its challenger, who returns with a public key $\pk_i$ and a ciphertext $\ct_i$. At this point, $\A_\secp$ generates $\{(\pk_j,\ct_j,\obfC_j)\}_{j \in [n] \setminus \{i\}}$ honestly and for party $i$, fixes $(\pk_i,\ct_i)$ along with $\obfC_i \gets \mathsf{Sim}^{\mathsf{CC}}(1^{|\spooky.\dec(\sk_i, \cdot)|},1^{|\sk_i| + |\ell(\secp)|},1^\secp)$. It then continues to run $\eR^*_\secp$ on input all of these tuples.
    
    
    When $\eR^*_\secp$ returns $\ct_i^{(\td)}$, $\A_\secp$ checks if it received some $k_i$ as part of its non-uniform advice, and if so, it decrypts $\ct_i^{(\td)}$ using key $k_i$ to recover a message $\td_i$. It returns $\td_i$ to the challenger, who then determines if $\A_\secp$ succeeded. 
    
    Note that, by the simulation security of compute-and-compare obfuscation, the probability that $\A_\secp$ succeeds in this game is negligibly close to the probability it succeeds if it gave $\eR^*_\secp$ an honest compute-and-compare obfuscation $\obfC_i$. This follows because the lock value $\lock_i$ is completely independent of $\eR^*_\secp$'s view through Step 6. Finally, the probability that $\A_\secp$ succeeds in returning $\td$ is at least the probability that $\cm_i^{(\key)}$ is a well-formed commitment to $k_i$ and $\ct_i^{(\td)}$ is an encryption of $\td$ under key $k_i$, which is exactly the probability that the circuits described above are \emph{not} functionally equivalent. Thus $\A_\secp$ has non-negligible advantage in this game, a contradiction.
    \item $\hyb_{i,2} \approx_c \hyb_{i,3}$: This follows directly from the simulation security of compute-and-compare obfuscation, since at this point, the lock value $\lock_i$ is independent of the rest of the distribution. 
    \item $\hyb_{i,3} \approx_c \hyb_{i,4}$: Same argument as $\hyb_{i,2} \approx_c \hyb_{i,3}$.
    \item $\hyb_{i,4} \approx_c \hyb_{i,5}$: Same argument as $\hyb_{i,1} \approx_c \hyb_{i,2}$.
    \item $\hyb_{i,5} \approx_c \hyb_{i,6}$: Same argument as $\hyb_{i,0} \approx_c \hyb_{i,1}$.
\end{itemize}

\end{proof}

\subsection{Extractability}

\begin{lemma}
\proref{fig:ext_com} is multi-committer extractable.
\end{lemma}

\begin{proof}

In the following we describe the extractor. For notational convenience we assume that the set of corrupted parties $I$ is of size $|I| = \ell$ and we assume without loss of generality that $I = [\ell]$.

\medskip

\noindent $\cE(1^\secp,1^n,I,\eC^*_\secp,\rho_\secp)$:
\begin{enumerate}
    \item Set $\rho_\secp$ to be the inner state of $\eC^*_\secp$. Begin running $\eC^*_\secp$ until it outputs $\{\cm_i^{(\msg)},\cm_i^{(\td)}\}_{i \in I}$.
    \item Compute the commitments $\{\cm_i^{(\key)}\}_{i\in I}$ as specified in the protocol.
    \item Receive $\{\pp_i\}_{i \in I}$ from $\eC^*_\secp$, draw $\pp_i \gets \spooky.\setup(1^\secp)$ for each $i \notin I$, and send $|I|$ copies of $\pp \coloneqq \{\pp_i\}_{i \in [n]}$ to $\eC^*_\secp$.
    \item Receive $\{\pk_i,\ct_i,\obfC_i\}_{i \in I}$.
    \item Let $\ket{\phi}$ be the inner state of $\eC^*_\secp$ at this point, compute $\ct_{\ket{\phi}} \gets \spooky.\qenc(\pk_{1},\ket{\phi})$.
    \item For each $i \in I$, compute $\tilde{\ct}_i^{(\td)} \gets \spooky.\eval(\pk_i,\enc(k_i, \cdot),\ct_i)$ and
    $\tilde{\ct}_i^{(\sfe)} \gets \spooky.\enc(\pk_i,\ct_i^{(\sfe)})$
    where $\ct_i^{(\sfe)}$ is computed as specified in the protocol.

    \item Let $\eC^*_\mathsf{Final}$ be the quantum circuit (derived from the adversary) that, on input the ciphertexts $\{\ct_i^{(\td)},\ct_i^{(\sfe)}\}_{i\in I}$ and the quantum state of the adversary, computes the messages of the corrupted parties (corresponding to Step 7 of the protocol) and the updated quantum state of the adversary. Compute the following spooky evaluation procedure:
    $$(\bar{\ct}_{1},\dots,\bar{\ct}_{\ell},\widehat{\ct}_{\ket{\phi}}) \gets \spooky.\eval((\pk_1,\dots,\pk_\ell),\eC^*_\mathsf{Final},(\ct_{\ket{\phi}},\{\tilde{\ct}_i^{(\td)},\tilde{\ct}_i^{(\sfe)}\}_{i\in I}))$$
    where each $\bar{\ct}_{i} = (\bar{\ct}_{i}^{(1)},\dots,\bar{\ct}_{i}^{(\ell)})$ are classical ciphertexts encrypted under $\pk_1,\dots,\pk_\ell$.
    
    \item For all $i \in I$ compute 
$$(\widehat{\ct}_i^{(1)},\dots, \widehat{\ct}_i^{(\ell)})  \gets \spooky.\eval\left((\pk_1,\dots,\pk_\ell), \sfedec\left(\dk_i, \mathop{\bigoplus}\limits_{j=1}^\ell \cdot\right), \bar{\ct}_i\right).$$

    \item For all $i\in[\ell]$ define $\eC_{\mathsf{Nest},i}$ to be the classical circuit that, on input a set of strings $\{y_j \}_{j\in[i-1]}$, a public key $\pk$, and a ciphertext $\ct$, computes 
    $$\ct' \gets \spooky.\eval(\pk, y_1 \oplus \dots \oplus y_j \oplus \cdot, \ct).$$
    The circuit returns $\obfC_i(\ct')$.
    
\item For all $i \in [\ell, \dots, 2]$ compute iteratively
$$(\tilde{\ct}_{\sk, i}, \tilde{\ct}_{m,i}) \gets \spooky.\eval((\pk_1,\dots,\pk_{i-1}),\eC_{\mathsf{Nest},i}(\cdot, \pk_i, \widehat{\ct}^{(i)}_i),(\widehat{\ct}^{(1)}_i, \dots, \widehat{\ct}^{(i-1)}_i))$$ and for all $j \in [i-1]$ update the variables
$$(\widehat{\ct}_j^{(1)},\dots,\widehat{\ct}_{j}^{(i-1)}) \gets \spooky.\eval((\pk_1,\dots,\pk_{i-1}), \eC_{\mathsf{Rec},i}[\widehat{\ct}_j^{(i)}], (\widehat{\ct}_j^{(1)},\dots,\widehat{\ct}_{j}^{(i-1)}, \tilde{\ct}_{\sk, i}))$$
where $\eC_{\mathsf{Rec},i}[\widehat{\ct}_j^{(i)}]$ is the circuit that takes as input $2i-2$ strings $(z_1, \dots, z_{i-1})$ and $(s_1, \dots, s_{i-1})$ and computes $$\mathop{\bigoplus}\limits_{k=1}^{i-1} z_k\oplus \spooky.\dec\left(\mathop{\bigoplus}\limits_{k=1}^{i-1} s_k, \widehat{\ct}_j^{(i)} \right).$$
    
\item At the end of the iteration compute $({\sk}_1, m_1) \gets \obfC_1(\widehat{\ct}_1^{(1)})$, then for all $i \in [2, \dots, \ell]$ compute
$${\sk}_i \gets \mathop{\bigoplus}\limits_{k=1}^{i-1} \spooky.\dec(\sk_k, \tilde{\ct}_{\sk, i}^{(k)})$$ and
$${m}_i \gets \mathop{\bigoplus}\limits_{k=1}^{i-1} \spooky.\dec(\sk_k, \tilde{\ct}_{m, i}^{(k)}).$$

\item Use the extracted keys $({\sk}_1, \dots, {\sk}_\ell)$ to decrypt the state of the adversary from $\widehat{\ct}_{\ket{\phi}}$ and the ciphertexts $(\widehat{\ct}^{(\sfe)}_1, \dots, \widehat{\ct}^{(\sfe)}_\ell)$ from the ciphertexts $(\bar{\ct}_1, \dots, \bar{\ct}_\ell)$ as defined in Step 7 of the extractor. Return the transcript together with the state and the extracted messages $(m_1, \dots, m_\ell)$.
\end{enumerate}
We are now going to show that the transcript output by the extractor is computationally indistinguishable (with respect to compliant distinguishers) from that resulting from the real execution of the protocol. We do this by defining a sequence of hybrid distributions (for all $i\in I$) where we modify the interaction with the $i$-th corrupted party. In some of the following hybrids, the simulator inefficiently extracts the messages $m_i$ and the trapdoors $\td_i$ from the messages $(\cm_i^{(\msg)}, \cm_i^{(\td)})$ of the corrupted parties. Note that this implies that each distribution is not necessarily computable in polynomial time. However, these hybrids should be thought of as mental experiments, which are going to be helpful in arguing about the indistinguishability of the simulator (which instead runs in strict quantum polynomial time).

\begin{itemize}
    \item $\hyb_{0}$: $\{\mathsf{VIEW}^{\mathsf{msg}}_{\eC^*_\secp}(\dist{\eR,\eC^*_\secp(\rho_\secp),\{\eC_i(m_{i,\secp})\}_{i \notin I}}(1^\secp,1^n)\}_{\lambda\in\mathbb{N}}$. Recall that this distribution includes the messages $\{m_i\}_{i \in I}$ committed by the transcript (if they exist). 
    
    
    \item $\hyb_{i,1}$: Same as $\hyb_{0}$ except that $\ct_i^{(\sfe)}$ is computed as $\sfeenc(\dk_i, (0^\lambda,0^\lambda))$.
    
    \item $\hyb_{i,2}$: Same as $\hyb_{i,1}$ except that $\cm_i^{(\key)}$ is computed as a commiment to $0^\lambda$.
    
    \item $\hyb_{i,3}$: Same as $\hyb_{i,2}$ except that $\ct_i^{(\td)}$ is computed as $\enc(k_i, \td_i)$, where $\td_i$ is extracted (inefficiently) from $\cm_i^{(\td)}$.
    
    \item $\hyb_{i,4}$: Same as $\hyb_{i,3}$ except that $\cm_i^{(\key)}$ is computed as a commitment to $k_i$, as specified in the protocol.
    
    \item $\hyb_{i,5}$: Same as $\hyb_{i,4}$ except that $\ct_i^{(\sfe)}$ is computed as $\sfeenc(\dk_i, (k_i, r_i))$, as specified in the protocol. 
\end{itemize}
We then define the last hybrid below. Note that the distribution induced by this hybrid is computable in (quantum) polynomial time.
\begin{itemize}
    \item $\hyb_{6}$: This is the output of the extractor as described above.
\end{itemize}
Now we argue indistinguishability between each hybrid. All reductions for distinguishers between hybrids below may receive the state of $\eC^*$ after Step 1 along with the corresponding committed values $\{m_i,\td_i\}_{i \in I}$ (if they exist) as non-uniform advice. Also, we only consider distinguishers that succeed with non-negligible probability, even conditioned on the event that the transcript received is \emph{explainable}, i.e., each message lies in the support of the corresponding algorithm. This is because we consider only distinguishers that are \emph{compliant}, i.e., they output $0$ with overwhelming probability if the transcript is not explainable.

\begin{itemize}

    \item $\hyb_{0} \approx_c \hyb_{i,1}$: This follows from the quantum input privacy of the SFE protocol. The reduction takes the transcript of the protocol after executing Step 1 (including the state of the adversary and the messages $\{m_i\}_{i\in I}$) as non-uniform advice and continues to run the protocol honestly through Step 5. It then sets $(0,0)$ and $(k_i, r_i)$ as the challenge messages for SFE. In Step 6, the reduction sets $\ct_i^{(\sfe)} = \ct^\ast$, where $\ct^\ast$ is the challenge ciphertext. The rest of the protocol proceeds without changes. The reduction returns whatever the distinguisher returns.
    
    Clearly if $\ct^\ast = \sfeenc(\dk_i, (0,0))$, then the distribution is identical to $\hyb_{i,2}$. On the other hand if $\ct^\ast = \sfeenc(\dk_i, (k_i,r_i))$, then the distribution induced by the reduction is identical to $\hyb_{i,1}$. This implies that the two hybrids are computationally indistinguishable.
    
    \item $\hyb_{i,1} \approx_c \hyb_{i,2}$:  This follows from an invocation of the (non-uniform) quantum computational hiding of the commitment scheme.
    
    \item $\hyb_{i,2} \approx_c \hyb_{i,3}$: This follows from a (non-uniform) reduction to the quantum semantic security of the secret-key encryption scheme, where $\td_i$ (together with the transcript so far and the messages $\{m_i\}_{i\in I}$) is given as non-uniform advice to the reduction.
    
    \item $\hyb_{i,3} \approx_c \hyb_{i,4}$:
    Same argument as $\hyb_{i,1} \approx_c \hyb_{i,2}$.
    \item $\hyb_{i,4} \approx_c \hyb_{i,5}$: Same argument as $\hyb_{0} \approx_c \hyb_{i,1}$.
    
    \item $\hyb_{i,5} \approx_s \hyb_{6}$: We are going to argue that, conditioned on the event that the messages of the corrupted parties are explainable, the two hybrids are identical, except if an error in the evaluation (and consequently in the decryption) of the AFS-spooky encryption scheme occurs. Once that is estabilshed, statistical indistinguishability follows from the correctness of the AFS-spooky encryption scheme (which holds for \emph{all choices} of the random coins used in the setup, key generation, and encryption algorithms).
    To substantiate this claim, recall that 
        $$(\widehat{\ct}_{\ket{\phi}},\bar{\ct}_{1},\dots,\bar{\ct}_{\ell}) \gets \spooky.\eval((\pk_1,\dots,\pk_\ell),\eC^*_\mathsf{Final},(\ct_{\ket{\phi}},\{\tilde{\ct}_i^{(\td)},\tilde{\ct}_i^{(\sfe)}\}_{i\in I}))$$
    where 
    \begin{align*}
    \tilde{\ct}_i^{(\sfe)} &= \spooky.\enc(\pk_i,\ct_i^{(\sfe)})\\
    &= \spooky.\enc(\pk_i,\sfeenc(\dk_i, (k_i, r_i)))\\
    \end{align*}
    and
    \begin{align*}
        \tilde{\ct}_i^{(\td)} &= \spooky.\eval(\pk_i,\enc(k_i, \cdot),\ct_i)\\
        &= \spooky.\eval(\pk_i,\enc(k_i, \cdot),\spooky.\enc(\pk_i, \td_i))\\
        &= \spooky.\enc(\pk_i,\enc(k_i, \td_i)).
    \end{align*}
    Therefore, by definition of $\eC^*_\mathsf{Final}$ we have that for all $i\in[\ell]$ and $j \in [\ell]$
    \begin{align*}
    \bar{\ct}^{(j)}_i = \spooky.\enc(\pk_j, x^{(j)}_i)
    \end{align*}
    such that
    \begin{align*}
    \mathop{\bigoplus}\limits_{j=1}^{\ell} x^{(j)}_i &= \widehat{\ct}_i^{(\sfe)}
    =  \sfe.\eval(\mathsf{C}[\cm_i^{(\key)},\ct_i^{(\td)},\td_i,\lock_i],\ct_i^{(\sfe)}).
    \end{align*}
    Recall that
    \begin{align*}
        (\widehat{\ct}_i^{(1)},\dots, \widehat{\ct}_i^{(\ell)})  &= \spooky.\eval\left((\pk_1,\dots,\pk_\ell), \sfedec\left(\dk_i, \mathop{\bigoplus}\limits_{j=1}^\ell \cdot\right), \bar{\ct}_i\right)
    \end{align*}
    and therefore for all $i\in[\ell]$ and $j\in [\ell]$ we have that $\widehat{\ct}_i^{(j)} = \spooky.\enc(\pk_j, y_i^{(j)})$ such that
    \begin{align*}
    \mathop{\bigoplus}\limits_{j=1}^{\ell} y^{(j)}_i &= \sfedec\left(\dk_i, \mathop{\bigoplus}\limits_{j=1}^\ell x_i^{(j)}\right) \\
    &= \sfedec\left(\dk_i, \sfe.\eval(\mathsf{C}[\cm_i^{(\key)},\ct_i^{(\td)},\td_i,\lock_i],\ct_i^{(\sfe)})\right) \\
    &= \lock_i
    \end{align*}
    by the perfect correctness of the SFE protocol. Now recall that
    $$(\tilde{\ct}_{\sk, \ell}, \tilde{\ct}_{m,\ell}) = \spooky.\eval((\pk_1,\dots,\pk_{\ell-1}),\eC_{\mathsf{Nest},\ell}(\cdot, \pk_\ell, \widehat{\ct}^{(\ell)}_\ell),(\widehat{\ct}^{(1)}_\ell, \dots, \widehat{\ct}^{(\ell-1)}_\ell))$$
    which implies that the two ciphertexts encode the output of the obfuscated program $\obfC_\ell(\ct_\ell')$ where
    \begin{align*}
        \ct_\ell' &= \spooky.\eval(\pk_\ell, y_\ell^{(1)} \oplus \dots \oplus y_{\ell}^{(\ell-1)} \oplus \cdot, \widehat{\ct}_\ell^{(\ell)})\\
        &= \spooky.\enc(\pk_\ell, y_\ell^{(1)} \oplus \dots \oplus y_{\ell}^{(\ell)})\\
        &= \spooky.\enc(\pk_\ell, \lock_\ell).
    \end{align*}
 By the perfect correctness of the compute-and-compare obfuscation, the two variables $(\tilde{\ct}_{\sk, \ell}, \tilde{\ct}_{m,\ell})$ are AFS-spooky encryptions of $(\sk_\ell, m_\ell)$, under $(\pk_1, \dots, \pk_\ell)$.
This implies that the variables $(\widehat{\ct}_1, \dots, \widehat{\ct}_{\ell-1})$ are correctly updated to
\begin{align*}
    \widehat{\ct}_i &= (\widehat{\ct}_i^{(1)},\dots,\widehat{\ct}_{i}^{(\ell-1)})\\
    &= \spooky.\eval((\pk_1,\dots,\pk_{\ell-1}), \eC_{\mathsf{Rec},\ell}[\widehat{\ct}_i^{(\ell)}], (\widehat{\ct}_i^{(1)},\dots,\widehat{\ct}_{i}^{(\ell-1)}, \tilde{\ct}_{\sk, \ell}))
\end{align*}
where $\widehat{\ct}_i^{(j)} = \spooky.\enc(\pk_j, \tilde{y}_i^{(j)})$ such that
\begin{align*}
\mathop{\bigoplus}\limits_{j=1}^{\ell-1} \tilde{y}^{(j)}_i &= \mathop{\bigoplus}\limits_{k=1}^{\ell-1} z_k\oplus \spooky.\dec\left(\mathop{\bigoplus}\limits_{k=1}^{\ell-1} s_k, \widehat{\ct}_i^{(\ell)} \right)\\
&= \mathop{\bigoplus}\limits_{k=1}^{\ell-1} y_i^{(k)}\oplus \spooky.\dec\left(\sk_\ell, \widehat{\ct}_i^{(\ell)} \right)\\
&= y_i^{(1)} \oplus \dots \oplus y_i^{(\ell)}\\
&= \lock_i.
\end{align*}
by the definition of $\eC_{\mathsf{Rec},\ell}[\widehat{\ct}_i^{(\ell)}]$.
Recursively applying the above procedure, we obtain that
\begin{align*}
    \obfC_1(\widehat{\ct}_1)
     &= \obfC_1(\spooky.\enc(\pk_1, \lock_1))\\
     &=  ({\sk}_1, m_1)
\end{align*}
by the perfect correctness of the compute-and-compare obfuscation. It follows that the extractor successfully recomputes $\sk_1$, which allows it to iteratively recover $(\sk_2, \dots, \sk_\ell)$ from $(\tilde{\ct}_{\sk, 2}, \dots, \tilde{\ct}_{\sk, \ell})$. Consequently, the decrypted transcript, the (possibly quantum) state of the adversary, and the messages $(m_1, \dots, m_\ell)$ are distributed identically as in the previous hybrid, conditioned on the fact that no error occurs during the evaluation algorithm. 
\end{itemize}

\end{proof}
\section{Quantum-Secure Multi-Verifier Zero-Knowledge}
\label{sec:pzk}

In this section, we use standard techniques to derive a multi-verifier zero-knowledge protocol from our multi-committer extractable commitment. We follow the approach given in~\cite{JC:GolKah96} to upgrade a commit-challenge-response $\Sigma$ protocol to a full-fledged zero-knowledge protocol in constant rounds. In particular, the (multiple) verifiers will each commit to their challenge before the $\Sigma$ protocol is executed, using our multi-committer extractable commitment scheme. A simulator will then be able to extract the challenge from all verifiers \emph{simultaneously} and proceed to simulate each $\Sigma$ protocol. 

As in~\cite{BS20}, a couple of subtleties arise in the proof. First, the extractable commitment guarantee does not hold against arbitrary malicious verifiers, as captured by our notion of simulation indistinguishability against \emph{compliant} distinguishers. Thus, we have the verifier attach a witness indistinguishable proof (WI) that it acted explainably during the commitment phase, and indeed committed to the challenge that is sent during the $\Sigma$ protocol. However, in the proof of soundness, the verifier's initial commitment must be switched to a commitment to 0, since the reduction will receive the $\Sigma$ protocol challenge from its challenger. This requires the verifier to prove a different statement under the WI, which must only be possible when interacting with a cheating prover. Details can be found in the description of~\proref{fig:parallel_zk}.

\paragraph{Simulation Strategy.} Following~\cite{BS20}, we construct a zero-knowledge simulator that makes use of two non-rewinding sub-routines. Given an arbitrary malicious (multi-)verifier $\zkV^*$, we consider the following two distributions. First, consider the real distribution over the final state of $\zkV^*$ on interaction with the honest prover, except that any time $\zkV^*$ aborts, the distribution outputs only a $\bot$ symbol. We refer to this as $\mathsf{RealNoAbort}_\bot(\zkV^*)$. Next, consider the real distribution except that any time $\zkV^*$ \emph{does not} abort, the distribution outputs only a $\bot$ symbol. We refer to this as  $\mathsf{RealAbort}_\bot(\zkV^*)$. 

As a stepping stone towards proving zero-knowledge, we construct an entirely straight-line simulator $\mathsf{SimNoAbort}_\bot$ such that $\mathsf{SimNoAbort}_\bot(\zkV^*)$ is indistiguishable from $\mathsf{RealNoAbort}_\bot(\zkV^*)$. By entirely straight-line, we mean that not only does $\mathsf{SimNoAbort}_\bot$ not rewind $\zkV^*$, it never even re-starts $\zkV^*$ from the beginning. Analogously, we also construct a simulator $\mathsf{SimAbort}_\bot$ such that $\mathsf{SimAbort}_\bot(\zkV^*)$ is indistiguishable from $\mathsf{RealAbort}_\bot(\zkV^*)$, and $\mathsf{SimAbort}_\bot$ is entirely straight-line.

Now, we combine the above simulators into a straight-line simulator $\mathsf{SimComb}_\bot$ that succeeds with probability negligibly close to 1/2. $\mathsf{SimComb}_\bot$ simply chooses uniformly at random whether to run $\mathsf{SimNoAbort}_\bot$ or $\mathsf{SimAbort}_\bot$ and outputs the resulting view if the sub-routine is successful and $\bot$ otherwise. Finally, we invoke the Watrous rewinding lemma to amplify the success probability of $\mathsf{SimComb}_\bot$, resulting in the final simulator $\simulator$. 

We will actually make use of the sub-routines $\mathsf{SimNoAbort}_\bot$ and $\mathsf{SimAbort}_\bot$ explicitly in later sections, where the entirely straight-line nature of these procedures will be useful. In particular, we use both simulators in constructing non-malleable commitments (\cref{sec:nmc}) and just $\mathsf{SimNoAbort}_\bot$ in the coin-flipping protocol in~\cref{sec:coin-tossing} (since we define an alternate/simpler abort generation procedure in that protocol).

\subsection{Definition}


\begin{definition}[Quantum-Secure Multi-Verifier Zero-Knowledge Argument for $\mathsf{NP}$]
\label{def:qsczk}
A quantum-secure multi-verifier zero-knowledge argument for a language $\mathcal{L} \in \mathsf{NP}$ is a pair $(\zkP, \zkV)$ of classical PPT interactive Turing machines. $\zkP$ interacts with $n$ copies $\{\zkV_i\}_{i \in [n]}$ of $\zkV$ (who do not interact with each other) on common input $1^\secp$ and $1^n$, with each $\zkV_i$ additionally taking an input $x_i \in \cL$, and $\zkP$ additionally taking inputs $\{x_i, w_i \in \mathcal{R}_{\mathcal{L}}(x_i)\}_{i \in [n]}$. At the end of the interaction, each $\zkV_i$ outputs a bit, indicating whether it accepts or rejects.
\begin{enumerate}
    \item {\bf Perfect Completeness:} For any $\secp,n \in \mathbb{N}, i \in [n], x \in \mathcal{L} \cap \{0,1\}^{\secp}$, and $w \in \mathcal{R}_{\mathcal{L}}(x)$,
    $$\Pr[\mathsf{OUT}_{\zkV_i}\langle\zkP(x,w),\zkV_i(x)\rangle (1^\secp,1^n) = 1] = 1.$$
    \item {\bf Quantum Computational Soundness:} For any non-uniform quantum polynomial-size prover $\zkmP = \{\zkmP_\secp, \rho_\secp\}_{\secp \in \mathbb{N}}$, there exists a negligible function $\mu(\cdot)$ such that for all $\secp,n \in \mathbb{N}$, $i \in [n]$, and any $x \in \{0,1\}^\secp \setminus \mathcal{L}$,
    $$\Pr[\mathsf{OUT}_{\zkV_i}\langle\zkmP_\secp(\rho_\secp),\zkV_i(x)\rangle (1^\secp,1^n) = 1] \leq \mu(\secp).$$
    \item {\bf Quantum Computational Zero-Knowledge:} There exists a quantum expected polynomial-time simulator $\mathsf{Sim}$ such that for any non-uniform quantum polynomial-size adversary $\zkmV = \{\zkmV_\secp,\rho_\secp\}_{\secp \in \mathbb{N}}$ representing a subset of $n$ verifiers, namely, $\{\zkV_i\}_{i \in I}$ for some set $I \subseteq [n]$,
    
    \begin{align*}
    &\left\{\view_{\zkmV_\secp} \left\langle \begin{array}{c} \zkP(\{x_i,w_i\}_{i \in [n]}), \\\zkmV_\secp(\{x_i\}_{i \in I},\rho_\secp), \\ \{\zkV_i(x_i)\}_{i \notin I} \end{array}\right\rangle (1^\secp,1^n)  \right\}_{\secp,\{x_i\}_{i \in [n]},\{w_i\}_{i \in [n]}}\\ \approx_c 
    &\{\mathsf{Sim}(1^\secp,1^n,I,\{x_i\}_{i \in n}, \zkmV_\secp, \rho_\secp)\}_{\secp,\{x_i\}_{i \in [n]},\{w_i\}_{i \in [n]}},
    \end{align*}
    
    where $\secp \in \mathbb{N},x_i \in \cL \cap \{0,1\}^\secp,w_i \in \cR_\cL(x_i)$.
    
    
\end{enumerate}
\end{definition}

\subsection{Construction}

\paragraph{Ingredients:} All of the following are assumed to be quantum-secure.
\begin{itemize}
    \item A non-interactive perfectly-binding commitment $\COM$.
    \item A multi-committer extractable commitment $\PECom = (\PECom.\eC,\PECom.\eR)$.
    \item A WI proof system $\mathsf{WI} = (\wiP,\wiV)$.
    \item A sigma protocol for $\NP$ $\Sigma = (\Sigma.\zkP,\Sigma.\zkV)$.
\end{itemize}

\begin{remark}
Observe that since the sigma protocol $\Sigma$ is public-coin, \proref{fig:parallel_zk} is publicly-verifiable. That is, any third party, upon observing the transcript of interaction between $\zkP$ and $\zkV$, can deduce whether $\zkV$ accepted or not. This fact will be used in~\cref{sec:coin-tossing}.
\end{remark}

\protocol
{\proref{fig:parallel_zk}}
{A constant-round quantum-secure multi-verifier zero-knowledge argument for $\cL \in$ NP.}
{fig:parallel_zk}
{
\begin{description}
    \item[Common input:] $1^\secp$ and $1^n$.
    \item[$\zkV_i$'s additional input:] $x_i \in \cL$.
    \item[$\zkP$'s additional input:] $\{x_i,w_i \in \cR_\cL(x_i)\}_{i \in [n]}$.
\end{description}
\begin{enumerate}
    \item For each $i \in [n]$, $\zkP$ computes and sends $\cm_i \gets \COM(1^\secp,w_i)$ to $\zkV_i$.
    \item Each $\zkV_i$ computes a challenge $\beta_i \gets \Sigma.\zkV_1(1^{|x_i|})$. Then, $\zkP$ and $\{\zkV_i\}_{i \in [n]}$ interact, with $\zkP$ taking the role of $\PECom.\eR$ and $\zkV_i$ taking the role of $\PECom.\eC_i(\beta_i)$, to produce $\{\tau_i\}_{i \in [n]} 
    \gets \dist{\PECom.\eR,\{\PECom.\eC_i(\beta_i)\}_{i \in [n]}}(1^\secp,1^n)$.
    
    \item For each $i \in [n]$, $\zkP$ computes $(\alpha_i,\state_i) \gets \Sigma.\zkP_1(x_i,w_i)$ and sends $\alpha_i$ to $\zkV_i$.
    \item Each $\zkV_i$ sends $\beta_i$.
    \item For each $i \in [n]$, $\zkP$ and $\zkV_i$ interact (in parallel) with $\zkP$ taking the role of $\wiV$ and $\zkV_i$ taking the role of $\wiP$ to give $\zkP$ a WI proof that 
    \begin{itemize}
        \item $\tau_i$ is explainable, and opens to $\beta_i$,
        \item \textbf{OR,} $\cm_i$ opens to a non-witness $z_i \notin \cR_\cL(x_i)$.
    \end{itemize}
    \item For each $i \in [n]$, $\zkP$ and $\zkV_i$ interact (in parallel) with $\zkP$ taking the role of $\wiP$ and $\zkV_i$ taking the role of $\wiV$ to give $\zkV_i$ a WI proof that 
    \begin{itemize}
        \item $\cm_i$ opens to some string $z_i$,
        \item \textbf{OR,} $x_i \in \cL$.
    \end{itemize}
    \item For each $i \in [n]$, $\zkP$ computes and sends $\gamma_i = \Sigma.\zkP_2(x_i,w_i,\state_i,\alpha_i,\beta_i)$ to $\zkV_i$.
    \item Each $\zkV_i$ accepts if $\Sigma.\zkV_2(x_i,\alpha_i,\beta_i,\gamma_i) = 1$. 
\end{enumerate}
}

\subsection{Soundness}

\begin{lemma}
\proref{fig:parallel_zk} has quantum computational soundness.
\end{lemma}

\begin{proof}

Assume towards contradication that there exists a non-uniform quantum polynomial-size prover $\zkP^* = \{\zkP^*_\secp,\rho^*_\secp\}_{\secp \in \mathbb{N}}$ that with noticeable probability, convinces $\zkV_i$ to accept on input instances $\{x_\secp\}_{\secp \in \mathbb{N}}$ where $|x_\secp| = \secp$ and $x_\secp \notin \cL$. Let $\zkV = \zkV_i$, and by averaging, we can assume that $\zkP^*_\secp$ sends a fixed first message $\cm_\secp$ to $\zkV$. Furthermore, since $\zkP^*_\secp$ succeeds in convincing $\zkV$ to accept instances $x_\secp \notin \cL$ with noticeable probability, the statistical soundness of the WI in Step 6 of the protocol implies that $\cm_\secp$ must be a well-formed commitment, that is, $\cm_\secp = \COM(1^\secp,z_\secp;s_\secp)$ for some $(z_\secp,s_\secp)$. Now, consider the following sequence of computationally indistinguishable hybrid distributions.

\begin{itemize}
    \item $\hyb_0$: $\{\view_{\zkP^*}\dist{\zkP^*(\rho_\secp),\zkV}(1^\secp,x_\secp)\}_{\secp \in \mathbb{N}}$.
    \item $\hyb_1$: Same as $\hyb_0$ except that $\zkV$ uses $(z_\secp,s_\secp)$ as the witness for the second part of the WI statement given in Step 5 of the protocol. $\hyb_0 \approx_c \hyb_1$ follows from the witness indistinguishability of WI, where the reduction is given $(z_\secp,s_\secp)$ as non-uniform advice.
    \item $\hyb_2$: Same as $\hyb_1$ except that in Step 2, $\zkV$ takes the role of $\PECom.\eC(0^\secp)$ rather than $\PECom.\eC(\beta)$. $\hyb_1 \approx_c \hyb_2$ follows from the computational hiding of $\PECom$.
\end{itemize}

Using $\zkP^*$, we construct a cheating prover $\Sigma.\zkP^* = \{\Sigma.\zkP^*_\secp,\Sigma.\rho^*_\secp\}_{\secp \in \mathbb{N}}$ for the sigma protocol. The non-uniform advice $\Sigma.\rho^*_\secp$ is generated as follows. Run $\zkP^*(\rho^*_\secp)$ until it outputs its first message $\cm_\secp$. Extract from $\cm_\secp$ the message $z_\secp$ committed and the corresponding opening $s_\secp$ and define the resulting advice to consist of the state of $\zkP^*$ at this point, along with $(z_\secp,s_\secp)$.\\

\noindent$\Sigma.\zkP^*_\secp(\Sigma.\rho^*_\secp):$
\begin{enumerate}
    \item Interact with $\zkP^*_\secp$, taking the role of $\PECom.\eC(0^\secp)$.
    \item Continue running $\zkP^*_\secp$, obtaining the message $\alpha$.
    \item Send $\alpha$ to $\Sigma.\zkV$, receive $\beta$, and send $\beta$ to $\zkP^*_\secp$.
    \item Interact with $\zkP^*_\secp$ to give $\zkP^*_\secp$ a WI proof as in Step 5 of the protocol, using witness $(z_\secp,s_\secp)$.
    \item Interact with $\zkP^*_\secp$ to receive a WI proof from $\zkP^*_\secp$ as in Step 6 of the protocol.
    \item Continue running $\zkP^*_\secp$, obtaining the message $\gamma$, and send $\gamma$ to $\Sigma.\zkV$.
\end{enumerate}

Now note that $\zkP^*_\secp$'s view in this interaction is exactly $\hyb_2$. Thus, since $\zkP^*$ succeeds in convincing $\zkV$ to accept with noticeable probability, and $\hyb_0 \approx_c \hyb_2$, it must be the case that $\Sigma.\zkV$ accepts with noticeable probability, a contradiction.

\end{proof}

\subsection{Zero-Knowledge}
\label{subsec:ZK}

\begin{theorem}
\label{thm:ZK}
\proref{fig:parallel_zk} is quantum computational zero-knowledge.
\end{theorem}

\begin{proof}


We begin by describing the two sub-routines $\zkSimNoAbort$ and $\zkSimAbort$ mentioned above. Then we combine them into $\mathsf{SimComb}_\secp$, which we use to derive the final simulator $\simulator$.\\



\noindent $\zkSimNoAbort(1^\secp,1^n,I,\{x_i\}_{i \in [n]},\zkV^*_\secp,\rho_\secp)$:
\begin{enumerate}
    \item Set $\rho_\secp$ to be the inner state of $\zkV^*_\secp$. For each $i \in I$, compute and send $\cm_i \gets \COM(1^\secp,0)$ to $\zkV^*_\secp$.
    
    \item Let $\zkV.\PECom^*_\secp$ be the portion of $\zkV^*_\secp$ that interacts with $\zkP$ in Step 2 above. Note that its state at the beginning of this interaction is $\rho_\secp, \{\cm_i\}_{i \in I}$. Compute 
    
    $$(\{\tau_i\}_{i \in I},\state,\{\beta'_i\}_{i \in I}) \gets \PECom.\cE(1^\secp,1^n,I,\zkV.\PECom^*_\secp,(\rho_\secp, \{\cm_i\}_{i \in I})).$$
    
     If $\PECom.\cE$ produced an abort transcript, then halt and return $\bot$. Otherwise continue, setting $\state$ to be the inner state of $\zkV^*_\secp$.
    \item For all $i \in I$, compute $(\alpha_i,\gamma_i) \gets \Sigma.\zkSim(x_i,\beta'_i)$ and send $\alpha_i$ to $\zkV^*_\secp$.
    \item $\zkV^*_\secp$ returns $\{\beta_i\}_{i \in I}$.
    \item Take the role of the honest prover $\wiP$ in the $|I|$ WI proofs that $\zkV^*_\secp$ gives. If $\zkV^*_\secp$ fails to prove any of the statements, then halt and output $\bot$.
    \item Give $\zkV^*_\secp$ a total of $|I|$ WI proofs using the $|I|$ witnesses that show $\{\cm_i\}_{i \in I}$ are valid commitments. Then, send $\{\gamma_i\}_{i \in I}$ to $\zkV^*_\secp$.
    \item Output the inner state of $\zkV^*_\secp$.
\end{enumerate}

\medskip

\noindent $\zkSimAbort(1^\secp,1^n,I,\{x_i\}_{i \in [n]},\zkV^*_\secp,\rho_\secp)$:
 \begin{enumerate}
        \item Set $\rho_\secp$ to be the inner state of $\zkV^*_\secp$.
        \item Interact with $\zkV^*_\secp$ as the honest prover until the end of Step 6 of the protocol, with exactly 3 differences:
        \begin{itemize}
            \item The commitments $\cm_i$ in Step 1 are to $0$ rather than $w_i$.
            \item The messages $\alpha_i$ sent in Step 3 are generated by the simulator of the sigma protocol, $(\alpha_i,\gamma_i) \gets \Sigma.\zkSim(x_i,0^\secp)$.
            \item In Step 6, the witnesses used for the WI proofs are for the first statement (that $\cm_i$ is a valid commitment).
        \end{itemize}
        \item If at some point during the interaction $\zkV^*_\secp$ either aborts or fails in one of its WI proofs, halt and output $\zkV^*_\secp$'s inner state. Otherwise, output $\bot$.
    \end{enumerate}
    
\medskip

\noindent $\mathsf{SimComb}_\bot(1^\secp,1^n,I,\{x_i\}_{i \in [n]},\zkV^*_\secp,\rho_\secp)$: With probability 1/2, execute $\zkSimNoAbort(1^\secp,1^n,I,\{x_i\}_{i \in [n]},\zkV^*_\secp,\rho_\secp)$ and otherwise execute $\zkSimAbort(1^\secp,1^n,I,\{x_i\}_{i \in [n]},\zkV^*_\secp,\rho_\secp)$.

\medskip
\noindent $\simulator(1^\secp,1^n,I,\{x_i\}_{i \in [n]},\zkV^*_\secp,\rho_\secp)$: Let $\overline{\mathsf{SimComb}}_\bot(\cdot) \coloneqq \mathsf{SimComb}_\bot(1^\secp,1^n,I,\{x_i\}_{i \in [n]},\zkV^*_\secp,\cdot)$ be the circuit $\mathsf{SimComb}_\bot$ with all inputs hard-coded except for $\rho_\secp$, and output $\R(\overline{\mathsf{SimComb}}_\bot,\rho_\secp,\secp)$, where $\R$ is the algorithm from~\cref{lemma:rewinding}.

\medskip
    
Next, we introduce some notation. For $\zkV^* = \{\zkV^*_\secp,\rho_\secp\}_{\secp \in \bbN}$, let $\mathsf{Real}(\zkV^*)$ denote 

$$\left\{\view_{\zkmV_\secp} \left\langle \begin{array}{c} \zkP(\{x_i,w_i\}_{i \in [n]}), \\\zkmV_\secp(\{x_i\}_{i \in I},\rho_\secp), \\ \{\zkV_i(x_i)\}_{i \notin I} \end{array}\right\rangle (1^\secp,1^n)  \right\}_{\secp,\{x_i\}_{i \in [n]},\{w_i\}_{i \in [n]}}$$ for $\secp \in \mathbb{N},x_i \in \cL \cap \{0,1\}^\secp,w_i \in \cR_\cL(x_i)$, and let

\begin{itemize}
    \item $\mathsf{RealNoAbort}_\bot(\zkV^*)$ be the distribution $\mathsf{Real}(\zkV^*)$, except that whenever an abort occurs, the distribution outputs $\bot$, and 
    \item $\mathsf{RealAbort}_\bot(\zkV^*)$ be the distribution $\mathsf{Real}(\zkV^*)$, except that if an abort \emph{does not} occur, the distribution outputs $\bot$.
\end{itemize}



We continue by proving two lemmas that will be useful on their own in later sections, and will also be useful in proving the quantum computational zero-knowledge of $\zk$.

\begin{lemma}
\label{claim:ind-non-abort}
For any $\zkV^*$, $\mathsf{RealNoAbort}_\bot(\zkV^*) \approx_c \zkSimNoAbort(\zkV^*)$.
\end{lemma}
\begin{proof} We prove this via a sequence of hybrids.
\begin{itemize}
    \item $\hyb_0$: $\mathsf{RealNoAbort}_\bot(\zkV^*)$.
    \item $\hyb_1$: Same as $\hyb_0$ except that after Step 2, the values $\{\beta'_i\}_{i \in I}$ committed by the transcripts $\{\tau_i\}_{i \in I}$ are inefficiently extracted, and after Step 4, if any $\beta'_i \neq \beta_i$, the hybrid aborts (outputs $\bot$).
    \item $\hyb_2$: Same as $\hyb_1$ except that the transcript of the extractable commitment is simulated. In particular, the hybrid computes  $$(\{\tau_i\}_{i \in I},\state,\{\beta'_i\}_{i \in I}) \gets \PECom.\cE(1^\secp,1^n,I,\zkV.\PECom^*_\secp,(\rho_\secp, \{\cm_i\}_{i \in I})),$$
    and proceeds to run the verifier with inner state $\state$.
    \item $\hyb_3$: Same as $\hyb_2$ except that the equality checks introduced in $\hyb_1$ are removed, and the sigma protocol is simulated. In particular, for each $i \in I$, the hybrid computes $(\alpha_i,\gamma_i) \gets \Sigma.\zkSim(x_i,\beta_i')$.
    \item $\hyb_4$: Same as $\hyb_3$ except that each commitment for $i \in I$ in the first message sent by the prover is $\cm_i \gets \COM(1^\secp,0)$.
    \item $\hyb_5$: Same as $\hyb_4$ except that the WI proofs given by the prover for $i \in I$ are generated with witnesses showing that $\cm_i$ is a valid commitment.
\end{itemize}
Observe that $\hyb_5$ is exactly $\zkSimNoAbort(\zkV^*)$. Now we show that each consecutive pair of hybrids is indistiguishable.
\begin{itemize}
    \item $\hyb_0 \approx_s \hyb_1$: This follows from the statistical soundness of the WI proved in Step 5. 
    \item $\hyb_1 \approx_c \hyb_2$: Assume there exists a distinguisher $\D$ for $\hyb_1$ and $\hyb_2$ that succeeds with non-negligible probability. We build a compliant\footnote{Recall that such a distinguisher is guaranteed to output 0 with overwhelming probability on input any non-explainable view.} distinguisher $\D'$ that breaks the extractability property of $\PECom$.
    
    First, we fix a sequence of instance-witness pairs $\{\{x_{i,\secp},w_{i,\secp}\}_{i \in I}\}_{\secp \in \mathbb{N}}$, and first messages $\{\{\cm_{i,\secp}\}_{i \in I}\}_{\secp \in \mathbb{N}}$ for which $\D$ succeeds with non-negligible probability. The witnesses $\{\{w_{i,\secp}\}_{i \in I}\}_{\secp \in \mathbb{N}}$ will be given as non-uniform advice to $\D'$. Now, $\D'$ will take as input either the real or the simulated view with respect to committer $\{\zkV.\PECom^*_\secp,(\rho_\secp,\{\cm_{i,\secp}\}_{i \in I})\}_{\secp \in \mathbb{N}}$. This view includes the messages committed and the final state of the committer, which is the state of $\zkV^*$ after Step 2 of the protocol. $\D'$ proceeds to simulate the rest of the interaction between $\zkP$ and $\zkV^*$, making use of the witnesses it received as non-uniform advice during Steps 3 and 7, as well as the committed messages it received from its challenger to implement the check introduced in $\hyb_1$. If $\zkV^*$ aborts or fails to prove any of the WI statements in Step 5, $\D'$ outputs 0. Otherwise, it queries $\D$ with $\zkV^*$'s final view and outputs what $\D'$ outputs.
    
    Observe that $\D'$'s advantage is equivalent to $\D$'s advantage. This follows becuase i) whenever $\D'$ queries $\D$ with a transcript, it is a faithful execution of either $\hyb_1$ or $\hyb_2$, depending on whether $\PECom$ was simulated or not, and ii) whenever $\D'$ does \emph{not} query $\D$, it means that $\zkV^*$ failed to prove one of its WI statements, so $\D$'s input would have been $\bot$. Finally, $\D'$ is compliant by the statistical soundness of the WI.

    \item $\hyb_2 \approx_c \hyb_3$: This follows from the special zero-knowledge property of $\Sigma$.
    \item $\hyb_3 \approx_c \hyb_4$: This follows from the computational hiding of $\COM$.
    \item $\hyb_4 \approx_c \hyb_5$: This follows from the witness indistinguishability of WI.
\end{itemize}

\end{proof}

\begin{lemma}
\label{claim:ind-abort}
For any $\zkV^*$, $\mathsf{RealAbort}_\bot(\zkV^*) \approx_c \zkSimAbort(\zkV^*)$.
\end{lemma}
\begin{proof}
We prove this via a sequence of hybrids.
\begin{itemize}
    \item $\hyb_0$: $\mathsf{RealAbort}_\bot(\zkV^*)$. Note that if $\zkV^*$ has not aborted at some point during Steps 1-6 of the protocol, this distribution outputs $\bot$.
    \item $\hyb_1$: Same as $\hyb_0$ except that in Step 3, for each $i \in I$, $\zkP$ sends $\alpha_i$ where  $(\alpha_i,\gamma_i) \gets \zkSim(x_i,0^\secp)$.
    \item $\hyb_2$: Same as $\hyb_1$ except that each commitment for $i \in I$ in the first message sent by the prover is $\cm_i \gets \COM(1^\secp,0)$.
    \item $\hyb_3$: Same as $\hyb_2$ except that the WI proofs for $i \in I$ given by the prover are generated with witnesses showing that $\cm_i$ is a valid commitment.
\end{itemize}
Observe that $\hyb_3$ is exactly $\zkSimAbort(\zkV^*)$. Now we show that each consecutive pair of hybrids in indistinguishable.
\begin{itemize}
    \item $\hyb_0 \approx_c \hyb_1$: This follows from the first-message indistinguishability of $\Sigma$.
    \item $\hyb_1 \approx_c \hyb_2$: This follows from the computational hiding of $\COM$.
    \item $\hyb_2 \approx_c \hyb_3$: This follows from the witness indistinguishability of WI.
\end{itemize}
\end{proof}

\noindent To finish the proof of zero-knowledge, we introduce some more notation.

\begin{itemize}
    \item Let $\mathsf{Pr}^{\mathsf{Abort}}_{\mathsf{Real}}(\zkV^*)$ be the probability that $\zkV^*$ aborts in the real interaction with the honest prover.
    \item Let $\mathsf{Pr}^{\mathsf{Abort}}_{\mathsf{SimNoAbort}}(\zkV^*)$ be the probability that $\zkV^*$ aborts in $\zkSimNoAbort$ (i.e. the outcome is $\bot$).
    \item Let $\mathsf{Pr}^{\mathsf{Abort}}_{\mathsf{SimAbort}}(\zkV^*)$ be the probability that $\zkV^*$ aborts in $\zkSimAbort$ (i.e. the outcome is not $\bot$).
    \item Let $\zkSimNoAbort(\zkV^*) \coloneqq \{\zkSimNoAbort(1^\secp,n,I,\{x_i\}_{i \in [n]},\zkV^*_\secp,\rho_\secp)\}_{\secp,\{x_i\}_{i \in [n]},\{w_i\}_{i \in [n]}}$.
    \item Let $\zkSimAbort(\zkV^*) \coloneqq \{\zkSimAbort(1^\secp,n,I,\{x_i\}_{i \in [n]},\zkV^*_\secp,\rho_\secp)\}_{\secp,\{x_i\}_{i \in [n]},\{w_i\}_{i \in [n]}}$.
    \item Let $\mathsf{SimComb}_\bot(\zkV^*) \coloneqq \{\mathsf{SimComb}_\bot(1^\secp,1^n,I,\{x_i\}_{i \in [n]},\zkV^*_\secp,\rho_\secp)\}_{\secp,\{x_i\}_{i \in [n]},\{w_i\}_{i \in [n]}}$.
    \item Let $\zkSim(\zkV^*) \coloneqq \{\zkSim(1^\secp,1^n,I,\{x_i\}_{i \in [n]},\zkV^*_\secp,\rho_\secp)\}_{\secp,\{x_i\}_{i \in [n]},\{w_i\}_{i \in [n]}}$.
    \item Let $\mathsf{RealNoAbort}(\zkV^*)$ be the distribution $\mathsf{Real}(\zkV^*)$ conditioned on there not being an abort.
    \item Let $\mathsf{RealAbort}(\zkV^*)$ be the distribution $\mathsf{Real}(\zkV^*)$ conditioned on there being an abort.
    \item Let $\mathsf{SimNoAbort}(\zkV^*)$ be the distribution $\zkSimNoAbort$ conditioned on there not being an abort (i.e. conditioned on the output not being $\bot$).
    \item Let $\mathsf{SimAbort}(\zkV^*)$ be the distribution $\zkSimAbort$ conditioned on there being an abort (i.e. conditioned on the output not being $\bot$).
    \item Let $\mathsf{SimComb}(\zkV^*)$ be the distribution $\mathsf{SimComb}_\bot$ conditioned on the output not being $\bot$.
\end{itemize}

Following~\cite{BS20}, we show that $\mathsf{Real}(\zkV^*) \approx_c \mathsf{SimComb}(\zkV^*)$ via a sequence on hybrids. In particular, we show that 

\begin{align*}
    \mathsf{Real}(\zkV^*) &\substack{(1) \\ \equiv \\ \ } (1-\mathsf{Pr}^\mathsf{Abort}_\mathsf{Real}(\zkV^*))\mathsf{RealNoAbort}(\zkV^*) + (\mathsf{Pr}^\mathsf{Abort}_\mathsf{Real}(\zkV^*))\mathsf{RealAbort}(\zkV^*)\\
    &\substack{(2) \\ \approx_s \\ \ }(1-\mathsf{Pr}^\mathsf{Abort}_\mathsf{SimNoAbort}(\zkV^*))\mathsf{RealNoAbort}(\zkV^*) + (\mathsf{Pr}^\mathsf{Abort}_\mathsf{SimAbort}(\zkV^*))\mathsf{RealAbort}(\zkV^*)\\
    &\substack{(3) \\ \approx_c \\ \ } (1-\mathsf{Pr}^\mathsf{Abort}_\mathsf{SimNoAbort}(\zkV^*))\mathsf{SimNoAbort}(\zkV^*) + (\mathsf{Pr}^\mathsf{Abort}_\mathsf{SimAbort}(\zkV^*))\mathsf{RealAbort}(\zkV^*)\\
    &\substack{(4) \\ \approx_c \\ \ } (1-\mathsf{Pr}^\mathsf{Abort}_\mathsf{SimNoAbort}(\zkV^*))\mathsf{SimNoAbort}(\zkV^*) + (\mathsf{Pr}^\mathsf{Abort}_\mathsf{SimAbort}(\zkV^*))\mathsf{SimAbort}(\zkV^*)\\
    &\substack{(5) \\ \approx_s \\ \ }\mathsf{SimComb}(\zkV^*),
\end{align*}
    
where 

\begin{enumerate}
    \item The equality $(1)$ follows by definition.
    \item The indistinguishability $(2)$ follows as a corollary of~\cref{claim:ind-non-abort} and~\cref{claim:ind-abort}. Indeed, $\mathsf{RealNoAbort}_\bot(\zkV^*) \approx_c \zkSimNoAbort(\zkV^*)$ in particular implies that the difference in the probability that the verifier aborts in the real interaction versus the simulated interaction is negligible, and likewise for $\mathsf{RealAbort}_\bot(\zkV^*) \approx_c \zkSimAbort(\zkV^*)$.
    \item The indistinguishability $(3)$ follows as a corollary of~\cref{claim:ind-non-abort}. This can be seen by considering two cases. First, if the probability that the verifier aborts in the real interaction is negligible, then $\mathsf{RealNoAbort}(\zkV^*) \approx_c \mathsf{SimNoAbort}(\zkV^*)$ directly follows from~\cref{claim:ind-non-abort}, and the indistinguishability follows. Otherwise, this probability is non-negligible, meaning that $\mathsf{RealAbort}(\zkV^*)$ is efficiently sampleable. Thus, a reduction to~\cref{claim:ind-non-abort} can sample from the distribution $\mathsf{RealAbort}(\zkV^*)$ whenever it receives $\bot$ from its challenger.\footnote{A more formal analysis of this can be found in~\cite[Proposition~3.4]{BS20}.}
    \item The indistinguishability $(4)$ follows as a corollary of~\cref{claim:ind-abort} via a similar analysis as the last step.
    \item The indistinguishability $(5)$ follows from the definition of $\zkSim(\zkV^*)$ and the claim that the difference between $\mathsf{Pr}^\mathsf{Abort}_\mathsf{SimNoAbort}$ and $\mathsf{Pr}^\mathsf{Abort}_\mathsf{SimAbort}$ is negligible (which follows as a corollary of~\cref{claim:ind-non-abort} and~\cref{claim:ind-abort}).
\end{enumerate}

Finally, this implies that $\mathsf{Real}(\zkV^*) \approx_c \zkSim(\zkV^*)$ by applying~\cref{lemma:rewinding} for each $(\secp,\{x_i\}_{i \in [n]},\{w_i\}_{i \in [n]})$ with the following parameters. Set $\Q \coloneqq \overline{\mathsf{SimComb}}_\bot$ as defined in the description of $\simulator$, and set $\epsilon \coloneqq \negl(\secp) + 2^{-\secp \cdot \frac{3}{4}},p_0 \coloneqq 1/4,$ and $q \coloneqq 1/2$, as described in~\cite[Proposition 3.5]{BS20}. This completes the proof of quantum computational zero-knowledge.

\end{proof}

\section{Quantum-Secure Non-Malleable Commitments}

\label{sec:nmc}
\subsection{Definition}
In this section, we define quantum-secure non-malleable commitments w.r.t. commitment.
We consider the synchronous setting where there is a quantum man-in-the-middle adversary $\mim = \{\mim_\lambda, \rho_\lambda\}_{\lambda \in \bbN}$ interacting with a classical honest committer $\cC$ with tag $\tagg_\cC$ (where $\cC$ commits to value $v$) in the left session, and interacting with classical honest receiver $\cR$ in the right session. The \mim uses tag $\tagg_\mim$ in its interaction with \cR. Prior to the interaction, the value $v$ is given to $\cC$ as local input.

Then the commit phase is executed. 
After obtaining an honest left message in any round, the \mim sends its own right message. 
And after obtaining an honest right message in any round, the \mim sends its own left message.
Let $\viewval_{\mim_\lambda} \dist{\cC(v), \mim(\rho_\lambda), \cR} (1^\lambda, \tagg_\cC, \tagg_\mim)$ denote a random variable that describes the value $v'$ committed by the \mim in the right session, jointly with the view of the \mim in the full (both left and right sessions) experiment. 
If the $\tagg_\cC$ used by $\cC$ in the left interaction is identical to the $\tagg_\mim$ used by the \mim in the right interaction, then the value $v'$ committed to in the right interaction is defined to be $\bot$. If the \mim sends a message that causes an honest party to abort in either the left or the right execution, then the value $v'$ committed to in the right interaction is also defined to be $\bot$.

We will concern ourselves with computationally hiding and statistically binding commitments that additionally satisfy the non-malleability property defined below.

\begin{definition}
[Quantum-Secure Non-Malleable Commitments with respect to Commitment] \label{def:nmc}
For any $\ell = \ell(\secp)$ and $p = p(\secp)$, a commitment scheme $\dist{\cC, \cR}$ is said to be quantum secure non-malleable with respect to commitment for tags in $[\ell]$ if 
for every $v_1, v_2 \in \{0, 1\}^{2p(\secp)}$, for every quantum polynomial-size $\mim = \{\mim_\secp, \rho_\secp\}_{\secp \in \bbN}$ and every quantum polynomial-size distinguisher $\eD = \{\eD_\secp, \sigma_\secp\}_{\secp \in \bbN}$, there exists a negligible function $\eta(\cdot)$ such that for all large enough $\secp \in \bbN$, and for all $\tagg_\cC, \tagg_\mim \in [\ell]$ where $\tagg_\cC \ne \tagg_\mim$, the following holds:
\begin{align}
& \Big|
\Pr[\eD_\secp \big( \viewval_{\mim_\secp} \dist{\cC(v_1), \mim_\secp(\rho_\secp), \cR} (1^\secp, \tagg_\cC, \tagg_\mim), \sigma_\secp \big) = 1] \nonumber \\
& - \Pr[\eD_\secp \big( \viewval_{\mim_\secp} \dist{\cC(v_2), \mim_\secp(\rho_\secp), \cR} (1^\secp, \tagg_\cC, \tagg_\mim), \sigma_\secp \big) = 1] \Big|
 = \eta(\secp)
\end{align}
\end{definition}

We will also consider a more general setting where the \mim interacts with polynomially many committers in the left session, and a single honest receiver in the right session. 
For any polynomial $n = n(\secp)$ number of left sessions, we will let $\viewval_{\mim_\lambda} \dist{\cC(\{v_i\}_{i \in [n]}), \mim(\rho_\lambda), \cR} (1^\lambda, \tagg_\cC, \tagg_\mim)$ denote a random variable that describes the value $v'$ committed by the \mim in the right session, jointly with the view of the \mim in the full (both left and right sessions) experiment. 

\begin{definition}
[Many-one Quantum-Secure Non-Malleable Commitments with respect to Commitment] \label{def:many-nmc}
For any $\ell = \ell(\secp), p = p(\secp)$ and $n = n(\secp)$, a commitment scheme $\dist{\cC, \cR}$ is said to be quantum secure many-one non-malleable with respect to commitment for tags in $[\ell]$ if 
for every pair of tuples $(\{v_i^1\}_{i \in [n]}), (\{v_i^2\}_{i \in [n]}) \in \{0, 1\}^{2np(\secp)}$, for every quantum polynomial-size $\mim = \{\mim_\secp, \rho_\secp\}_{\secp \in \bbN}$ and every quantum polynomial-size distinguisher $\eD = \{\eD_\secp, \sigma_\secp\}_{\secp \in \bbN}$, there exists a negligible function $\eta(\cdot)$ such that for all large enough $\secp \in \bbN$, and for all $(\{\tagg_i^\cC\}_{i \in [n]}), \tagg^\mim$ where each tag is in $[\ell]$ such that $\tagg^\mim \not\in \{\tagg_i^\cC\}_{i \in [n]}$, the following holds:
\begin{align}
& \Big|
\Pr[\eD_\secp \big( \viewval_{\mim_\secp} \dist{\cC(\{v_i^1\}_{i \in [n]}), \mim_\secp(\rho_\secp), \cR} (1^\secp, \tagg_\cC, \tagg_\mim), \sigma_\secp \big) = 1] \nonumber \\
& - \Pr[\eD_\secp \big( \viewval_{\mim_\secp} \dist{\cC(\{v_i^2\}_{i \in [n]}), \mim_\secp(\rho_\secp), \cR} (1^\secp, \tagg_\cC, \tagg_\mim), \sigma_\secp \big) = 1] \Big|
 = \eta(\secp)
\end{align}
\end{definition}

\subsection{Non-Malleable Commitments for Small Tags}

First, we provide an overview of our scheme for tags in $[N]$ where $N = \secp^{\mathsf{ilog}(c+1,\secp)}$. We will assume non-interactive perfectly binding commitments and two-message SFE which can be broken with advantage at most $\negl(\secp^{\mathsf{ilog}(c,\secp)})$ by polynomial size quantum circuits.
Recall that as discussed in the technical overview, we will have the committer and receiver establish an erasure channel via a two-party input-hiding SFE. This channel will transmit the committer's value with probability $\epsilon$, depending on their $\tagg$.
Here, we discuss our construction in more detail.

The committer on input $m \in \{0,1\}^{p(\secp)}$ sends a perfectly binding, computationally hiding commitment to $m$, denoted by $\COM(m)$. 
Next, the committer and receiver run an SFE execution, where the receiver input is a uniformly random $r_1$ and committer input is $m$ along with uniformly random $s_1$ that is of the same length as $r_1$,
and $m$ is transmitted to the receiver if and only if $s_1 = r_1$. 
The length of $r_1$ and $s_1$ is carefully chosen so that the probability that they are equal is $\eta^{-\tagg}$, where $\eta = \secp^{\mathsf{ilog}(c+1,\secp)}$ is a small superpolynomial value.
Additionally, the SFE scheme is such that evaluations of agreeing circuits are (subexponentially) statistically close.
This essentially means that no matter how a malicious committer or receiver may behave, the message $m$ is revealed with probability close to $\eta^{-\tagg}$ (with an error of $\negl(\eta^{\tagg})$).

Now, let us consider a setting where the \mim uses $\tagg_\mim$ and honest committer uses $\tagg_\cC$ such that $\tagg_\mim < \tagg_\cC$.
In this case, the SFE statistically hides the committed message except with probability roughly $\eta^{-\tagg_\cC}$, and on the other hand, the \mim's message is revealed with probability roughly $\eta^{-\tagg_\mim}$, which is greater than $\eta^{-\tagg_\cC}$. 
Intuitively, this means that any \mim that tries to maul or copy the committed message cannot succeed, at least in executions where the \mim's value was revealed but the honest committer's was not. 
We generalize this to all transcripts by relying on the fact that no \mim can actually tell whether the \mim's value was revealed, 
and therefore cannot behave any differently in transcripts where extraction occured vs where it didn't.

Formally, we will prove that the joint distribution $\cV_1$ of the \mim's view and committed value when the honest commitment is to $v_1$, is indistinguishable from the joint distribution $\cV_2$ when the honest commitment is to $v_2$. This is done as follows.
\begin{itemize}
\item First, we use the input-hiding property of SFE to argue that any distinguisher $\eD$ that distinguishes $\cV_1$ from $\cV_2$, {\em must also distinguish these distributions when restricted to executions where $s_1 = r_1$ in the right execution} where the \mim is the committer. We prove that if this is not the case, then $\eD$ can be used to guess the input $r_1$ of the honest receiver in the right execution, contradicting the input-hiding property of SFE.
\item Once this is established, we restrict ourselves to transcripts where $s_1 = r_1$ in the right execution.
\item We rely on our setting of parameters to ensure that the transcripts where $s_1 = r_1$ in the left execution can only form a negligible fraction of all transcripts where $s_1 = r_1$ in the right execution.
\item Roughly, this means that for an overwhelming fraction of transcripts where $s_1 = r_1$ in the right SFE execution, the left SFE execution perfectly erases the honest committer's message.
\item As a result, any $\eD$ that distinguishes between the distributions $\cV_1$ and $\cV_2$, also distinguishes between these distributions when restricted to $s_1 = r_1$ in the right execution. We note that conditioned on $s_1 = r_1$, the \mim's message can be efficiently extracted, and therefore $\eD$ can be used to carefully break the (super-polynomial) hiding of the commitment $\COM$.
\end{itemize}
This completes a sketch of our argument when $\tagg_\mim < \tagg_\cC$. 
In case $\tagg_\mim > \tagg_\cC$, this argument does not go through, since the honest committer's message is revealed with probability that is larger than the \mim's message.
To deal with this situation, we append another sequential instance of SFE to our commitment, where the probability of extraction varies as a function of $2N - \tagg$, instead of as a function of $\tagg$. This means that a committer with $\tagg_\cC$ will run two instances of SFE, one which transmits the committed message with probability $\eta^{-\tagg_\cC}$, and another that transmits it with probability $\eta^{2N - \tagg_\cC}$.
Now, for $\tagg_\mim \neq \tagg_\cC$, in at least one of these sessions, the probability that the \mim's message is revealed will be larger than the probability that the committer's message is revealed. 
Moreover, since all these probabilities of revealing messages are negligible, the other session will not reveal the committer's message except with negligible probability, and therefore, we can switch to a hybrid where the other session never outputs the committer's message.
Finally, since proving security against synchronous \mim adversaries suffices for our applications, we only focus on formally proving synchronous security here, but we suspect that similar arguments would suffice to prove security of our construction against non-synchronous adversaries.

\subsubsection{Construction}
\paragraph{Ingredients and notation:} We will assume the existence of
\begin{itemize}
\item
  A non-interactive perfectly-binding quantum computationally hiding commitment scheme $\COM$ where there exists a constant $c_1 > 0$ s.t. no QPT adversary has advantage better than $\negl(\secp^{\mathsf{ilog}(c_1,\secp)})$ in the hiding game.
\item
 A two-message SFE satisfying Definition \ref{def:sfe2}. This means that there exists a constant $c_2 > 0$ s.t. no QPT adversary has advantage better than $\negl(\secp^{\mathsf{ilog}(c_2,\secp)})$ in the quantum input privacy game.
 \item A quantum-secure zero-knowledge argument for NP ($\zk.\zkP,\zk.\zkV$). (We do not require multi-verifier zero-knowledge for this section.)
\end{itemize} 
Let $c = \mathsf{max}(c_1, c_2), \eta = \secp^{\mathsf{ilog}(c+1, \secp)}$, and $N = \mathsf{ilog}(c+1, \secp)$. We describe the protocol for tags or identities in $[N]$ in \proref{fig:basic_nmcom}. Also, define the language
\begin{gather*}
 \mathcal{L} = \left\{\left(\begin{array}{c}\mathsf{c}, \ct_1, \ct_2, \ct_1', \\\ct_2', \tagg, N\end{array}\right): \exists \left(\begin{array}{c}m, r, s_1, s_2, \\ u_1, u_2\end{array}\right) \text{ s.t. } \begin{array}{l}
 |s_1| = \tagg \cdot (\log \eta),\\ |s_2| = (2N - \tagg) \cdot (\log \eta),\\
\mathsf{c} = \COM(1^\secp, m; r),\\
 \ct_1 = \sfeeval (\CC{\mathsf{Id}(\cdot)}{s_1}{(m||r)}, \ct_1'; u_1 ),\\ \ct_2 = \sfeeval (\CC{\mathsf{Id}(\cdot)}{s_2}{(m||r)}, \ct_2'; u_2 ) \end{array}\right\}.
 \end{gather*}
where $\mathsf{Id}(\cdot)$ denotes the identity function.

\protocol
{\proref{fig:basic_nmcom}}
{A constant round non-malleable commitment for tags in $[N]$, where $N = \mathsf{ilog}(c+1,\lambda)$.}
{fig:basic_nmcom}
{
\begin{description}
    \item[Common Input:] $1^\secp$ and a $\tagg \in [N]$.
    Set $t_1 = \tagg \cdot (\log \eta)$, and $t_2 = (2N - \tagg) \cdot (\log \eta)$.
    \item[$\eC$'s Input:] A message $m \in \{0,1\}^{p(\secp)}$.
    \item[Commit Stage:]
\end{description}

\begin{enumerate}
\item
$\eC$ samples $r \leftarrow U_\secp$, and sends $\mathsf{c}_1 = \COM(1^\secp,m;r)$.

\item 
$\eR$ samples $\dk_1 \leftarrow \sfegen(1^\secp), r_1 \leftarrow U_{t_1}$, and 
 sends $\ct_{1,\eR} \leftarrow \sfeenc_{\dk_1}(r_1)$.
\item $\eC$ samples $s_1 \leftarrow U_{t_1}$, $u_1 \leftarrow U_\secp$ and
sends $\ct_{1} = \sfeeval \Big(\CC{\mathsf{Id}(\cdot)}{s_1}{(m||r)}, \ct_{1,\eR}; u_1 \Big)$ and $s_1$, where $\mathsf{Id}(\cdot)$ is the identity function.

\item 
$\eR$ samples $\dk_2 \leftarrow \sfegen(1^\secp), r_2 \leftarrow U_{t_2}$, and 
 sends $\ct_{2,\eR} \leftarrow \sfeenc_{\dk_2}(r_2)$.
\item $\eC$ samples $s_2 \leftarrow U_{t_2}$, $u_2 \leftarrow U_\secp$ and sends $\ct_2 \leftarrow \sfeeval \Big(\CC{\mathsf{Id}(\cdot)}{s_2}{(m||r)}, \ct_{2,\eR}; u_2 \Big)$ and $s_2$, where $\mathsf{Id}(\cdot)$ is the identity function.

\item $\eC$ runs $\zk.\zkP(x,w)$ and $\eR$ runs $\zk.\zkV(x)$ in an execution of $\zk$ (with common input $(1^\secp)$) for language $\mathcal{L}$ (defined above), where $x = (\mathsf{c}_1, \ct_1, \ct_2, \ct_{1,\eR}, \ct_{2,\eR}, \tagg, N)$ and $w = (m, r, s_1, s_2,u_1,u_2)$.
\end{enumerate}
}

\subsubsection{Analysis}
In the reveal stage, the committer outputs $(m, r)$ and the receiver accepts the decommitment if $\mathsf{c}_1 = \COM(1^\secp,m;r)$.
Perfect binding follows due to the perfect binding property of $\COM$, and hiding follows by non-malleability, which we formally prove below.
\begin{lemma}
\proref{fig:basic_nmcom} is a non-malleable commitment according to \defref{def:nmc} for tags in $[N]$.
\end{lemma}
\begin{proof}
It suffices to show that for every $v_1, v_2 \in \{0,1\}^{2p(\secp)}$ and every $\mathsf{QPT}$ $\mim = \{\mim_\secp, \rho_\secp\}$ and $\eD = \{\eD_\secp, \rho_\secp\}$, there exists a negligible function $\eta(\cdot)$ such that for large enough $\secp \in \mathbb{N}$, for all $\tagg_{\cC}, \tagg_\mim \in [\ell]$ where $\tagg_{\cC} \neq \tagg_\mim$, the following holds.
\begin{align*}
& \Big|
\Pr[\eD_\secp \big( \viewval_{\mim_\secp} \dist{\cC(v_1), \mim_\secp(\rho_\secp), \cR} (1^\secp, \tagg_\cC, \tagg_\mim), \sigma_\secp \big) = 1] \nonumber \\
& - \Pr[\eD_\secp \big( \viewval_{\mim_\secp} \dist{\cC(v_2), \mim_\secp(\rho_\secp), \cR} (1^\secp, \tagg_\cC, \tagg_\mim), \sigma_\secp \big) = 1] \Big|
 = \eta(\secp)
\end{align*}
To that end,
we define the distributions 
$$\{\hyb_{v_1,\sigma_\secp} := \big( \viewval_{\mim_\secp} \dist{\cC(v_1), \mim_\secp(\rho_\secp), \cR} (1^\secp, \tagg_\cC, \tagg_\mim), \sigma_\secp \big)\}_{\secp \in \bbN},$$
$$\{\hyb_{v_2,\sigma_\secp} := \big( \viewval_{\mim_\secp} \dist{\cC(v_2), \mim_\secp(\rho_\secp), \cR} (1^\secp, \tagg_\cC, \tagg_\mim), \sigma_\secp \big)\}_{\secp \in \bbN}$$
We also define the following collections of random variables (each indexed by $\secp$). Each is defined w.r.t. a (fixed) adversary $\mim = \{\mim_\secp, \rho_\secp\}_{\secp \in \bbN}$, but we sometimes drop this adversary from notation for convenience. We will also sometimes condition on the \mim aborting in the left execution (where it acts as receiver). By this, we will refer to an execution where the \mim sends a message that causes an honest party to abort.

\begin{itemize}
    \item Let $\mathsf{Pr}^{\mathsf{Abort}}_{x}$ be the probability that $\mim$ aborts in $\hyb_{x,\sigma_\secp}$, where $x \in \{v_1, v_2\}$.
    \item Let $\hyb_{x,\sigma_\secp}^{\mathsf{No} \ \mathsf{Abort}}$ be the distribution $\hyb_{x,\sigma_\secp}$ conditioned on there not being an abort.
    \item Let $\hyb_{x,\sigma_\secp}^{\mathsf{Abort}}$ be the distribution $\hyb_{x,\sigma_\secp}$ conditioned on there being an abort. Note that by definition, in this distribution, the value committed by the \mim is always set to $\bot$.
\end{itemize}

The following distributions will not be used explicitly in the hybrids, but will be convenient to define for the proof.

\begin{itemize}
    \item Let $\hyb_{x,\sigma_\secp,\bot}$ be the distribution $\hyb_{x,\sigma_\secp}$ except whenever an abort occurs, the distribution outputs $\bot$.
    \item Let $\hyb_{x,\sigma_\secp,\bot}^{\mathsf{Abort}}$ be the distribution $\hyb_{x,\sigma_\secp}$ except whenever an abort does not occur, the distribution outputs $\bot$.
\end{itemize}

We show that $\{\hyb_{v_1,\sigma_\secp}\}_{\secp \in \bbN} \approx_c \{\hyb_{v_2,\sigma_\secp}\}_{\secp \in \bbN}$ via a sequence of hybrids. 
In particular, we show that 
\begin{align*}
    \hyb_{v_1,\sigma_\secp} &\substack{(1) \\ \equiv \\ \ } (1-\mathsf{Pr}^\mathsf{Abort}_{v_1})
    \hyb_{v_1, \sigma_\secp}^{\mathsf{No} \ \mathsf{Abort}} + (\mathsf{Pr}^\mathsf{Abort}_{v_1})
    \hyb_{v_1, \sigma_\secp}^{\mathsf{Abort}}\\
    &\substack{(2) \\ \approx_s \\ \ } (1-\mathsf{Pr}^\mathsf{Abort}_{v_2})
    \hyb_{v_1, \sigma_\secp}^{\mathsf{No} \ \mathsf{Abort}} + (\mathsf{Pr}^\mathsf{Abort}_{v_2})
    \hyb_{v_1, \sigma_\secp}^{\mathsf{Abort}}\\
    &\substack{(3) \\ \approx_c \\ \ } (1-\mathsf{Pr}^\mathsf{Abort}_{v_2})
    \hyb_{v_2, \sigma_\secp}^{\mathsf{No} \ \mathsf{Abort}} + (\mathsf{Pr}^\mathsf{Abort}_{v_2})
    \hyb_{v_1, \sigma_\secp}^{\mathsf{Abort}}\\
    &\substack{(4) \\ \approx_c \\ \ } (1-\mathsf{Pr}^\mathsf{Abort}_{v_2})
    \hyb_{v_2, \sigma_\secp}^{\mathsf{No} \ \mathsf{Abort}} + (\mathsf{Pr}^\mathsf{Abort}_{v_2})
    \hyb_{v_2, \sigma_\secp}^{\mathsf{Abort}}\\
    &\substack{(5) \\ \equiv \\ \ }\hyb_{v_2,\sigma_\secp},
\end{align*}
where 
\begin{enumerate}
    \item The equalities $(1)$ and $(5)$ follow by definition.
    \item The indistinguishability $(2)$ follows as a corollary of~\cref{clm:nm-base-na}. Indeed, $\hyb_{v_1, \sigma_\secp, \bot} \approx_c \hyb_{v_2, \sigma_\secp, \bot}$ in particular implies that the difference in the probability that the $\mim$ aborts in both executions is negligible.
    \item The indistinguishability $(3)$ follows as a corollary of~\cref{clm:nm-base-na}. This can be seen by considering two cases. First, if the probability that the $\mim$ aborts in $\hyb_{v_1,\sigma_\secp}$ is negligible, then $\hyb_{v_1,\sigma_\secp}^{\mathsf{No} \ \mathsf{Abort}} \approx_c \hyb_{v_2,\sigma_\secp}^{\mathsf{No} \ \mathsf{Abort}}$ directly follows from~\cref{clm:nm-base-na}, and the indistinguishability follows. Otherwise, this probability is non-negligible, meaning that $\hyb_{v_1,\sigma_\secp}^{\mathsf{Abort}}$ is efficiently sampleable. Thus, a reduction to~\cref{clm:nm-base-na} can sample from the distribution $\hyb_{v_1,\sigma_\secp}^{\mathsf{Abort}}$ whenever it receives $\bot$ from its challenger.\footnote{A more formal analysis of this can be found in~\cite[Lemma~3.2]{BS20}.}
    \item The indistinguishability $(4)$ follows as a corollary
    of~\cref{clm:nm-base-abort} via a similar analysis as the last step.
\end{enumerate}

\begin{claim}
\label{clm:nm-base-na}
$$\{\hyb_{v_1, \sigma_\secp, \bot}\}_{\secp \in \bbN} \approx_c \{\hyb_{v_2, \sigma_\secp, \bot}\}_{\secp \in \bbN}$$
\end{claim}
\begin{proof}
We will prove this claim via the following sequence of hybrids. We set some notation before defining these hybrids. We will set $t_1 = \tagg_\cC \cdot (\log \eta)$ and $t_2 = (2N - \tagg_\cC) \cdot (\log \eta)$.
We also set $t_1' = \tagg_\mim \cdot (\log \eta)$ and $t_2' = (2N - \tagg_\mim) \cdot (\log \eta)$.
As a general rule, when refering to some protocol variable $y$ in the left execution, we will use the variable as is (and denote it by $y$), and in the right execution, we will denote this variable by $y'$.\\


\noindent We let $\hyb_{v_1, \sigma_\secp, \bot} = \hyb_0$.\\

\noindent $\hyb_1:$ In this hybrid, the challenger executes the simulator $\pibs.\zkSimNoAbort(1^\secp, x_\secp, \zkV^*_\secp, \sigma_\secp^{(x_\secp)})$\footnote{Note that we drop the input $I$ since there is only one verifier in this setting.} for $\pibs$ on $\zkV^*_\secp$, which denotes a wrapper around the portion of the \mim that participates in Step 6 of the protocol, and an instance-advice distribution $(x_\secp, \sigma_\secp^{(x_\secp)})$ defined as follows:
\begin{itemize}
    \item Set the state of $\mim_\secp$ to $\rho_\secp$.
    \item Execute Steps 1-5 of the protocol the same way as in the experiment $\hyb_{v_1, \sigma_\secp, \bot}$, and set $(x, w, \cL)$ according to \proref{fig:basic_nmcom} on behalf of $\mathcal{C}$.
    \item Let $\sigma_\secp^{(x_\secp)}$ denote the joint distribution of the protocol transcript, the state of the \mim at the end of Step 5, and the value $v'$ committed by the \mim in Step 1.
\end{itemize}
If there is an abort during sampling, then output $\bot$.
Otherwise, the output of this hybrid is the output of $\pibs.\zkSimNoAbort$. 
By \cref{claim:ind-non-abort}, $$\hyb_0 \approx_c \hyb_1.$$
\noindent $\hyb_2:$ This is identical to $\hyb_1$ except the following change.

In Step 3, $\eC$ sends $\ct_1 = \sfeeval \Big(\CC{\mathsf{Id}(\cdot)}{s_1}{(0^{p(\secp)+\secp})}, \ct_{1,\eR}; u_1 \Big)$ and $s_1$. 
Here $(x_\secp, \sigma_\secp^{(x_\secp)})$ and $\zkV^*_\secp$ are defined identically to $\hyb_1$ except with the updated $\ct_1$ from Step 3, and the simulator $\pibs.\zkSimNoAbort(1^\secp, x_\secp, \zkV^*_\secp, \sigma_\secp^{(x_\secp)})$ is executed.
If there is an abort during sampling, then output $\bot$.
Otherwise, the output of this hybrid is the output of $\pibs.\zkSimNoAbort$.
We prove in Claim~\ref{claim:ab}, that $$\hyb_1 \approx_s \hyb_2.$$

\noindent $\hyb_3:$ This is identical to $\hyb_2$ except the following change.

In Step 5, $\eC$ sends $\ct_2 = \sfeeval \Big(\CC{\mathsf{Id}(\cdot)}{s_2}{(0^{p(\secp)+\secp})}, \ct_{2,\eR}; u_2 \Big)$ and $s_2$. 
Here $(x_\secp, \sigma_\secp^{(x_\secp)})$ and $\zkV^*_\secp$ are defined identically to $\hyb_2$ except with the updated $\ct_2$ from Step 5, and the simulator $\pibs.\zkSimNoAbort(1^\secp, x_\secp, \zkV^*_\secp, \sigma_\secp^{(x_\secp)})$ is executed.
If there is an abort during sampling, then output $\bot$.
Otherwise, the output of this hybrid is the output of $\pibs.\zkSimNoAbort$.
We prove in Claim~\ref{claim:bc}, that $$\hyb_2 \approx_s \hyb_3.$$

\noindent $\hyb_4:$ This is identical to $\hyb_3$ except the following change.

In Step 1, $\eC$ sets $\mathsf{c}_1 = \COM(1^\secp,0;r)$. 
Here $(x_\secp, \sigma_\secp^{(x_\secp)})$ and $\zkV^*_\secp$ are defined identically to $\hyb_3$ except with the updated $\mathsf{c}_1$ from Step 1, and the simulator $\pibs.\zkSimNoAbort(1^\secp, x_\secp, \zkV^*_\secp, \sigma_\secp^{(x_\secp)})$ is executed.
If there is an abort during sampling, then output $\bot$.
Otherwise, the output of this hybrid is the output of $\pibs.\zkSimNoAbort$.
We prove in Claim~\ref{claim:cd}, that $$\hyb_3 \approx_c \hyb_4.$$

\begin{claim}
\label{claim:ab}
$$\Delta(\hyb_1, \hyb_2) \leq 2^{-t_1} + \negl(2^{t_1})$$
\end{claim}
\begin{proof}
Note that the output of $\sfeeval \Big(\CC{\mathsf{Id}(\cdot)}{s_1}{(m||r)}, \ct_{1,\eR}; u_1 \Big)$ is identical in both hybrids, unless $s_1 = r_1$.
Denote by $\hyb_1'$ the distribution that is identical to $\hyb_1$ except it outputs $\bot$ when $s_1 = r_1$.
Denote by $\hyb_2'$ the distribution that is identical to $\hyb_2$ except it outputs $\bot$ when $s_1 = r_1$.
Now by statistical circuit privacy, we have that there exists a constant $c > 0$ such that $\Delta(\hyb'_1, \hyb'_2) \leq 2^{-\secp^{c}}$.

Finally, note that in each one of $\hyb_1, \hyb_2, \hyb_1', \hyb_2'$, 
$$\Pr[s_1 = r_1] \leq {2^{-t_1}}.$$
Thus we have,
\begin{align*}
    \Delta(\hyb_1, \hyb_2) 
    \leq 
    \Delta(\hyb'_1, \hyb'_2) + 2 \cdot \Pr[s_1 = r_1] 
    \leq 2^{-t_1} + 2 \cdot 2^{-\secp^c} 
    \leq 2^{-t_1} + \negl(2^{t_1})
\end{align*}
where the last equation follows by our setting of $t_1$.

\end{proof}

\begin{claim}
\label{claim:bc}
$$\Delta(\hyb_2, \hyb_3) \leq 2^{-t_2} + \negl(2^{t_2})$$
\end{claim}
\begin{proof}
The proof follows nearly identically to that of \cref{claim:ab}.
\end{proof}

\begin{claim}
\label{claim:cd}
$$\hyb_3 \approx_c \hyb_4$$
\end{claim}
\begin{proof}

Throughout this proof, we will use the notation $\Pr[\mathsf{E}|\hyb]$ to refer to the probability that event $\mathsf{E}$ occurs in the output of distribution $\hyb$.

Suppose $$\Pr[\mim \text{ aborts}|\hyb_3] = 1 - \negl(\secp).$$
Then by hiding of the commitment $\COM$ $$\Pr[\mim \text{ aborts}|\hyb_4] = 1 - \negl(\secp),$$
so both hybrids output $\bot$ except with negligible probability, and are therefore computationally indistinguishable.

Thus for the rest of this proof, we will assume that there exists a polynomial $p(\cdot)$ such that:
$$\Pr[\mim \text{ does not abort }|\hyb_3] \geq \frac{1}{p(\secp)}.$$
By the hiding of the commitment $\COM$, 
$$\Pr[\mim \text{ does not abort }|\hyb_4] \geq \frac{1}{p(\secp)} - \negl(\secp).$$

For $x \in [0,4]$, we will denote by $\hyb_x^{\na}$ the distribution $\hyb_x$ conditioned on the \mim not aborting.
Recall that whenever the $\mim$ aborts, the two hybrids output $\bot$. Therefore, it suffices to prove that 
$$\hyb_3^{\na} \approx_c \hyb_4^{\na}.$$

Now, recall that the variables $r_1',s_1',t_1'$ and so on refer to the \emph{right} execution in each experiment. For $r_1' \leftarrow U_{t_1'}$ and $s_1'$ sampled independently of $r_1'$, 
we have that $\Pr[s_1' = r_1'] = 2^{-t_1'}.$ Then, by quantum input privacy of SFE according to \defref{def:sfe2}, there exists a negligible function $\mu(\cdot)$ such that for any $x \in [0,4]$, $r_1' \leftarrow U_{t_1'}$ and $s_1'$ chosen by the \mim in $\hyb_x$,
\begin{equation}
\label{eq:abcd}
{2^{-t_1'}} - {\mu(\secp^{\mathsf{ilog}(c,\secp)})}\leq
\Pr[s_1' = r_1'|\hyb_x^{\na}] \leq 
{2^{-t_1'}} + {\mu(\secp^{\mathsf{ilog}(c,\secp)})}.
\end{equation}

Similarly, by quantum input privacy of SFE according to \defref{def:sfe2}, there exists a negligible function $\mu'(\cdot)$ such that for any $x \in [0,4]$, $r_2' \leftarrow U_{t_2'}$ and $s_2'$ output by the \mim in $\hyb_x$,
\begin{equation}
\label{eq:bcd}
{2^{-t_2'}} - {\mu'(\secp^{\mathsf{ilog}(c,\secp)})} \leq
\Pr[s_2' = r_2'|\hyb_x^{\na}] \leq 
{2^{-t_2'}} + {\mu'(\secp^{\mathsf{ilog}(c,\secp)})}
\end{equation}
\noindent Soundness of the ZK argument, together with setting $x = 0$ in equations~(\ref{eq:abcd}) and~(\ref{eq:bcd}) implies that for $i \in [2]$,
\begin{align}
\label{eq:efg}
& \Pr\Big[\sfedec_{\dk_i'}(\ct_i') \rightarrow (m',r') \text{ s.t. } \mathsf{c}'_1 = \COM(1^\secp, m'; r') \wedge (s_i' = r_i') \Big| \hyb_0^{\na} \Big] \nonumber \\
& \geq (1 - \negl(\secp)) \cdot \Big( {2^{-t_i'}} - {\mu(\secp^{\mathsf{ilog}(c,\secp)})} \Big).
\end{align}
Note that for $i \in [2]$, $\sfedec_{\dk_i'}(\ct_i') \rightarrow (m',r') \text{ s.t. } \mathsf{c}'_1 = \COM(1^\secp, m'; r') \wedge (s_i' = r_i')$ can be efficiently checked by a challenger that samples $\dk_i'$ and $r_i'$. 

Therefore, by combining \cref{claim:ind-non-abort} with equation~(\ref{eq:efg}), we have that for $i \in [2]$, 
\begin{align}
\label{eq:xy}
& \Pr\Big[\sfedec_{\dk_i'}(\ct'_i) \rightarrow (m',r') \text{ s.t. } \mathsf{c}'_1 = \COM(1^\secp, m'; r') \wedge (s_i' = r_i') \Big|  \hyb_1^{\na} \Big] \nonumber \\
& \geq (1 - \negl(\secp) - \negl(\secp)) \cdot  \Big( {2^{-t_i'}} - {\secp^{\mathsf{ilog}(c,\secp)}} \Big) \nonumber \\
& \geq (1 - \negl(\secp) ) \cdot  \Big( {2^{-t_i'}} - {\secp^{\mathsf{ilog}(c,\secp)}} \Big)
\end{align}
where the previous equation, for $i \in [2]$, follows by considering a non-uniform reduction to Claim~\ref{claim:ind-non-abort} that fixes any transcript where $\sfedec_{\dk'_i}(\ct'_i) \rightarrow (m',r') \text{ s.t. } \mathsf{c}'_1 = \COM(1^\secp, m'; r') \wedge (r'_i = s'_i)$.


Next, we split our analysis into two cases. Depending on whether $\tagg_\cC$ is greater or smaller than $\tagg_\mim$, one of the two cases will always be true.
\begin{itemize}
\item {\bf Case 1: $\tagg_\cC > \tagg_\mim$.}
In this case, $t_1 = \tagg_\cC \cdot (\log \eta)$, $t_1' = \tagg_\mim \cdot (\log \eta)$.

Now combining \cref{claim:ab} and equation~(\ref{eq:xy}) with $i$ set to $1$, implies:
\begin{align}
\label{eq:z}
& \Pr\Big[\sfedec_{\dk'_1}(\ct'_1) \rightarrow (m',r') \text{ s.t. } \mathsf{c}'_1 = \COM(1^\secp, m'; r') \wedge (s_1' = r_1') \Big| \hyb_2^{\na} \Big] \nonumber \\
& \geq (1 - \negl(\secp))  \cdot \Big( 2^{-t'_1} - \mu({\secp^{\mathsf{ilog}(c,\secp)}}) \Big) - 2^{-t_1} - \negl(2^{t_1})
\end{align}
Next, we will carefully combine equation~(\ref{eq:z}) with \cref{claim:bc}. 
First, we note that the check 
$\sfedec_{\dk'_1}(\ct'_1) \rightarrow (m',r') \text{ s.t. } \mathsf{c}'_1 = \COM(1^\secp, m'; r') \wedge (r'_1 = s'_1)$ is performed before Step 5.
Additionally, the only difference between $\hyb_2^{\na}$ and $\hyb_3^{\na}$ is in Step 5. 
As a result, for {\em every} (fixed) prefix of the transcript until Step 4, the distribution of Steps 5 and 6 generated according to $\hyb_2^{\na}$ is at most $2^{-t_2} + \negl(2^{-t_2})$-far from their distribution generated according to $\hyb_3^{\na}$. 
This implies:
\begin{align}
\label{eq:f}
& \Pr\Big[\sfedec_{\dk'_1}(\ct'_1) \rightarrow (m',r') \text{ s.t. } \mathsf{c}'_1 = \COM(1^\secp, m'; r') \wedge (s_1' = r_1') \Big| \hyb_3^{\na} \Big] \nonumber \\
& \geq (1 - {2^{-t_2}} - \negl(2^{t_2}))  \cdot
\Bigg( (1 - \negl(\secp)) \cdot \Big( 2^{-t'_1} - \mu({\secp^{\mathsf{ilog}(c,\secp)}}) \Big) - 2^{-t_1} - \negl(2^{t_1}) \Bigg) \nonumber \\
& \geq (1 - \negl(\secp) - {2^{-t_2}} - \negl(2^{t_2}))  \cdot \Big( 2^{-t'_1} - \mu({\secp^{\mathsf{ilog}(c,\secp)}}) \Big) - 2^{-t_1} + 2^{-t_1 - t_2} - \negl(2^{t_1}) \nonumber \\
& \geq (1 - \negl(\secp) )  \cdot \Big( 2^{-t'_1} - \mu({\secp^{\mathsf{ilog}(c,\secp)}}) \Big) - 2^{-t_1} - \negl(2^{t_1})
\end{align}
Combining equation~(\ref{eq:f}) with the $\negl(\secp^{\mathsf{ilog}(c, \secp)})$- hiding of the commitment $\COM$,
\begin{align}
\label{eq:g}
& \Pr\Big[\sfedec_{\dk'_1}(\ct'_1) \rightarrow (m',r') \text{ s.t. } \mathsf{c}'_1 = \COM(1^\secp, m'; r') \wedge (s_1' = r_1') \Big| \hyb_4^{\na} \Big] \nonumber \\
& \geq (1 - \negl(\secp))  \cdot \Big( 2^{-t'_1} - \mu({\secp^{\mathsf{ilog}(c,\secp)}}) \Big) - 2^{-t_1} - \negl(2^{t_1}) - \negl(\secp^{\mathsf{ilog}(c, \secp)}).
\end{align}
Combining equations~(\ref{eq:f}), (\ref{eq:g}) with equation~(\ref{eq:abcd}), for $x \in [3,4]$,
\begin{align}
\label{eq:h}
& \Pr\Big[\sfedec_{\dk'_1}(\ct'_1) \rightarrow (m',r') \text{ s.t. } \mathsf{c}'_1 = \COM(1^\secp, m'; r') \Big| (s_1' = r_1'), \hyb_x^{\na} \Big]\nonumber \\
& \geq 
\frac{(1 - \negl(\secp) )  \cdot \Big( 2^{-t'_1} - \mu({\secp^{\mathsf{ilog}(c,\secp)}}) \Big) - 2^{-t_1} - \negl(2^{t_1}) - \negl(\secp^{\mathsf{ilog}(c, \secp)})}{2^{-t_1'} + \mu(\secp^{\mathsf{ilog}(c,\secp)})} \nonumber \\
& \geq 1 - \negl(\secp)
\end{align}
where the last equation follows by recalling that 
$2^{-t_1'} = \eta^{-\tagg_\mim}, 2^{-t_1} = \eta^{-\tagg_\cC}$ and
$\tagg_\cC \geq (\tagg_\mim + 1)$, which implies
$$ 2^{-t_1} = \eta^{-\tagg_\cC} \leq \eta^{-\tagg_\mim - 1} = \frac{2^{-t_1'}}{\eta}$$
for $\eta = \secp^{\mathsf{ilog}(c+1,\secp)}$.

Let us assume towards a contradiction that there exists a quantum polynomial size distinguisher $\eD$ and a polynomial $\poly$ such that for large enough $\secp \in \bbN$: 
\begin{equation}
\label{eq:l}
\Pr[\eD = 1| \hyb_3^{\na}] - \Pr[\eD = 1|\hyb_4^{\na}] \geq \frac{1}{\poly(\secp)}.
\end{equation}

By quantum input privacy of the SFE scheme, for each $x \in [3,4]$, 
$$|\Pr[\eD = 1|(s_1' \neq r_1'), \hyb_x^{\na}] - \Pr[\eD = 1|(s'_1 = r'_1), \hyb_x^{\na}]| = \negl(\secp)$$
which combined with equation~(\ref{eq:l}) implies that there exists a polynomial $q(\cdot)$ such that
\begin{equation}
\label{eq:m}
\Pr[\eD = 1|(s'_1 = r'_1),\hyb_3^{\na}] - \Pr[\eD = 1|(s'_1 = r'_1),\hyb_4^{\na}] \geq \frac{1}{q(\secp)}
\end{equation}
Equations~(\ref{eq:h}) and~(\ref{eq:m}) together imply that:
\begin{align*}
& \Pr[\eD = 1 \wedge \sfedec_{\dk_1'}(\ct'_1) \rightarrow (m',r') \text{ s.t. } \mathsf{c}'_1 = \COM(1^\secp, m'; r') | (s'_1 = r'_1), \hyb_3^{\na}] \\
& - \Pr[\eD = 1 \wedge \sfedec_{\dk'_1}(\ct'_1) \rightarrow (m',r') \text{ s.t. } \mathsf{c}'_1 = \COM(1^\secp, m'; r') | (s'_1 = r'_1), \hyb_4^{\na}]\\
& \geq \frac{1}{q(\secp)} \cdot (1 - \negl(\secp) ) \geq \frac{1}{2q(\secp)}
\end{align*}
This, combined with equation~(\ref{eq:abcd}) gives a distinguisher that distinguishes $\COM(m)$ and $\COM(0)$ with advantage at least $2^{-2t_1'}$, contradicting the hiding of $\COM$ as desired. 

\item {\bf Case 2: $\tagg_\mim > \tagg_\cC$.}
A similar analysis as in Case 1 implies that for any distinguisher $\eD$ distinguishing $\hyb_3^{\na}$ from $\hyb_4^{\na}$, there exists a polynomial $q'(\cdot)$ such that
\begin{align*}
& \Pr[\eD = 1 \wedge \sfedec_{\dk_2'}(\ct'_2) \rightarrow (m',r') \text{ s.t. } \mathsf{c}'_1 = \COM(1^\secp, m'; r') | (s'_2 = r'_2), \hyb_3^{\na}] \\
& - \Pr[\eD = 1 \wedge \sfedec_{\dk'_2}(\ct'_2) \rightarrow (m',r') \text{ s.t. } \mathsf{c}'_1 = \COM(1^\secp, m'; r') | (s'_2 = r'_2), \hyb_4^{\na}]\\
& \geq \frac{1}{q'(\secp)}
\end{align*}
This, combined with equation~(\ref{eq:bcd}) gives a distinguisher that distinguishes $\COM(m)$ and $\COM(0)$ with advantage at least $2^{-2t_2'}$, contradicting the hiding of $\COM$ as desired.
\end{itemize}
\end{proof}

\cref{clm:nm-base-na} follows by observing that in $\hyb_4$, we erased all information about either one of $v_1$ or $v_2$. Therefore one can perform the above hybrids in reverse order, while arguing indistinguishability, until one ends up with $\hyb_{v_2, \sigma_\secp, \bot}$.
\end{proof}

\begin{claim}
\label{clm:nm-base-abort}
$$\{\hyb_{v_1, \sigma_\secp, \bot}^{\mathsf{Abort}}\}_{\secp \in \bbN} \approx_c \{\hyb_{v_2, \sigma_\secp, \bot}^{\mathsf{Abort}}\}_{\secp \in \bbN}$$
\end{claim}
\begin{proof}
Note that in this case, the value committed by the \mim is always $\bot$.
Therefore, proving indistinguishability of these distributions is significantly more straightforward than in \cref{clm:nm-base-na}. 
The proof again relies on a sequence of hybrid experiments, that we define below.
Recall that when refering to some protocol variable $y$ in the left execution, we will use the variable as is (and denote it by $y$), and in the right execution, we will denote this variable by $y'$.\\


\noindent We let $\hyb_{v_1, \sigma_\secp, \bot}^{\mathsf{Abort}} = \hyb_0$.\\

\noindent $\hyb_1:$ In this hybrid, the challenger executes one run of the simulator $\pibs.\zkSimAbort(1^\secp, x_\secp, \zkV^*_\secp, \sigma_\secp^{(x_\secp)})$\footnote{Note that we drop the input $I$ since there is only one verifier in this setting.} for $\pibs$ on $\zkV^*_\secp$, which denotes a wrapper around the portion of the \mim that participates in Step 6 of the protocol, and an instance-advice distribution $(x_\secp, \sigma_\secp^{(x_\secp)})$ defined as follows:
\begin{itemize}
    \item Set the state of $\mim_\secp$ to $\rho_\secp$.
    \item Execute Steps 1-5 of the protocol the same way as in the experiment $\hyb_{v_1, \sigma_\secp, \bot}$, and set $(x, w, \cL)$ according to \proref{fig:basic_nmcom} on behalf of $\mathcal{C}$.
    \item Let $\sigma_\secp^{(x_\secp)}$ denote the joint distribution of the protocol transcript, the state of the \mim at the end of Step 5, and the value $v'$ committed by the \mim in Step 1.
\end{itemize}
If there is an abort during sampling then output the transcript generated until the abort happens and the state of the \mim. Otherwise, the output of the hybrid is the output of $\pibs.\zkSimAbort(1^\secp, x_\secp, \zkV^*_\secp, \sigma_\secp^{(x_\secp)})$. 
By \cref{claim:ind-abort}, $$\hyb_0 \approx_c \hyb_1.$$
\noindent $\hyb_2:$ This is identical to $\hyb_1$ except the following change.

In Step 3, $\eC$ sends $\ct_1 = \sfeeval \Big(\CC{\mathsf{Id}(\cdot)}{s_1}{(0^{p(\secp)+\secp})}, \ct_{1,\eR}; u_1 \Big)$ and $s_1$. 
Here $(x_\secp, \sigma_\secp^{(x_\secp)})$ and $\zkV^*_\secp$ are defined identically to $\hyb_1$ except with the updated $\ct_1$ from Step 3, and the simulator $\pibs.\zkSimAbort(1^\secp, x_\secp, \zkV^*_\secp, \sigma_\secp^{(x_\secp)})$ is executed. If there is an abort during sampling then output the transcript generated until the abort happens and the state of the \mim. Otherwise, the output of the hybrid is the output of $\pibs.\zkSimAbort(1^\secp, x_\secp, \zkV^*_\secp, \sigma_\secp^{(x_\secp)})$.
We prove in Claim~\ref{claim:abp}, that $$\hyb_1 \approx_s \hyb_2.$$

\noindent $\hyb_3:$ This is identical to $\hyb_2$ except the following change.

In Step 5, $\eC$ sends $\ct_2 = \sfeeval \Big(\CC{\mathsf{Id}(\cdot)}{s_2}{(0^{p(\secp)+\secp})}, \ct_{2,\eR}; u_2 \Big)$ and $s_2$. 
Here $(x_\secp, \sigma_\secp^{(x_\secp)})$ and $\zkV^*_\secp$ are defined identically to $\hyb_2$ except with the updated $\ct_2$ from Step 5, and the simulator $\pibs.\zkSimAbort(1^\secp, x_\secp, \zkV^*_\secp, \sigma_\secp^{(x_\secp)})$ is executed. If there is an abort during sampling then output the transcript generated until the abort happens and the state of the \mim. Otherwise, the output of the hybrid is the output of $\pibs.\zkSimAbort(1^\secp, x_\secp, \zkV^*_\secp, \sigma_\secp^{(x_\secp)})$.
We prove in Claim~\ref{claim:bcp}, that $$\hyb_2 \approx_s \hyb_3.$$

\noindent $\hyb_4:$ This is identical to $\hyb_3$ except the following change.

In Step 1, $\eC$ sets $\mathsf{c}_1 = \COM(1^\secp,0;r)$. 
Here $(x_\secp, \sigma_\secp^{(x_\secp)})$ and $\zkV^*_\secp$ are defined identically to $\hyb_3$ except with the updated $\mathsf{c}_1$ from Step 1, and the simulator $\pibs.\zkSimAbort(1^\secp, x_\secp, \zkV^*_\secp, \sigma_\secp^{(x_\secp)})$ is executed. If there is an abort during sampling then output the transcript generated until the abort happens and the state of the \mim. Otherwise, the output of the hybrid is the output of $\pibs.\zkSimAbort(1^\secp, x_\secp, \zkV^*_\secp, \sigma_\secp^{(x_\secp)})$.
We prove in Claim~\ref{claim:cdp}, that $$\hyb_3 \approx_c \hyb_4.$$

\begin{claim}
\label{claim:abp}
$$\Delta(\hyb_1, \hyb_2) \leq 2^{-t_1} + \negl(2^{t_1})$$
\end{claim}
\begin{proof}
Note that the output of $\sfeeval \Big(\CC{\mathsf{Id}(\cdot)}{s_1}{(m||r)}, \ct_{1,\eR}; u_1 \Big)$ is identical in both hybrids, unless $s_1 = r_1$.
Denote by $\hyb_1'$ the distribution that is identical to $\hyb_1$ except it outputs $\bot$ when $s_1 = r_1$.
Denote by $\hyb_2'$ the distribution that is identical to $\hyb_2$ except it outputs $\bot$ when $s_1 = r_1$.
Now by statistical circuit privacy, we have that there exists a constant $c > 0$ such that $\Delta(\hyb'_1, \hyb'_2) \leq 2^{-\secp^{c}}$.

Finally, note that in each one of $\hyb_1, \hyb_2, \hyb_1', \hyb_2'$, 
$$\Pr[s_1 = r_1] \leq {2^{-t_1}}.$$
Thus we have,
\begin{align*}
    \Delta(\hyb_1, \hyb_2) 
    \leq 
    \Delta(\hyb'_1, \hyb'_2) + 2 \cdot \Pr[s_1 = r_1] 
    \leq 2^{-t_1} + 2 \cdot 2^{-\secp^c} 
    \leq 2^{-t_1} + \negl(2^{t_1})
\end{align*}
where the last equation follows by our setting of $t_1$.

\end{proof}

\begin{claim}
\label{claim:bcp}
$$\Delta(\hyb_2, \hyb_3) \leq 2^{-t_2} + \negl(2^{t_2})$$
\end{claim}
\begin{proof}
The proof follows nearly identically to that of \cref{claim:ab}.
\end{proof}

\begin{claim}
\label{claim:cdp}
$$\hyb_3 \approx_c \hyb_4$$
\end{claim}
\begin{proof}
Recall that both hybrids only output transcripts, the value committed by the \mim, and the \mim's state when the \mim aborts at some point. Otherwise both hybrids output $\bot$. 
As such, the value committed by the \mim in every transcript in the distributions $\hyb_3$, and $\hyb_4$, is $\bot$. 

It remains to prove that the joint distribution of the transcript and state of the \mim in $\hyb_3$ is indistinguishable from $\hyb_4$, which follows immediately by the hiding of $\COM$.
\end{proof}

\cref{clm:nm-base-abort} follows by observing that in $\hyb_4$, we erased all information about either one of $v_1$ or $v_2$. Therefore one can perform the above hybrids in reverse order, while arguing indistinguishability, until one ends up with $\hyb_{v_2, \sigma_\secp, \bot}^{\mathsf{Abort}}$.
\end{proof}

This concludes the proof of the lemma.
\end{proof}

\subsection{Tag Amplification}
\label{sec:tag-amp}

Beginning with quantum-secure non-malleable commitments that support tags in $[t]$, we describe how to obtain quantum-secure non-malleable commitments that support tags in $[2^{t/2}]$, for any $t \leq \poly(\secp)$. We will apply this compiler recursively a constant number of times to our base construction to obtain a scheme for tags in $[2^\secp]$.

The protocol itself follows nearly identically along the lines of existing tag amplification compilers in the literature~\cite{FOCS:Wee10,EC:PasWee10}. 
Each larger $\tagg$ (in $2^{t/2}$) is encoded into a set $S_\tagg$ of $t$ small tags, and this set satisfies the following property.
For each pair $\tagg, \tagg' \in 2^{t/2}$ where $\tagg \neq \tagg'$, there exists an element in the set $S_{\tagg'}$ that does not lie in the set $S_{\tagg}$. 
Now the committer on input a message $m$ and tag $\tagg$, generates $t$ commitments,
where the $i^{th}$ commitment commits via the underlying non-malleable commitment to the the message $m$ using as tag the $i^{th}$ member of the set $S_\tagg$.
The committer then proves to the receiver (in zero-knowledge) that all the commitments were correctly generated according to protocol specifications.

Now, the property of the tag encoding scheme guarantees that the \mim will always end up using {\em at least one small tag} that is different from all small tags used by the honest committer.
As such, we can argue that the value committed by the \mim using this differing tag is independent of the honest committer's input. 

We formally prove this via a hybrid argument, where we first rely on soundness of the ZK argument to argue that we can ``focus'' solely on the values committed by the \mim using the differing tag. 
Next, we simulate the ZK argument on behalf of the honest committer. Then we modify the values committed via all the honest small tags, while arguing that the value committed by the \mim under the differing tag does not change.

Finally, we point out an interesting feature of our amplification proof: 
we split the use of the ZK simulator into two cases: one simulator that only outputs non-aborting views (and otherwise outputs $\bot$), and a separate one that only outputs aborting views (and otherwise outputs $\bot$). 

\subsubsection{Construction}
We describe our compiler that converts a tag-based non-malleable commitment scheme for tags in $t(\secp)$, where $t(\secp) \leq \poly(\secp)$, into one that supports tags in $2^{t/2}$, while adding only a constant number of rounds. This is formalized in \proref{fig:tag_amplification_nmcom}.
We will let $\nmcsmall$ denote a non-malleable commitment for tags in $t(\secp)$, and we denote the message length by $p(\secp)$.\\

\noindent \textbf{Ingredients and notation:} We will assume the existence of a quantum-secure zero-knowledge argument for NP  ($\zk.\zkP,\zk.\zkV$). (We do not require multi-verifier zero-knowledge for this section.)



\protocol
{\proref{fig:tag_amplification_nmcom}}
{A constant round non-malleable commitment for tags in $[2^{t/2}]$.}
{fig:tag_amplification_nmcom}
{
\begin{description}
    \item[Common Input:] Security parameter $1^\secp$, $\tagg \in [2^{t/2}]$ represented as $\{ \tagg_i \}_{i \in [t/2]}$. Here $\tagg_i = i || \tagg[i]$ where $\tagg[i]$ denotes the $i^{th}$ bit of $\tagg$. 
    \item[$\mathsf{C}$'s Input:] A string $m \in \{0, 1\}^{p(\secp)}$.
\end{description}

\textbf{Commit Stage:}

\begin{itemize}
    
    \item {\bf Stage 1:} In parallel, for all $i \in [k]$,
            $\mathsf{C}$ runs $\nmcsmall.\mathsf{C}(m; r_{i, C})$ and $\mathsf{R}$ runs $\mathsf{R}(r_{i,R})$ with common input $\tagg_i$. Let $c_i$ and $\phi_i$
            denote the set of all messages generated by $\mathsf{C}$ and $\mathsf{R}$ respectively in the $i^{th}$ parallel execution. 
    
    \item {\bf Stage 2:} $\mathsf{C}$ executes $\zk.\zkP(x, w, \lang)$ and $\mathsf{R}$ executes $\zk.\zkV(x, \lang)$  where:
        \begin{itemize}
            \item $x =  \{c_i\}_{i \in [t/2]}$,  $w = \big(m, \{r_{i, \mathsf{C}}\}_{i \in [t/2]} \big)$, 
            \item $\lang = \big\{ \{ b_i \}_{i \in [t/2]} \; \big| \exists \big(a, \;\{s_i\}_{i \in [t/2]} \big) \text{ s.t. } \forall i \in [t/2], \nmcsmall.\mathsf{C}(a,\tagg_i,\phi_i;s_i) = b_i\big\}$
            \item where $\nmcsmall.\mathsf{C}(a,\tagg_i,\phi_i;s_i)$ denotes the transcript output by $\mathsf{C}_i$ on input receiver messages $\phi_i$. 
\end{itemize}
\end{itemize}


} 
                     

\subsubsection{Analysis.}
In the reveal stage, the committer outputs $(m,r_{1,C})$ and the receiver accepts the decommitment if this produces is a valid decommitment of $c_1$ according to $\nmcsmall$.
The perfect binding property of the scheme in \proref{fig:tag_amplification_nmcom} follows directly from the perfect binding property of the underlying protocol $\nmcsmall$. Hiding follows from non-malleability, the proof of which is in Appendix~\ref{app:tagamp}.

\begin{lemma} \label{lem:tagampsecurity}
\proref{fig:tag_amplification_nmcom} is a one-one non-malleable commitment according to Definition \ref{def:nmc} for tags in $[2^{t/2}]$.
\end{lemma}


%
We also have the following lemma, that follows by a standard hybrid argument, due to~\cite{LPV08}.
\begin{lemma}
\label{lem:lpv}
Every quantum-secure non-malleable commitment satisfying Definition \ref{def:nmc} also satisfies Definition \ref{def:many-nmc}.
\end{lemma}

We conclude this section with the following theorem, that can be obtained by applying the compiler in \proref{fig:tag_amplification_nmcom} $(c+1)$ times to the base non-malleable commitment from \proref{fig:basic_nmcom}, and then applying Lemma~\ref{lem:lpv}.
\begin{theorem}
Assuming there exists a constant $c \in \mathbb{N}$ such that at all quantum polynomial size circuits have advantage $\negl(\secp^{\mathsf{ilog}(c,\secp)})$ in distinguishing LWE samples from uniform, and spooky encryption for relations computable by quantum circuits, 
there exist quantum-secure constant round non-malleable commitments satisfying Definition \ref{def:many-nmc}.
\end{theorem}

Finally, we remark that a folklore technique~\cite{DDN91} where the committer and receiver participate in rounds, sending $\bot$ in every round, except for the committer sending $\mathsf{ECom}(m;r)$ in round $i$ (where $i = \tagg$) using a single-committer extractable commitment yields a one-to-one non-malleable commitment for $\tagg \in [N]$ in the synchronous setting, in $O(N)$ rounds, for any $N \leq \poly(\secp)$.
Setting $N = \mathsf{ilog}(c,\secp)$ for any constant $c \in \mathbb{N}$ yields a protocol with $O(\mathsf{ilog}(c,\secp))$ rounds for $\mathsf{ilog}(c,\secp)$ tags. Applying our tag amplification compiler to this scheme $(c+1)$ times, yields a non-malleable commitment for tags in $2^\secp$ against synchronous adversaries, in $O(\mathsf{ilog}(c,\secp))$ rounds.
The underlying extractable commitment can be instantiated using the technique of~\cite{BS20} based on polynomial quantum hardness of LWE and polynomial hardness of QFHE. This yields $O(\mathsf{ilog}(c,\secp))$ round post-quantum non-malleable commitments from polynomial hardness assumptions.
\section{Quantum-Secure Multi-Party Coin-Flipping}
\label{sec:coin-tossing}

We now combine the primitives constructed in earlier sections to build a constant-round coin-flipping protocol secure against quantum polynomial-time adversaries. The protocol was described at a high level in~\cref{sec:over-putting}, and is given in full detail in \proref{fig:ct}. As explained in the overview, each party will first commit to random strings $c_i,r_i$ using a non-malleable commitment, then commit to $c_i$ using our parallel extractable commitment with randomness $r_i$, and finally broadcast $c_i$. The parties will output $\bigoplus_{i \in [n]} c_i$ as the common output if all parties manage to prove in zero-knowledge that they behaved honestly throughout the protocol.

\paragraph{Proof Strategy.} Our simulator will be structurally similar to the zero-knowledge simulator described in~\cref{sec:pzk}, in the sense that we build a simulator $\mathsf{SimNoAbort}_\bot$ specifically for non-aborting transcripts and a simulator $\mathsf{SimAbort}_\bot$ specifically for aborting transcripts. The bulk of the work in $\mathsf{SimNoAbort}_\bot$ involves sampling instances and an advice state (consisting of the adversary's view through Step 5 of the protocol) for the final part of the adversary, which in particular interacts with honest parties in order to verify their zero-knowledge arguments in Step 6. These arguments are then simulated by (part of) the zero-knowledge simulator $\zk.\mathsf{SimNoAbort}_\bot$, which takes as input the adversary and the sampled advice state. 

However, as alluded to in~\cref{sec:over-putting}, once this simulation is performed in the hybrids, it is no longer possible to directly invoke the soundness of the adversary's zero-knowledge arguments in Step 6, when changing how $\mathsf{SimNoAbort}_\bot$ samples the advice state. Thus, we invoke soundness in the very first hybrid to claim that the following ``check'' never fails, except with negligible probability. 

The check fails if the Step 1-5 messages sent by at least one of the malicious parties are not explainable, yet the honest parties do not abort. 
If this check fails, we simply append $\cfail$ to the transcript.
Now, this check will continue to be computed in later hybrids, but we can claim that since (as we show) all later hybrids are indistinguishable from the first hybrid, $\cfail$ must also only appear with negligible probability in these later hybrids. When it comes time to invoke the non-malleability of the honest party commitments, we can use the fact that $\cfail$ appears with negligible probability to show that a malicious party cannot even change its \emph{extractable} commitment based on the changing simulated view. If it could, then since non-malleability implies that its previously sent non-malleable commitment couldn't change, then it must be the case that its messages are no longer explainable (since its two commitments are no longer consistent). Thus, the check will fail and appear in the hybrid's output, a contradiction. Of course, turning this intuition into a formal proof requires much care, especially since this check is inefficient. Thus, we will make liberal use of non-uniform fixing arguments. 

We also remark that it would be most natural to rely on \emph{many-to-many} non-malleable commitments in this multi-party setting. However, we only have a post-quantum construction of one-one commitments in Section~\ref{sec:nmc}. Thus, when invoking non-malleability, a reduction must isolate the commitment of a single malicious party that would constitute a mauling attack. We again use both non-uniformity and the $\cfail$ condition here, showing that any mauling attack would cause $\cfail$ to appear, and thus that there must exist \emph{some} malicious party for which the check fails over specifically its messages. The identity of this party can then be given as non-uniform advice to a reduction.

\subsection{Definition}

\begin{definition}[Quantum-Secure Fully-Simulatable Multi-Party Coin-Flipping] \label{def:coin tossing}
Let $k = k(\lambda)$ be any fixed polynomial.
An fully-simulatable $n$-party $k$-coin-flipping protocol with quantum security is given by $n$ classical interactive Turing machines $(P_1, \ldots, P_n)$ with joint input $(1^\lambda,1^n)$ and outputs $r_i \in \zo^{k(\lambda)} \cup \{\bot\}$. 

Given a coin-flipping protocol and an adversary $\A^* = \{\A^*_\lambda, \rho_\lambda\}_{\lambda \in \bbN}$ that corrupts a set of parties $\mathbb{S} \subset [n]$, let $\honest$ denote $[n]\setminus \mathbb{S}$ and define the random variable $\mathsf{Real}(\A^*_\secp,\rho_\secp)$ to consist of the outputs of honest parties $\{P_i\}_{i \in \honest}$ as well as the view $\view_{\A^*_\lambda}\dist{\A^*_\secp(\rho_\lambda), \{P_i\}_{i \in \honest} }(1^\lambda,1^n)$ of $\A^*_\secp$ after executing the protocol in the presence of $\A^*_\secp$. 

We require the following security property. Fix any $\mathbb{S} \subset [n]$. There exists a quantum expected polynomial-time simulator \simulator, such that for any quantum polynomial-size adversary $\A^* = \{\A^*_\lambda, \rho_\lambda\}_{\lambda \in \bbN}$ that participates in the protocol, generating joint messages on behalf of all algorithms in $\mathbb{S}$,
    \begin{align}
        \{\mathsf{Real}(\A^*_\secp,\rho_\secp)\}_{\lambda \in \bbN} 
        \approx_c
        \{\simulator(1^\secp,1^n,r, \A^*_\lambda, \rho_\lambda) \ | \ r \leftarrow U_{k(\lambda)}\}_{\lambda \in \bbN},
    \end{align}
    where for the ``protocol output'' part of its simulation output, $\simulator(1^\secp,1^n,r, \A^*_\lambda, \rho_\lambda)$ is restricted to output either $r^{|\honest|}$ or $\bot^{|\honest|}$.

        
\end{definition}



\subsection{Construction}

\paragraph{Ingredients}: All of the following are assumed to be quantum-secure.
\begin{itemize}
    \item A many-to-one non-malleable commitment $\nmCom = (\nmCom.\eC,\nmCom.\eR)$.
    \item A multi-committer extractable commitment $\PECom = (\PECom.\eC,\PECom.\eR)$.
    \item A multi-verifier publicly-verifiable zero-knowledge argument for NP $\zk = (\zk.\zkP,\zk.\zkV)$.
\end{itemize}

\paragraph{Languages.} We define two NP languages $\LL$ and $\LLL$. Let $(x,y) \coloneqq \dist{\nmCom.\eC(c;r),\nmCom.\eR_y}$ denote the transcript of an execution of $\nmCom$ where the receiver messages are fixed to $y$ and $x$ is the set of resulting committer messages. Similarly, let $(x,y) \coloneqq \dist{\PECom.\eC(c;r),\PECom.\eR_{y}}$ denote the transcript of an execution between some $\PECom.\eC$ and $\PECom.\eR$ (which may be part of a larger $\PECom$ transcript involving other committers) where the receiver messages are fixed to $y$ and $x$ is the set of resulting committer messages. Then $\LL$ and $\LLL$ are  defined as follows.

\begin{align*}
    &\LL \coloneqq \left\{(x,y) \ \bigg| \ \exists (c,r) \text{ s.t. } (x,y) \coloneqq \dist{\PECom.\eC(c;r),\PECom.\eR_{y}}\right\} \\
    &\LLL \coloneqq \left\{(x,y,x',y',c) \ \bigg| \ \exists (r,s) \text{ s.t. } \begin{array}{l} (x,y) \coloneqq \dist{\nmCom.\eC((c,r);s),\nmCom.\eR_y}, \\ (x',y') \coloneqq \dist{\PECom.\eC(c;r),\PECom.\eR_{y'}}\end{array}\right\} 
\end{align*}

\protocol
{\proref{fig:ct}}
{A quantum-secure constant-round coin-flipping protocol.}
{fig:ct}
{
\textbf{Common input}: Security parameter $1^\lambda$ and number of parties $1^n$.

\begin{enumerate}
    \item For all $i \in [n]$, $P_i$ samples $c_i \leftarrow \zo^{k(\secp)}, \{r_{i,j}, s_{i,j}\}_{j \in [n] \setminus \{i\}} \leftarrow \zo^{2(n-1)\secp}$.
    \item For all $i \in [n], j \in [n] \setminus \{i\}$, $P_i$ runs
    $\nmCom.\eC((c_i,r_{i,j});s_{i,j})$ with tag $i$ and $P_j$ runs $\nmCom.\eR$ to produce $$\left(\alpha_{i,j}^{(\com)},\alpha_{i,j}^{(\rec)}\right) \gets \dist{\nmCom.\eC((c_i,r_{i,j});s_{i,j}),\nmCom.\eR}(1^\secp),$$ where $\alpha_{i,j}^{(\com)}$ denotes the committer messages sent by $P_i$ and $\alpha_{i,j}^{(\rec)}$ denotes the receiver messages sent by $P_j$.
    \item For all $j \in [n]$, $P_j$ runs $\PECom.\eR$ and each $P_i$ for $i \neq j$ runs $\PECom.\eC_i(c_i;r_{i,j})$ to produce $$\left(\left\{\beta_{i,j}^{(\com)}\right\}_{i \in [n] \setminus \{j\}},\left\{\beta_{i,j}^{(\rec)}\right\}_{i \in [n] \setminus \{j\}}\right) \gets \dist{\{\PECom.\eC_i(c_i;r_{i,j})\}_{i \in [n] \setminus \{j\}},\PECom.\eR}(1^\secp,1^{n-1})$$ where $\beta_{i,j}^{(\com)}$ denotes the committer messages sent by $P_i$ and $\beta_{i,j}^{(\rec)}$ denotes the receiver messages sent by $P_j$.
    
    \item For all $i \in [n]$, $P_i$ runs $\zk.\zkP(\{x_{i,j},w_{i,j}\}_{j \in [n] \setminus \{i\}})$ and each $\{P_j\}_{j \in [n] \setminus \{i\}}$ runs $\zk.\zkV_j(x_{i,j})$ in an execution of $\zk$ with common input $(1^\secp,1^{n-1})$ for language $\LL$ (defined above), where $x_{i,j} = (\beta_{i,j}^{(\com)}, \beta_{i,j}^{(\rec)})$ and $w_{i,j} = (c_i,r_{i,j})$.
    \item For all $i \in [n]$, $P_i$ broadcasts $c_i$.
    \item For all $i \in [n]$, $P_i$ runs $\zk.\zkP(\{x_{i,j},w_{i,j}\}_{j \in [n] \setminus \{i\}})$ and each $\{P_j\}_{j \in [n] \setminus \{i\}}$ runs $\zk.\zkV_j(x_{i,j})$ in an execution of $\zk$ with common input $(1^\secp,1^{n-1})$ for language $\LLL$ (defined above), where $x_{i,j} = (\alpha_{i,j}^{(\com)}, \alpha_{i,j}^{(\rec)}, \beta_{i,j}^{(\com)}, \beta_{i,j}^{(\rec)}, c_i)$ and $w_{i,j} = (r_{i,j},s_{i,j})$. 
    \item For all $i \in [n]$, $P_i$ runs the ZK verification algorithm on all $2n(n-1)$ proofs provided in Steps 4 and 6. If all proofs are accepting, $P_i$ outputs $\bigoplus_{i \in [n]} c_i$, and otherwise outputs $\bot$.
\end{enumerate}
}



\subsection{Security}


\begin{theorem}
For any $n$ and polynomial $k(\secp)$, \proref{fig:ct} is a quantum-secure fully-simulatable $n$-party $k$-coin-flipping protocol.
\end{theorem}

\begin{proof}

Fix a number of parties $n$, a polynomial $k = k(\lambda)$, and a set $\mathbb{S} \subset [n]$ of malicious parties. 
We construct a simulator $\simulator$ that for every quantum polynomial-size adversary $\A^* = \{\A^*_\lambda, \rho_\lambda\}_{\lambda \in \mathbb{N}}$ corrupting parties in $\mathbb{S}$,  outputs a distribution
$$\{\simulator(1^\secp, 1^n, r, \A^*_\lambda, \rho_\lambda) \ | \ r \leftarrow U_{k(\lambda)}\}_{\lambda \in \bbN}$$
that satisfies the conditions of Definition \ref{def:coin tossing}. Similar to the proof of zero-knowledge in~\cref{subsec:ZK}, the simulator $\simulator$ will make use of two sub-routines, $\cfSimNoAbort$ and $\cfSimAbort$.\\

\noindent $\cfSimNoAbort(1^\secp, 1^n, r, \A^*_\lambda,\rho_\secp)$:
\begin{enumerate}

\item Let $i^*$ denote the smallest index of a party in $\mathbb{H}$. Define the machine $\zkV^*_\secp$ as follows. $\zkV^*_\secp$ will act on behalf of verifiers $\{\zk.\zkV_j\}_{j \in \corrupt}$ in the $\zk$ session in Step 6 of \proref{fig:ct} where party $P_{i^*}$ is the prover. Thus, it consists of the portion of $\A^*_\secp$ that interacts during this step \emph{as well as} the portion of the honest parties $\honest$ that interact in the $n-1$ sessions where $P_{i^*}$ is \emph{not} the prover.

We will next describe how a particular instance-advice distribution $(\{x_j\}_{j \in \corrupt}, \sigma^{\{x_j\}_{j \in \corrupt}})$ is generated for $\zkV^*_\secp$. Generating this distribution will involve simulating Steps 1-5 of \proref{fig:ct} for adversary $\A^*_\secp$. In particular, the advice state $\sigma^{\{x_j\}_{j \in \corrupt}}$ will include the transcript $\tau^{(5)}$ of the entire simulated execution through Step 5, the inner state $\rho^{(5)}$ of $\A^*_\secp$ at this point, as well as the witnesses $\{r_{i,j},s_{i,j}\}_{i \in \honest \setminus \{i^*\},j \in [n] \setminus \{i\}}$ to be used by parties $\{P_i\}_{i \in \honest \setminus \{i^*\}}$ in Step 6 of \proref{fig:ct}. The instances $\{x_j\}_{j \in \corrupt}$ will be a subset of $\tau^{(5)}$. In particular, for each $j \in \corrupt$, $x_j$ will be set to $(\alpha_{i^*,j}^{(\com)},\alpha_{i^*,j}^{(\rec)},\beta_{i^*,j}^{(\com)},\beta_{i^*,j}^{(\rec)},c'_{i^*})$, which are the messages exchanged by $i^*$ and $j$ during Steps 2 and 3 when $P_{i^*}$ was acting as a committer, as well the value $c'_{i^*}$ broadcast by $P_{i^*}$ in Step 5. 

This instance-advice distribution is generated as follows.

    
\begin{enumerate}
    
    \item For each party $\{P_i\}_{i \in \honest \setminus \{i^*\}}$, sample $c_i,\{r_{i,j},s_{i,j}\}_{j \in [n] \setminus \{i\}}$ as in Step 1 of \proref{fig:ct}.
    
    \item Set $\rho^{(1)} \coloneqq \rho_\secp$ to be the inner state of $\A^*_\secp$, and interact with $\A^*_\secp$ to run Step 2 of \proref{fig:ct} honestly, with the only difference being that party $P_{i^*}$ commits to $(0^{k(\secp)},0^\secp)$. Let $\rho^{(2)}$ be the resulting inner state of $\A^*_\secp$, and let $\{\alpha_{i^*,j}^{(\com)},\alpha_{i^*,j}^{(\rec)}\}_{j \in \corrupt}$ be the messages sent between $P_{i^*}$ and $\corrupt$ in commitments where $P_{i^*}$ was the committer.
    
    
    
    \item Define the machine $\eC^*_\secp$ as follows. $\eC^*_\secp$ will act on behalf of committers $\{\PECom.\eC_j\}_{j \in \corrupt}$ in the $\PECom$ session in Step 3 of \proref{fig:ct} where party $P_{i^*}$ is the receiver. Thus, it consists of the portion of $\A^*_\secp$ that interacts during this step \emph{as well as} the portion of the honest parties $\honest$ that interact in the $n-1$ sessions where $P_{i^*}$ is \emph{not} the receiver. The advice state $\sigma_\secp$ given to $\eC^*_\secp$ will include the transcript $\tau^{(2)}$ of the execution so far, the inner state $\rho^{(2)}$ of $\A^*_\secp$ at this point, and the messages and randomness $\{c_i,r_{i,j}\}_{i \in \honest,j \in [n] \setminus \{i\}}$ to be used in the commitments by honest players in Step 3, where $c_{i^*} = 0^{k(\secp)}$ and $r_{i^*,j}$ is uniformly and independently sampled from the rest of the transcript. The view of $\eC^*_\secp$ at the end of this interaction includes the updated execution transcript $\tau^{(3)}$ as well as the updated inner state $\rho^{(3)}$ of $\A^*_\secp$. 
    
    Now, compute $$(\{\tau_j\}_{j \in \corrupt},\state,\{c'_j\}_{j \in \corrupt}) \gets \PECom.\cE(1^\secp,1^{n-1},\corrupt,\eC^*_\secp,\sigma_\secp),$$
    
    and parse $(\{\tau_j\}_{j \in \corrupt},\state)$ to obtain $\tau^{(3)}$ and $\rho^{(3)}$, where $\tau^{(3)}$ in particular includes the messages $\{\beta_{i^*,j}^{(\com)},\beta_{i^*,j}^{(\rec)}\}_{j \in \corrupt}$ exchanged by $P_{i^*}$ and $\corrupt$ in commitments where some party $j \in \corrupt$ was the receiver.  If $\PECom.\cE$  produced an abort transcript, then return $\bot$, and otherwise continue, setting $\rho^{(3)}$ to be the state of $\A^*_\secp$.

    \item Interact with $\A^*_\secp$ to run Steps 4 and 5 of the protocol honestly, with the only difference being that party $P_{i^*}$ broadcasts $c'_{i^*} \coloneqq \bigoplus_{j \in \mathbb{S}}c'_j \bigoplus_{j \in \mathbb{H} \setminus \{i^*\}}c_j \oplus r$ in Step 5.

    
    
    \item If there exists $j \in \corrupt$ such that the output $c_j$ of $P_j$ in Step 5 is not equal to $c'_j$, then return $\bot$. Otherwise, let $\tau^{(5)}$ denote the transcript so far and let $\rho^{(5)}$ denote the state of $\A^*_\secp$ at the end of Step 5. For each $j \in \corrupt$ set $x_j = (\alpha_{i^*,j}^{(\com)},\alpha_{i^*,j}^{(\rec)},\beta_{i^*,j}^{(\com)},\beta_{i^*,j}^{(\rec)},c'_{i^*})$, and set $\sigma^{\{x_j\}_{j \in \corrupt}} = (\tau^{(5)},\rho^{(5)},\{r_{i,j},s_{i,j}\}_{i \in \honest \setminus \{i^*\},j \in [n] \setminus \{i\}})$.
    
\end{enumerate}


\item Now, compute $$\view \leftarrow \zk.\zkSimNoAbort(1^\secp,1^{n-1},\corrupt,\{x_j\}_{j \in \corrupt},\zkV^*_\secp,\sigma^{\{x_j\}_{j \in \corrupt}}).$$ If the proofs (which are included in $\view$) output by $\A^*_\secp$ are accepting, then output $\view$, and otherwise output $\bot$. 
\end{enumerate}

\noindent $\cfSimAbort(1^\secp, 1^n, r, \A^*_\secp,\rho_\secp)$:
\begin{enumerate}
    \item Set $\rho_\secp$ as the initial state of $\A^*_\secp$.
    Execute Steps 1-6 with $\A^*_\secp$ using honest party strategy according to \proref{fig:ct} on behalf of parties $\{P_j\}_{j \in \mathbb{H}}$.
    If sampling this distribution leads to an abort at any point, halt and return the view of $\A^*_\secp$. Otherwise, return $\bot$.
\end{enumerate}

\medskip

\noindent $\mathsf{SimComb}_\bot(1^\secp,1^n,I,\A^*_\secp,\rho_\secp)$: With probability 1/2, execute $\zkSimNoAbort(1^\secp,1^n,I,\{x_i\}_{i \in [n]},\zkV^*_\secp,\rho_\secp)$ and otherwise execute $\zkSimAbort(1^\secp,1^n,I,\A^*_\secp,\rho_\secp)$.

\medskip
\noindent $\simulator(1^\secp,1^n,I,\A^*_\secp,\rho_\secp)$: Let $\overline{\mathsf{SimComb}}_\bot(\cdot) \coloneqq \mathsf{SimComb}_\bot(1^\secp,1^n,I,\A^*_\secp,\cdot)$ be the circuit $\mathsf{SimComb}_\bot$ with all inputs hard-coded except for $\rho_\secp$, and output $\R(\overline{\mathsf{SimComb}}_\bot,\rho_\secp,\secp)$, where $\R$ is the algorithm from~\cref{lemma:rewinding}.

\medskip

This concludes the description of the simulator. Before proceeding to the proof of indistinguishability, we define the following collections of random variables (each indexed by $\secp$). Each is defined with respect to the adversary $\A^* = \{\A^*_\secp,\rho_\secp\}_{\secp \in \mathbb{N}}$ that we are considering. Throughout, whenever we say abort, we mean that either one of the parties controlled by the adversary aborts, or it fails to prove one of its statements. 

\begin{itemize}
    \item Let $\zkSimNoAbort(\A^*) \coloneqq \{\zkSimNoAbort(1^\secp,1^n,r,\A^*_\secp,\rho_\secp)\}_{\secp \in \mathbb{N}}$.
    \item Let $\zkSimAbort(\A^*) \coloneqq \{\zkSimAbort(1^\secp,1^n,r,\A^*_\secp,\rho_\secp)\}_{\secp \in \bN}$.
    \item Let $\mathsf{SimComb}_\bot(\A^*) \coloneqq \{\mathsf{SimComb}_\bot(1^\secp,1^n,r,\A^*_\secp,\rho_\secp)\}_{\secp \in \bN}$.
    \item Let $\zkSim(\zkV^*) \coloneqq \{\zkSim(1^\secp,1^n,r,\A^*_\secp,\rho_\secp)\}_{\secp \in \bN}$.
    \item Let $\mathsf{RealNoAbort}_\bot(\A^*)$ be the distribution $\mathsf{Real}(\A^*)$, except that whenever an abort occurs, the distribution outputs $\bot$.
    \item Let $\mathsf{RealAbort}_\bot(\A^*)$ be the distribution $\mathsf{Real}(\A^*)$, except that if an abort \emph{does not} occur, the distribution outputs $\bot$.
    \item Let $\mathsf{SimComb}(\A^*)$ be the distribution $\mathsf{SimComb}_\bot$ conditioned on the output not being $\bot$.
\end{itemize}

Next, we prove the following claim.

\begin{claim}
\label{lem:ind-non-abort-1}
$$\mathsf{RealNoAbort}_\bot(\A^*) \approx_c \cfSimNoAbort(\A^*)$$
\end{claim}
\begin{proof}
This can be proved via the following sequence of hybrids.
\begin{itemize}

\item $\hyb_0:$ $\mathsf{RealNoAbort}_\bot(\A^*)$.

\item $\hyb_1:$ 
This hybrid is the same as $\hyb_0$, except that it attaches $\cfail$ to the output if the following (inefficient check) on Steps 1-4 of the transcript, fails. 

Let $i^*$ denote the smallest index in $\mathbb{H}$. For $j \in \corrupt$, let $y_j \coloneqq (\alpha_{j,i^*}^{(\com)}, \alpha_{j,i^*}^{(\rec)}, \beta_{j,i^*}^{(\com)}, \beta_{j,i^*}^{(\rec)},c_j)$ be the messages exchanged between $P_{i^*}$ and $P_j$ in Steps 2 and 3 when $P_j$ was acting as the committer, along with the message broadcast by $P_j$ in Step 4. The check fails if there exists a $j \in \corrupt$ such that $y_j \notin \LLL$.

\item $\hyb_2:$ Let $\zkV^*_\secp$ be the machine defined in the description of $\cfSimNoAbort$. Sample instance-advice distribution $(\{x_j\}_{j \in \corrupt}, \sigma^{\{x_j\}_{j \in \corrupt}})$ as described below.
\begin{enumerate}
\item Execute Steps 1-4 of the protocol identically to $\hyb_2$.
Let $\tau^{(4)}$ denote the transcript generated so far, let $\rho^{(4)}$ denote the state of $A^*_\secp$ at the end of Step 4, and let $\{r_{i,j},s_{i,j}\}_{i \in \honest \setminus \{i^*\},j \in [n] \setminus \{i\}}$ be strings drawn in Step 1 of the protocol.
\item If the check described in $\hyb_2$ fails, then attach $\cfail$ to the transcript.
\item For $j \in \corrupt$, let $x_j \coloneqq (\alpha_{i^*,j}^{(\com)},\alpha_{i^*,j}^{(\rec)},\beta_{i^*,j}^{(\com)},\beta_{i^*,j}^{(\rec)},c'_{i^*})$ be the messages exchanged between $P_{i^*}$ and $P_j$ in Steps 2 and 3 when $P_{i^*}$ was acting as the committer and $P_j$ was acting as the receiver, along with the message broadcast by $P_{i^*}$ in Step 4. Set $\sigma^{\{x_j\}_{j \in \corrupt}} = (\tau^{(4)},\rho^{(4)},\{r_{i,j},s_{i,j}\}_{i \in \honest \setminus \{i^*\},j \in [n] \setminus \{i\}})$.
\end{enumerate}
Now, compute $$\view \leftarrow \zk.\zkSimNoAbort(1^\secp,1^{n-1},\corrupt,\zkV^*_\secp,(\{x_j\}_{j \in \corrupt}, \sigma^{\{x_j\}_{j \in \corrupt}})).$$ If the proofs (which are included in $\view$) output by $A^*_\secp$ are accepting, then output $\view$, and otherwise output $\bot$.  \\

\item $\hyb_3:$ 
Sample instance-advice distribution identically to $\hyb_2$, except that in Step 2, party $P_{i^*}$ commits to $(0^{k(\secp)},0^\secp)$.

\item $\hyb_4:$ Sample instance-advice distribution identically to $\hyb_3$, except that in Step 4, the $\zk$ session where $P_{i^*}$ is the prover is simulated.

\item $\hyb_5:$ Sample instance-advice distribution identically to $\hyb_4$, except that in Step 3, for all $j \in {[n] \setminus \{i^*\}}$, party $P_{i^*}$ commits to $0^{k(\secp)}$.

\item $\hyb_6:$ Sample instance-advice distribution identically to $\hyb_5$, except that in Step 4, the $\zk$ session where $P_{i^*}$ is the prover is performed honestly, with witnesses $\{(0^{k(\secp)},r_{i^*,j})\}_{j \in \corrupt}$.

\item $\hyb_7:$ Let $\eC^* = (\eC^*_\secp,\sigma_\secp)_{\secp \in \bN}$ be the machine and corresponding non-uniform advice as defined in the description of $\cfSimNoAbort$. Sample instance-advice distribution identically to $\hyb_6$, except that the values committed by $\A^*$ in interaction with $P_{i^*}$ are extracted as in the description of $\simulator$. Additionally, this hybrid outputs $\bot$ if the values $\{c_j\}_{j \in \mathbb{S}}$ output by parties $\{P_j\}_{j \in \mathbb{S}}$ in Step 5 do not match the extracted $\{c'_j\}_{j \in \corrupt}$.


\item $\hyb_8:$ This hybrid is the distribution $\cfSimNoAbort(A^*)$. The only differences between $\hyb_7$ and $\hyb_8$ are the check introduced in $\hyb_1$ is removed, and:
\begin{itemize}
    \item In $\hyb_7$, the challenger generates $P_{i^*}$'s message in Step 4 by sampling uniformly random $c_{i^*}$.
    \item In $\hyb_8$, the challenger generates $P_{i^*}$'s message in Step 4 as $c_{i^*} = \bigoplus_{j \in \mathbb{S}}c'_j \bigoplus_{j \in \mathbb{H} \setminus \{i^*\}}c_j \oplus r$.
\end{itemize}
\end{itemize}



Now we show that each consecutive pair of hybrids is indistinguishable. We let $\mathsf{BAD}_i$ be the event that in hybrid $\hyb_i$, $\mathsf{check}\text{-}\mathsf{fail}$ appears in the output distribution (meaning that the check introduced in $\hyb_1$ failed AND there was no abort).

\begin{itemize}
    \item $\hyb_0 \approx_s \hyb_1:$ It suffices to show that $\Pr[\mathsf{BAD}_1] = \negl(\secp)$, which follows directly from the quantum computational soundness of $\zk$.
    
    \item $\hyb_1 \approx_c \hyb_2:$ This follows from~\cref{claim:ind-non-abort}.
    \item $\hyb_2 \approx_c \hyb_3:$ We define distributions $\hyb_2'$ and $\hyb_3'$ that are identical to $\hyb_2$ and $\hyb_3$ respectively, except that $\hyb_2'$ and $\hyb_3'$ {\em do not} perform the additional check described in $\hyb_1$, and as such, never attach $\cfail$ to the output.
    
    To prove that $\hyb_2$ and $\hyb_3$ are computationally indistinguishable, it suffices to prove that $\hyb_2 \approx_c \hyb_2'$, $\hyb_2' \approx_c \hyb_3'$, and $\hyb_3' \approx_c \hyb_3$. The first indistinguishability follows from the fact that $\hyb_2 \approx_c \hyb_1$ and $\Pr[\mathsf{BAD}_1] = \negl(\secp)$, which means that $\Pr[\mathsf{BAD}_2] = \negl(\secp)$. The second indistinguishability follows directly from the hiding of $\nmCom$ (implied by~\cref{def:nmc}). In what follows, we show that $\Pr[\mathsf{BAD}_3] = \negl(\secp)$, which implies that $\hyb_3' \approx_c \hyb_3$.

Let $\mathsf{BAD}_{3,j}$ be the event that, in hybrid $\hyb_3$, $y_j \notin \LLL$ (where $y_j$ was defined in $\hyb_1$) and yet the hybrid did not abort. Now suppose that there exists a polynomial $p(\cdot)$ such that for large enough $\secp \in \bbN$, $\Pr[\mathsf{BAD}_3] \geq 1/p(\secp)$.  Assuming that $n = \poly(\secp)$, this implies that there exists a polynomial $p'(\cdot)$ such that for large enough $\secp \in \mathbb{N}$, there exists some $j^*_\secp$ such that $\Pr[\mathsf{BAD}_{3,j^*}] \geq 1/p'(\secp)$.

   
We will use this to contradict many-to-one non-malleability of $\nmCom$, by building a $\mathsf{QPT}$ man-in-the-middle adversary $\mim = \{\mim_\secp,\sigma_\secp\}_{\secp \in \bbN}$ that uses $\A^*$ to contradict~\cref{def:nmc}.


$\mim_\secp$ obtains as non-uniform advice i) the index $j^*_\secp$ that maximizes $\Pr[\mathsf{BAD}_{3,j^*}]$ and ii) $\A^*_\secp$'s advice state $\rho_\secp$. It simulates the first two steps of the coin-flipping protocol in the presence of $\A^*_\secp$. During Step 2, it interacts with a challenger on the left committing to either $(c_{i^*}, r_{i^*,j})$ for each ${j \in \mathbb{S}}$, or to $(0^{k(\secp)},0^\lambda)$ for each ${j \in \mathbb{S}}$, on behalf of $P_{i^*}$. It forwards these to $\A^*_\secp$ on behalf of $P_{i^*}$ and uses the strategy in $\hyb_2$ to generate messages on behalf of all other honest parties. When $\A^*_\secp$ outputs committer messages computed on behalf of $P_{j^*_\secp}$ in its interaction with $P_{i^*}$, $\mim_\secp$ forwards these to a challenger on the right, and in return obtains receiver messages on behalf of $P_{i^*}$.

In other words, $\mim_\secp$'s interaction with its challengers generates either the random variable (defined in~\cref{def:nmc})

$$\viewval_{\mim_\secp} \dist{\{\cC(c_{i^*},r_{i^*,j})\}_{j \in \corrupt}, \mim_\secp(\rho_\secp), \cR} (1^\secp, \tagg_{i^*}, \tagg_{j^*_\secp})$$

or the random variable 

$$\viewval_{\mim_\secp} \dist{\{\cC(0^{k(\secp)},0^\secp)\}_{j \in \corrupt}, \mim_\secp(\rho_\secp), \cR} (1^\secp, \tagg_{i^*}, \tagg_{j^*_\secp}),$$

depending on which strings the challenger on the left is committing to. 


We now show the existence of a quantum polynomial-time $\D = \{\D_\secp\}_{\secp \in \bbN}$ that succeeds in distinguishing these distributions with non-negligible advantage, which contradicts the non-malleability of $\nmCom$ as defined in~\cref{def:nmc}. The distribution received by $\D_\secp$ includes the message $(c^*, r^*)$ committed by $\mim_\secp$ in its interaction with $P_{i^*}$ on the right, along with the final view $\view$ of $\mim_\secp$, which includes $\A^*$'s view after Step 2 of \proref{fig:ct}. It then simulates the remainder of the coin-flipping protocol in the presence of $\A^*_\secp$, with one difference. Instead of implementing the check introduced in $\hyb_1$, it checks only that $y_{j^*_\secp} \in \LLL$, using the message $(c^*,r^*)$. If this check failed and there was no abort, it outputs 1 and otherwise outputs 0.

Now observe that
\begin{align*}
    &\Pr\left[\D_\secp(\viewval_{\mim_\secp} \dist{\{\cC(c_{i^*},r_{i^*,j})\}_{j \in \corrupt}, \mim_\secp(\rho_\secp), \cR} (1^\secp, \tagg_{i^*}, \tagg_{j^*_\secp})) = 1 \right]\\ & \ \ \ \ \ = \Pr[\mathsf{BAD}_{2,j^*_\secp}] \leq \Pr[\mathsf{BAD}_2] = \negl(\secp), \text{ and}\\
    &\Pr\left[\D_\secp(\viewval_{\mim_\secp} \dist{\{\cC(0^{k(\secp)},0^\secp)\}_{j \in \corrupt}, \mim_\secp(\rho_\secp), \cR} (1^\secp, \tagg_{i^*}, \tagg_{j^*_\secp})) = 1 \right] \\ &\ \ \ \ \ =\Pr[\mathsf{BAD}_{3,j^*_\secp}] \geq 1/p'(\secp),
\end{align*}

which establishes that $\D$ has a non-negligible advantage, a contradiction.




\item $\hyb_3 \approx_c \hyb_4:$ This follows from the quantum zero-knowledge of $\zk$. The non-uniform advice given to the malicious verifier derived from $\A^*$ will include the transcript of the first two rounds of \proref{fig:ct} executed with adversary $\A^*$, along with the openings (if they exist) of the commitments made by $\A^*$ in Step 2. The final view of this verifier will also include these openings, allowing the reduction to efficiently simulate the remainder of the protocol, in particular using these openings to efficiently implement the check introduced in $\hyb_1$.

\item $\hyb_4 \approx_c \hyb_5:$ We consider a sequence of sub-hybrids $\hyb_{4,0},\dots,\hyb_{4,|\corrupt|}$. Associate the set $\corrupt$ with the set $[1,|\corrupt|]$, and define $\hyb_{4,j}$ so that in Step 3, $P_{i^*}$ commits to $0^{k(\secp)}$ when interacting with adversarial parties $P_k$ for $k \leq j$ and commits to $c_{i^*}$ when interacting with adversarial parties $P_k$ for $k > j$. Observe that $\hyb_4 = \hyb_{4,0}$ and $\hyb_5 = \hyb_{4,|\corrupt|}$. 

We now show that for any $j \in [1,\corrupt]$, the indistinguishability $\hyb_{4,j-1} \approx_c \hyb_{4,j}$ follows from the quantum computational hiding of $\PECom$. Indeed, define a receiver $\eR^* = (R^*_\secp, \rho_\secp)$ that interacts with a single committer committing to either $c_{i^*}$ or $0^{k(\secp)}$ as follows. It takes as non-uniform advice $\rho_\secp$ the transcript of the first two rounds of \proref{fig:ct} executed with adversary $\A^*_\secp$, along with the openings (if they exist) of the commitments made by $\A^*_\secp$ in Step 2. It then simulates the remainder of the protocol, interacting with the challenger to implement $P_{i^*}$'s messages in Step 3 during the $\PECom$ session when $P_j$ is the receiver. Observe that $\eR^*$ can indeed efficiently simulate the entire protocol, in particular it can implement the check introduced in $\hyb_1$ since it has the openings to the commitments given by $\A^*_\secp$ in Step 2. Any efficient distinguisher that distinguishes between $\hyb_{4,j-1}$ and $\hyb_{4,j}$ with non-negligible advantage immediately implies that $\eR^*$ distinguishes with non-negligible advantage, breaking quantum computational hiding of $\PECom$.

\item $\hyb_5 \approx_c \hyb_6:$ Same argument as $\hyb_3 \approx_c \hyb_4$.

\item $\hyb_6 \approx_c \hyb_7:$ Assume that there exists a distinguisher $\eD$ that can distinguish between the outputs of these hybrids with non-negligible advantage. We build a compliant\footnote{Recall that such a distinguisher is guaranteed to output 0 with overwhelming probability on input any non-explainable view.} distinguisher $\D'$ that breaks the extractability property of $\PECom$.

$\D'$ will receive as non-uniform advice i) the transcript of the first two rounds of \proref{fig:ct} executed with adversary $\A^*$, ii) the state of $\A^*$ at this point, and iii) the openings (if they exist) of the commitments made by \emph{each} party in Step 2. Note that the non-uniform advice given to the committer $\eC^*$ defined in the description of $\cfSimNoAbort$ is a strict subset of this advice. Now, $\D'$ will forward this subset (which consists of the state of $\A^*$ and the commitment openings of parties $\{P_i\}_{i \in \honest}$) to its challenger, and receive either the real or simulated view with respect to committer $\eC^*$. It can then efficiently generate the rest of the distribution using its non-uniform advice and the view it received from the challenger, additionally returning an abort if the messages $\{c'_i\}_{i \in \corrupt}$ it received as part of the challenge distribution do not match the messages $\{c'_i\}_{i \in \corrupt}$ broadcast in Step 5. If during Step 4, any of the parties $\{P_i\}_{i \in \corrupt}$ fails to prove it $\zk$ statement, $\D'$ outputs 0. Otherwise, it queries $\D$ with the final distribution and outputs what $\D$ outputs.

It remains to show that i) $\D'$'s advantage is negligibly close to $\D$'s advantage, and ii) $\D'$ is compliant. The first point requires two observations. First, whenever $\D'$ does \emph{not} query $\D$, it means that $\A^*$ failed to prove one of it $\zk$ statements, so $\D$'s input would have been $\bot$ anyway. Next, we need to show that when $\D'$ \emph{does} query $\D$ with a transcript, it is a faithful execution of either $\hyb_6$ or $\hyb_7$, depending on whether $\PECom$ was simulated or not. If $\PECom$ was simulated, the distribution is equivalent to $\hyb_7$. If not, the distribution is equivalent to $\hyb_6$, except for the extra abort condition carried out by the reduction. However, observe that the probability that the reduction produces an abort but $\hyb_6$ does not is at most $\Pr[\mathsf{BAD}_6]$. Since $\hyb_6 \approx_c \hyb_1$, and $\Pr[\mathsf{BAD}_1] = \negl(\secp)$, this quantity is negligible. Now it remains to argue that $\D'$ is compliant, but this follows directly from the quantum computational soundness of $\zk$.



\item $\hyb_7 \approx_s \hyb_8$: First, $\Pr[\mathsf{BAD}_7] = \negl(\secp)$, since $\hyb_7 \approx_c \hyb_1$. Next, switching $P_{i^*}$'s message in Step 4 is perfectly indistinguishable since $r$ is uniformly random.
\end{itemize}
\end{proof}

Now, note that $\mathsf{RealAbort}_\bot(\A^*) \approx_c \cfSimAbort(\A^*)$ follows by definition. Then, it follows identically to the proof of~\cref{thm:ZK} that $\mathsf{Real}(\A^*) \approx_c \mathsf{SimComb}(\A^*)$, and then by applying~\cref{lemma:rewinding}, that $\mathsf{Real}(\A^*) \approx_c \mathsf{Sim}(\A^*)$. 




\end{proof}

\section{Quantum-Secure Multi-Party Computation}
\label{sec:mpc}

\subsection{Definition}

We follow the standard real/ideal world paradigm for defining secure multi-party computation (MPC) as in \cite{Goldbook}, replacing classical adversaries with quantum adversaries. 

Consider $n$ parties $P_1,\ldots,P_n$ with inputs $x_1,\ldots,x_n$ that wish to interact in a protocol $\Pi$ to evaluate any functionality $f$ on their joint inputs. The security of protocol $\Pi$ (with respect to a functionality $f$) is defined by comparing the real-world execution of the protocol with an ideal-world evaluation of $f$ by a trusted party. Informally, it is required that for every quantum adversary $\A = \{\A_\secp\}_{\secp \in \bN}$ that corrupts some subset of the parties $I \subset [n]$ and participates in the real execution of the protocol, there exists an adversary $\simulator$, also referred to as a simulator, that can \emph{achieve the same effect} in the ideal world. In fact, we provide a strictly stronger definition that allows the adversary $\A$ some arbitrary non-uniform quantum advice $\{\rho_\secp\}_{\secp \in \bbN}$ thay may even depend on the inputs $x_1,\dots,x_n$.


We now formally describe the security definition, which only considers the case of fully \emph{malicious} adversaries. Let $\vec{x} = (x_1,\ldots,x_n)$ be the set of inputs. 

\paragraph{The Real Execution}. In the real execution, the $n$-party protocol $\Pi$ for computing $f$ is executed in the presence of a quantum polynomial-time adversary $\A = \{\A_\secp,\rho_\secp\}_{\secp \in \bbN}$, where $\A$ corrupts some set $I \subset [n]$ of the parties. The honest parties follow the instructions of $\Pi$, and $\A$ sends all messages of the protocol on behalf of the corrupted parties following any arbitrary quantum polynomial-time strategy. 
 
The interaction of $(\A_\secp,\rho_\secp)$ in the protocol $\Pi$ defines a random variable $\REAL_{\Pi,\A}(\secp,\vec{x},\rho_\secp)$ whose value is determined by the randomness of the adversary and the honest parties. This random variable contains the output of the adversary (which may be an arbitrary function of its view and in particular may be a quantum state) as well as the outputs of the honest parties.

\paragraph{The Ideal Execution}. In the ideal execution, an ideal world adversary $\simulator$ interacts with a trusted party, as follows. 


\begin{itemize}

\item{ \bf Send inputs to the trusted party:} Each honest party sends its input to the trusted party. Each corrupt party $P_i$, (controlled by $\simulator$) may either send its input $x_i$, or send some other input of the same length to the trusted party. Let $x'_i$ denote the value sent by party $P_i$. 

\item{ \bf Trusted party sends output to the adversary:} The trusted party computes $f(x'_1,\ldots,x'_n) = (y_1,\ldots,y_n)$ and sends $\{y_i\}_{i\in I}$ to the adversary.

\item{ \bf Adversary instructs trusted party to abort or continue:} This is formalized by having the adversary send either an abort or continue message to the trusted party. In the latter case, the trusted party sends to each honest party $P_i$ its output value $y_i$. In the former case, the trusted party sends the special symbol $\bot$ to each honest party.

\item{\bf Outputs:} 
$\simulator$ outputs an arbitrary function of its view, and the honest parties output the values obtained from the trusted party.

\end{itemize}

The interaction of $\simulator$ with the trusted party defines a random variable $\IDEAL_{f,\simulator}(\secp,\vec{x},\rho_\secp)$. Having defined the real and the ideal worlds, we now proceed to define our notion of security.

\begin{definition}
\label{def:mpcdefn}
Let $f$ be an $n$-party functionality, and $\Pi$ be an $n$-party protocol. Protocol $\Pi$ securely computes $f$ if for every quantum polynomial-time real-world adversary $\A = \{\A_\secp\}_{\secp \in \bN}$ corrupting a set of at most $n-1$ players, there exists a quantum polynomial-time ideal-world adversary $\simulator$ such that for any set of inputs $\vec{x} \in (\{0,1\}^*)^n$ and any non-uniform quantum advice $\rho = \{\rho_\secp\}_{\secp \in \bbN}$, 

$$\{\REAL_{\Pi,\A}(\secp,\vec{x},\rho_\secp)\}_{\secp \in \bbN} \approx_c \{\IDEAL_{f,\simulator}(\secp,\vec{x},\rho_\secp)\}_{\secp \in \bbN}.$$
\end{definition}



\subsection{Construction}
\label{subsec:mpc}

Given the construction of quantum-secure multi-party coin-flipping from~\cref{sec:coin-tossing}, it is straightforward to achieve quantum-secure multi-party computation, due to the following lemma adapted from~\cite{EC:KatOstSmi03}. For completeness, we give a sketch of the proof.

\begin{lemma}
Given a quantum-secure multi-party coin-flipping protocol and a quantum-secure protocol $\Pi$ for computing $f$ in the common random string (CRS) model with straight-line black-box simulation, the natural composition of the two is a quantum-secure protocol for computing $f$ with no setup assumptions.
\end{lemma}

\begin{proof} (sketch) Consider any adversary $(\A_\secp,\rho_\secp)$ for the composed protocol. $\A_\secp$ may be split into two parts: $\A_1$ interacts in the coin-flipping protocol and produces a state $\state$, which is passed to $\A_2$, who interacts in $\Pi$. We now construct a simulator $\simulator$ for the composed protocol as follows. It begins by running the straight-line black-box simulator $\simulator_\Pi$ for $\Pi$ until it outputs a CRS $r$ (note that since $\simulator_\Pi$ is straight-line, this CRS-generation step is independent of the adversary and advice, and does not require a call to the ideal functionality). At this point, $\simulator$ runs the simulator for the multi-party coin-flipping protocol on input $r$, adversary $\A_1$, and non-uniform advice $\rho_\secp$.  This simulation produces a final state $\state$. Finally, $\simulator$ completes the execution of $\simulator_\Pi$ on input $\A_2(\state)$ and outputs what $\simulator_\Pi$ outputs.

\end{proof}

\section{Acknowledgments} 
Part of this work was done during a visit to the Simons Institute Berkeley for the “Lattices: Algorithms, Complexity, and Cryptography” program.

This material is based on work supported in part by DARPA under Contract Nos. HR001120C0024 (for AA and DK) and HR001120C0025 (for VG). Any opinions, findings and conclusions or recommendations expressed in this material are those of the author(s) and do not necessarily reflect the views of the United States Government or DARPA.

The authors thank Zvika Brakerski and Rishab Goyal for insightful discussions. The authors are also grateful to Daniel Wichs for pointing out the counterexample in~\cref{app:zkcounterexample}, which we included with his permission.

\bibliography{abbrev2,crypto,main.bib,refs}

\appendix
\section{Simple Polynomial-Round Extractable Commitments}
\label{app:poly-round}

In what follows, we describe ideas in~\cite{BS20} that can be used to convert a post-quantum zero-knowledge protocol to an extractable commitment scheme (assuming quantum hardness of LWE). Specifically, we start with the (polynomial-round) zero-knowledge protocol in~\cite{10.1137/060670997} that can be based on any quantum one-way function, and convert it into a (polynomial-round) extractable commitment scheme.

Let $\Com(\alpha;\beta)$ denote a non-interactive perfectly binding, quantum-hiding commitment to classical string $\alpha$ with randomness $\beta$.
We also let $\cdsSch$ denote a two-party two-message conditional disclosure of secrets protocol, where in the first message, the receiver outputs a statement $x$, an NP language $L$, and purportedly commits to an NP witness for this statement. 
Next, the sender encodes a secret $m$ in such a way that the receiver can recover $m$ if and only if it previously committed to an NP witness for $x$. The (informal) security property is that the NP witness is hidden from a semi-honest sender, and the sender's secret $m$ is hidden from a malicious receiver whenever $x \not \in L$. It is well-known (eg.,~\cite{OstrovskyPP14,BadrinarayananG17}) that this can be achieved by combining a specific type of two-message OT (called statistically sender-private OT, that can itself be based on quantum hardness of LWE~\cite{BrakerskiD18}) with garbled circuits.

Given these components, a polynomial-round extractable commitment is described in Figure~\ref{fig:extcom-intro}.

\begin{figure}[ht!]
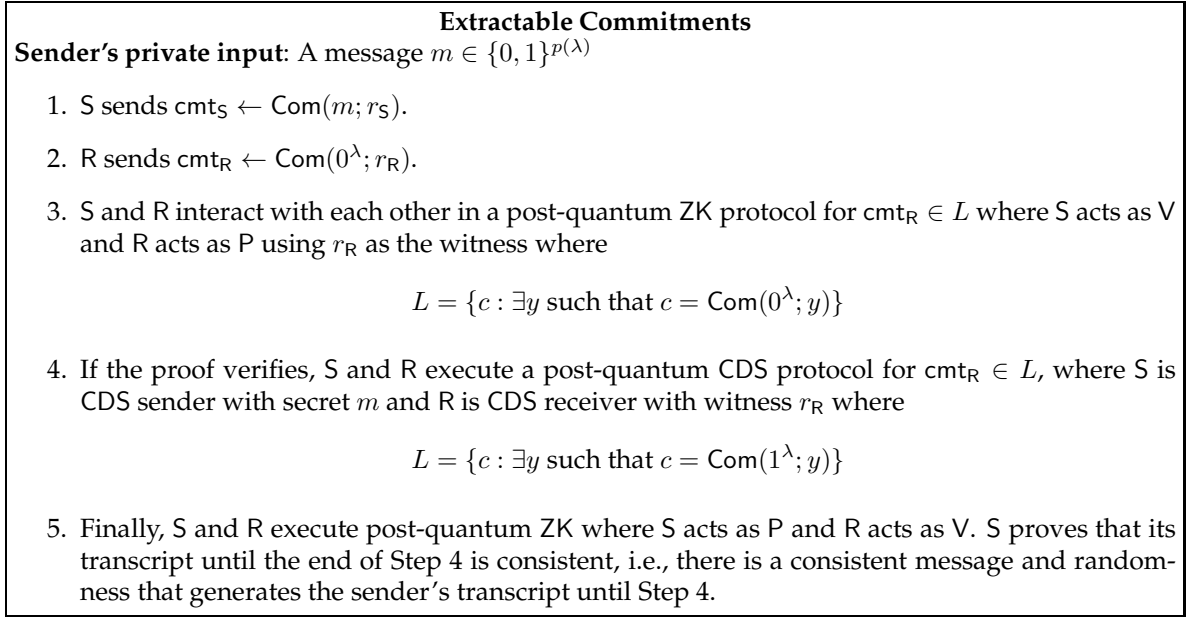

\begin{boxedalgo}
\begin{center}
\textbf{Extractable Commitments}
\end{center}
\textbf{Sender's private input}: A message $m \in \zo^{p(\lambda)}$
\begin{enumerate}
\item $\sender$ sends $\cmt_\sender \gets \COM(m; r_\sender)$.

\item $\receiver$ sends $\cmt_\receiver \gets \COM(0^\secp; r_\receiver)$.

\item $\sender$ and $\receiver$ interact with each other in a post-quantum $\zk$ protocol for $\cmt_\receiver \in L$ where $\sender$ acts as $\zkV$ and $\receiver$ acts as $\zkP$ using $r_\receiver$ as the witness where
$$ L = \{ c: \exists y \text{ such that } c = \Com(0^\secp; y)\} $$

\item If the proof verifies, $\sender$ and $\receiver$ execute a post-quantum $\cdsSch$ protocol for $\cmt_\receiver \in L$, 
where $\sender$ is $\cdsSch$ sender with secret $m$ and $\receiver$ is $\cdsSch$ receiver with witness $r_\receiver$ where
$$ L = \{ c: \exists y \text{ such that } c = \Com(1^\secp; y)\} $$

\item Finally, $\sender$ and $\receiver$ execute post-quantum $\zk$ where $\sender$ acts as $\zkP$ and $\receiver$ acts as $\zkV$. $\sender$ proves that its transcript until the end of Step 4 is consistent, i.e., there is a 
consistent message and randomness that generates the sender's transcript until Step 4.
\end{enumerate}


\end{boxedalgo}
\caption{Extractable Commitments}
\label{fig:extcom-intro}
\end{figure}

At a high level, the commitment is hiding because any cheating receiver that completes Step 3 will, by the soundness of the ZK protocol, have committed to $0^\lambda$ in Step 2. By the perfect binding property of the commitment, this means that the statement of the CDS protocol, in Step 4, is false. Therefore, because the indistinguishability-based security of CDS (as discussed above), the committer's message $m$ remains hidden from a QPT receiver. 

The commitment is extractable against quantum committers (which also implies binding), because of the following argument: Consider extractor $\extractor$ that in Step 2 generates $\cmt_\receiver$ as $\Com(1^\lambda;r_\receiver)$ instead of committing to $0^\lambda$. Next, $\extractor$ runs $\zk.\zkSim$ to simulate the proof in Step 3. After this point, $\extractor$ uses $r_\receiver$ as witness in Step 4, which enables it to successfully retrieve $m$ from $\cdsSch$. 
\section{An explicit quantum attack against a classically-secure ZK protocol}
\label{app:zkcounterexample}

In this section, we will present the construction of ZK protocol for which the zero knowledge property holds against classical verifiers, however, there exists an explicit attack w.r.t a malicious quantum verifier. This example is inspired by the recent construction of quantum extraction schemes for NP relations\cite{EPRINT:ALP19} where the setting is as follows: The sender $\cS$ and the receiver $\cR$ hold an NP instance $x$. Additionally, $\cS$ holds a witness $w$ for the instance $x$. The desired property is the following: i) Extractability: For any QPT malicious $\cS$, there exists a QPT extractor that can extract a valid witness $w'$ for $x$, ii) Zero-Knowledge: For any PPT malicious $\cR$, there exists a PPT simulator which can simulate the view of $\cR$ without having access to $w$.

The extraction scheme presented in \cite{EPRINT:ALP19} makes use of the ``test of quantumness'' protocol \cite{FOCS:BCMVV18} as a key ingredient which, as the name suggests, is used to attest whether the prover is a quantum machine or a classical one. The authors leverage this test to construct a quantum extraction protocol (i.e it admits a quantum extractor) secure against classical receivers which is referred to as to as cQEXT.  We note that the quantum extractor construction presented in \cite{EPRINT:ALP19} is straight-line i.e. it does not perform any kind of quantum rewinding \cite{10.1137/060670997, EC:Unruh12} on $\cS$. Also, the extractor makes only black-box use of the malicious $\cS$ i.e. it does not make any use of the circuit representation of a malicious $\cS$. As we will see shortly, these two properties will be crucial in the construction of our counterexample. Our counterexample involves a classical prover $\cP(x, w)$ interacting with a classical verifier $\cV(x)$ in the following manner:

\begin{enumerate}
    \item $\cP$ and $\cV$ engage in a cQEXT protocol where $\cP$ acts a cQEXT sender using $(x, w)$ and $\cV$ acts as a cQEXT receiver using $x$.
    \item $\cP(x, w)$ and $\cV(x)$ engage in a standard classical zero knowledge proof protocol which is post-quantum secure.
    \item $\cV$ outputs 1 if the proof in Step 2 is accepting. Otherwise, it outputs 0.
\end{enumerate}


The soundness of the above protocol follows from the soundness of the ZK protocol in Step 2. Also, zero-knowledge property of the above protocol w.r.t classical verifiers follows from the zero-knowledge property of the cQEXT protocol in Step 1 and ZK protocol in Step 2. However, the protocol is not zero-knowledge against a malicious QPT verifier $\cV^*$ for the following simple reason: $\cV^*$ can simply execute the extractor algorithm for the cQEXT protocol in Step 1 and therefore retrieve the witness $w$ completely. The reason it will be able to do so without any issue is because the cQEXT extractor is black-box and straight-line. Hence, a malicious QPT verifier, which does not have any rewinding ability or access to the code of prover, can still execute the extractor algorithm seamlessly.
\section{Tag Amplification: Remaining Analysis}
\label{app:tagamp}
Here, we prove Lemma \ref{lem:tagampsecurity}.

Let $\tagg \in [2^{t/2}]$ denote the tag used by the committer in the left session and $\tagg' \in [2^{t/2}]$ be the tag used by the \mim in the right session. 
Observe that, two sets of decomposed tags $\{ \tagg_i \}_{i \in [t/2]}$ and $\{ \tagg'_i \}_{i \in [t/2]}$, derived from two distinct tags, $\tagg$ and $\tagg'$, 
are such that $\exists \alpha \in [t/2]\; \suchthat \; \forall i \in [t/2]: \tagg'_\alpha \ne \tagg_i$

For any values $u$ (respectively $v$) committed to by $\mathsf{C}$ in the left session, denote by $u'$ (respectively $v'$) the value committed to by the \mim in the right session. Additionally, let $u_i$ (resp, $v_i$) denote the value committed to by $\mathsf{C}$ in the $i^{th}$ parallel execution of \nmcsmall as part of the left commitment and let $u'_i$ (resp, $v'_i$) denote the value committed to by \mim in the $i^{th}$ parallel execution of \nmcsmall as part of the right commitment. 

\noindent The soundness of $\zk$ ensures that when the proof verifies:
\begin{align} \label{eq:zkprop}
    \Pr [ u' \ne u'_\alpha] = \negl(\secp) \text{ and }
    \Pr [ v' \ne v'_{\alpha}] = \negl(\secp)
\end{align}
where the probability is over the randomness of honest verifier, and $\alpha$ denotes the first index in the real (resp., simulated) experiments such that for every $i \in [t/2]$, $\tagg'_\alpha \neq \tagg_i$.
Whenever the proof does not verify, the commitment is not `valid' and $u'$ (resp. $v'$) $ = \bot$.\\

Next, 
recall that in the real world, $\view_{\mim_\secp} \dist{\mathsf{C}(u), \mim(\rho_\secp), \mathsf{R}} (1^\secp, \tagg, \tagg')$ denotes the joint distribution of the view of \mim along with the value $u'$ committed to in the right session when the left committer obtains input $u$. Similarly, $\view_{\mim_\secp} \dist{\mathsf{C}(v), \mim(\rho_\secp), \mathsf{R}} (1^\secp, \tagg, \tagg')$ denotes the joint distribution of the view of \mim along with the value $v'$ committed to in the right session when the right committer obtains input $v$.\\

By Equation (\ref{eq:zkprop}), whenever the \mim's proof verifies, $v'$ can be replaced by $v'_\alpha$ in the distribution $\view_{\mim_\secp} \dist{\mathsf{C}(v), \mim(\rho_\secp), \mathsf{R}} (1^\secp, \tagg, \tagg')$ to yield a statistically indistinguishable distribution $\view'_{\mim_\secp} \dist{\mathsf{C}(v), \mim(\rho_\secp), \mathsf{R}} (1^\secp, \tagg, \tagg')$. Similarly, $u'$ can be replaced by $u'_{\alpha}$ in the distribution $\view_{\mim_\secp} \dist{\mathsf{C}(u), \mim(\rho_\secp), \mathsf{R}} (1^\secp, \tagg, \tagg')$ to yield a statistically indistinguishable distribution $\view'_{\mim_\secp} \dist{\mathsf{C}(u), \mim(\rho_\secp), \mathsf{R}} (1^\secp, \tagg, \tagg')$.\\

It suffices to prove that:
\begin{align}
    \{ \view'_{\mim_\secp} \dist{\mathsf{C}(u), \mim(\rho_\secp), \mathsf{R}} (1^\secp, \tagg, \tagg') \}_{\secp \in \bbN}
    \approx_c
    \view'_{\mim_\secp} \dist{\mathsf{C}(v), \mim(\rho_\secp), \mathsf{R}} (1^\secp, \tagg, \tagg')
\end{align}

To that end, we define the following collections of random variables (each indexed by $\secp$). Each is defined with respect to the adversary $\mim = \{\mim_\secp,\rho_\secp\}_{\secp \in \mathbb{N}}$ that we consider. Throughout, when we say abort, we mean that $\mim^*$ aborts before Step 2, or that the $\mim$ fails to provide an accepting proof.

\begin{itemize}
    \item Let $\mathsf{Pr}^{\mathsf{Abort}}_{\mathsf{u}}(\mim)$ be the probability that $\mim$ aborts in $\view'_{\mim_\secp} \dist{\mathsf{C}(u), \mim(\rho_\secp), \mathsf{R}} (1^\secp, \tagg, \tagg')$.
    \item Let $\mathsf{Pr}^{\mathsf{Abort}}_{\mathsf{v}}(\mim)$ be the probability that $\mim$ aborts in $\view'_{\mim_\secp} \dist{\mathsf{C}(v), \mim(\rho_\secp), \mathsf{R}} (1^\secp, \tagg, \tagg')$.
    \item Let $\view_u(\mim) \coloneqq \{\view'_{\mim_\secp} \dist{\mathsf{C}(u), \mim(\rho_\secp), \mathsf{R}} (1^\secp, \tagg, \tagg')\}_{\lambda \in \bbN} $. 
    \item Let $\view_v(\mim) \coloneqq \{\view'_{\mim_\secp} \dist{\mathsf{C}(v), \mim(\rho_\secp), \mathsf{R}} (1^\secp, \tagg, \tagg')\}_{\lambda \in \bbN}$.
    \item Let $\view_u^{\mathsf{No} \ \mathsf{Abort}}(\mim)$ be the distribution $\view_u(\mim)$ conditioned on there not being an abort.
    \item Let $\view_u^{\mathsf{Abort}}(\mim)$ be the distribution $\view_u(\mim)$ conditioned on there being an abort.
    \item Let $\view_v^{\mathsf{No} \ \mathsf{Abort}}(\mim)$ be the distribution $\view_v(\mim)$ conditioned on there not being an abort.
    \item Let $\view_v^{\mathsf{Abort}}(\mim)$ be the distribution $\view_v(\mim)$ conditioned on there being an abort.
\end{itemize}

The following distributions will not be used explicitly in the hybrids, but will be convenient to define for the proof.

\begin{itemize}
    \item Let $\view_{v,\bot}(\mim)$ be the distribution $\view_v(\mim)$, except that whenever an abort occurs, the distribution outputs $\bot$.
    \item Let $\view^{\mathsf{Abort}}_{v,\bot}(\mim)$ be the distribution $\view_v(\mim)$, except that if an abort \emph{does not} occur, the distribution outputs $\bot$.
    \item Let $\view_{u,\bot}(\mim)$ be the distribution $\view_u(\mim)$, except that whenever an abort occurs, the distribution outputs $\bot$.
    \item Let $\view^{\mathsf{Abort}}_{u,\bot}(\mim)$ be the distribution $\view_u(\mim)$, except that if an abort \emph{does not} occur, the distribution outputs $\bot$.
\end{itemize}
We show that $\view_u(\mim) \approx_c \view_v(\mim)$ via a sequence on hybrids. In particular, we prove: 
\begin{align*}
    \view_v(\mim) &\substack{(1) \\ \equiv \\ \ } (1-\mathsf{Pr}^\mathsf{Abort}_v(\mim))\view_v^{\mathsf{No} \ \mathsf{Abort}}(\mim) + (\mathsf{Pr}^\mathsf{Abort}_v(\mim))\view_v^{\mathsf{Abort}}(\mim)\\
    &\substack{(2) \\ \approx_s \\ \ }(1-\mathsf{Pr}^\mathsf{Abort}_u(\mim))\view_v^{\mathsf{No} \ \mathsf{Abort}}(\mim) + (\mathsf{Pr}^\mathsf{Abort}_u(\mim))\view_v^{\mathsf{Abort}}(\mim)\\
    &\substack{(3) \\ \approx_c \\ \ } (1-\mathsf{Pr}^\mathsf{Abort}_u(\mim))\view_u^{\mathsf{No} \ \mathsf{Abort}}(\mim) + (\mathsf{Pr}^\mathsf{Abort}_u(\mim))\view_v^{\mathsf{Abort}}(\mim)\\
    &\substack{(4) \\ \approx_c \\ \ } (1-\mathsf{Pr}^\mathsf{Abort}_u(\mim))\view_u^{\mathsf{No} \ \mathsf{Abort}}(\mim) + (\mathsf{Pr}^\mathsf{Abort}_u(\mim))\view_u^{\mathsf{Abort}}(\mim)\\
    &\substack{(5) \\ \equiv \\ \ }\view_u(\mim),
\end{align*}

where 
\begin{enumerate}
    \item The equalities $(1)$ and $(5)$ follow by definition.
    \item The indistinguishability $(2)$ follows as a corollary of~\cref{clm:ind-non-abort-ta}. Indeed, $\view_{u,_\bot}(\mim) \approx_c \view_{v,\bot}(\mim)$ in particular implies that the difference in the probability that the $\mim$ aborts in the real interaction versus the simulated interaction is negligible.
    \item The indistinguishability $(3)$ follows as a corollary of~\cref{clm:ind-non-abort-ta}. This can be seen by considering two cases. First, if the probability that the $\mim$ aborts in the real interaction is negligible, then $\view_v^{\mathsf{No} \ \mathsf{Abort}}(\mim) \approx_c \view_u^{\mathsf{No} \ \mathsf{Abort}}(\mim)$ directly follows from~\cref{clm:ind-non-abort-ta}, and the indistinguishability follows. Otherwise, this probability is non-negligible, meaning that $\view_v^{\mathsf{Abort}}(\mim)$ is efficiently sampleable. Thus, a reduction to~\cref{clm:ind-non-abort-ta} can sample from the distribution $\view_v^{\mathsf{Abort}}(\mim)$ whenever it receives $\bot$ from its challenger.\footnote{A more formal analysis of this can be found in~\cite[Lemma~3.2]{BS20}.}
    \item The indistinguishability $(4)$ follows as a corollary
    of~\cref{claim:ind-abort-ta} via a similar analysis as the last step.
\end{enumerate}

\begin{claim}
\label{clm:ind-non-abort-ta}
$$\view_{u,\bot}(\mim) \approx_c \view_{v,\bot}(\mim)$$
\end{claim}
\begin{proof}
We prove this via a sequence of hybrids.
We use $\hyb_k$ to denote the joint distribution of \mim's view (consisting of commitment and proof transcripts along with \mim's state) and the value that \mim commits to in the right session of Hybrid k, using tag $\tagg'_\alpha$, where $\alpha$ denotes the smallest index such that $\tagg'_\alpha \neq \tagg_i$ for every $i \in [t/2]$.\\


\noindent \textbf{$\hyb_1$}: In this hybrid, the challenger executes the simulator 
$\zk.\simulator(1^\secp, x_\secp, \zkV^*_\secp, \sigma_\secp^{(x_\secp)})$  on $\zkV^*_\secp$, which denotes a wrapper around the portion of the \mim that participates in Stage 2 of the protocol, and an instance-advice distribution $(x_\secp, \sigma_\secp^{(x_\secp)})$ defined as follows:
\begin{itemize}
    \item Set the state of $\mim_\secp$ to be $\rho_\secp$.
    \item Execute Stages $0$ and $1$ of the protocol the same way as in the experiment $\view_{u,\bot}(\mim)$, and set $x, w, \cL$ according to \proref{fig:tag_amplification_nmcom} on behalf of $\mathcal{C}$.
    \item Let $\sigma_\secp^{(x_\secp)}$ denote the joint distribution of the protocol transcript, the state of the \mim at the end of Stage $1$, and the value $v'_\alpha$ committed by the \mim with tag $\tagg'_\alpha$.
\end{itemize}
If there is an abort during sampling, or $\zk.\simulator$ causes $\zkV^*_\secp$ to abort (this includes the $\mim$ failing to provide an accepting proof), then output $\bot$.
By Claim~\ref{claim:ind-non-abort}, $$\view_{u,\bot}(\mim) \approx_c \hyb_1$$

\noindent \textbf{$\hyb_2$}: In this hybrid, the challenger behaves identically to $\hyb_1$, except when generating $(x_\secp, \sigma_\secp^{(x_\secp)})$, it replaces the commitment to $u$ with a commitment to $v$ in the first parallel repetition, with $\tagg_1$, of $\nmcsmall$ (while executing all other parallel repetitions the same way as $\hyb_1$). 
If there is an abort during sampling, or $\zk.\simulator$ causes $\zkV^*_\secp$ to abort, then output $\bot$.

We prove in Claim \ref{clm:one-one} that by one-to-one non-malleability of \nmcsmall, for every $u, v \in \{0,1\}^{p(\secp)}$,
$$\hyb_1 \approx_c \hyb_2$$

\noindent \textbf{$\hyb_i$ for $i \in [3,(t/2+1)]$}: In this hybrid, the challenger behaves identically to $\hyb_{i-1}$, except when generating $(x_\secp,\sigma_\secp^{(x_\secp)})$, it replaces the commitment to $u$ with a commitment to $v$ in the $(i - 1)^{th}$ parallel repetition, with tag $\tagg_{i-1}$, of $\nmcsmall$ (while executing all other parallel repetitions the same way as $\hyb_{i-1}$).
If there is an abort during sampling, or $\zk.\simulator$ causes $\zkV^*$ to abort, then output $\bot$.

We prove in Claim \ref{clm:one-one} that by one-to-one non-malleability of $\nmcsmall$, for every $u, v \in \{0,1\}^{p(\secp)}$ and every $i \in [3,t/2+1]$,
$$\hyb_{i-1} \approx_c \hyb_i$$

Finally, by claim \ref{claim:ind-non-abort}, we have that 
$$\hyb_{(t/2+1)} \approx_c \view_{u,\bot}(\mim)$$

Next, we state and prove Claim \ref{clm:one-one}.

\begin{claim} \label{clm:one-one}
For all $u, v \in \{0, 1\}^{p(\secp)}$ and all $i \in [2,t/2+1]$,
 \begin{align}
     \hyb_{i} \approx_c \hyb_{i-1} 
 \end{align}
\end{claim}

\begin{proof}
Suppose Claim \ref{clm:one-one} is false. Then there exists values $(u, v)$, some $i \in [2,t/2+1]$ and a polynomial $\poly(\cdot)$ such that for infinitely many $\secp \in \mathbb{N}$,
\begin{align} \label{eq:mimdistinguisher-i}
    \abs{
    \Pr[\mim(\hyb_i) = 1] -
    \Pr[\mim(\hyb_{i-1}) = 1]
    } \ge 
    \frac{1}{\poly(\secp)}
\end{align}
We will demonstrate an adversary $\mim^\beta$ that contradicts the non-malleability of \nmcsmall according to Definition \ref{def:nmc}, \ie we will show that
for infinitely many $\secp \in \bbN$,  
\begin{align}
& \Big|
\Pr[\mim^\beta \big( \view_{\mim^\beta_\secp} \dist{\mathsf{C}(v), \mim^\beta(\rho_\secp), \mathsf{R}} (1^\secp, \tagg_{i-1}, \tagg'_\alpha) \big) = 1] \nonumber \\
& - \Pr[\mim^\beta \big( \view_{\mim^\beta_\secp} \dist{\mathsf{C}(u), \mim^\beta(\rho_\secp), \mathsf{R}} (1^\secp, \tagg_{i-1}, \tagg'_\alpha) \big) = 1]
\Big|
 \ge \frac{1}{\poly(\secp)}
\end{align}
where the two distributions $\mim^\beta \big( \view_{\mim^\beta_\secp} \dist{\mathsf{C}(v), \mim^\beta(\rho_\secp), \mathsf{R}} (1^\secp, \tagg_{i-1}, \tagg'_\alpha) \big) = 1$ and \\
$\mim^\beta \big( \view_{\mim^\beta_\secp} \dist{\mathsf{C}(u), \mim^\beta(\rho_\secp), \mathsf{R}} (1^\secp, \tagg_{i-1}, \tagg'_\alpha) \big) = 1$ correspond to honest commitments to $v$ and $u$ respectively, for the small-tag commitment scheme.\\

\noindent $\mim^\beta$ is defined as follows:
\begin{enumerate}
\item Obtain input values $v,u$, and begin an interaction with a challenger for $\nmcsmall$.
\item Emulate the role of honest committer and honest receiver in an interaction with $\mim$ 
executing \piagk. In more detail, in the role of a committer in a left session, participate in a session of \piagk with $\mim$ as receiver. At the same time, play the role of the receiver in a right session with $\mim$ as committer.
Recall that \piagk contains $k$ repetitions of $\nmcsmall$ and $\zk$. 
\item In the left session, embed the challenger's messages in the $(i-1)^{th}$ instance of \nmcsmall, and forward the response of $\mim$ corresponding to the $(i-1)^{th}$ instance to the challenger. Execute remaining instances according to the strategy in $\hyb_{i-1}$.
\item In the right session, forward the message obtained from 
$\mim$ in the $\alpha^{th}$ instance of \nmcsmall to the challenger, and embed the challenger's response for that round as receiver message in the $\alpha^{th}$ instance. 
Use honest receiver strategy for all other instances of $\nmcsmall$ in the right session. 
\item Obtain value $v'_\alpha$ from the challenger of the non-malleable commitment (representing the value in the commitment sent by $\mim^\beta$ to the challenger on the right).
\item Use the transcript, the obtained value $v'_\alpha$ and the state of $\mim$ to define the instance-advice sample, and then execute $\zk.\simulator(1^\secp, x_\secp, \zkV^*_\secp, \sigma_\secp^{(x_\secp)})$.
\item If an abort occurs at any point, output $\bot$.
\end{enumerate}


We now analyze the probability that $\mim^\beta$ successfully contradicts Definition \ref{def:nmc}. To this end, we note that:
\begin{align}
    \Pr[\mim^\beta \big( \view_{\mim^\beta_\secp} \dist{\mathsf{C}(u), \mim^\beta(\rho_\secp), \mathsf{R}} (1^\secp, \tagg_{i-1}, \tagg'_\alpha) \big) = 1] = \Pr[\mim(\hyb_{i-1}) = 1] \\
    \Pr[\mim^\beta \big( \view_{\mim^\beta_\secp} \dist{\mathsf{C}(v), \mim^\beta(\rho_\secp), \mathsf{R}} (1^\secp, \tagg_{i-1}, \tagg'_\alpha) \big) = 1] = \Pr[\mim(\hyb_i) = 1]
\end{align}
Therefore, for infinitely many $\secp \in \mathbb{N}$,
\begin{align*}
\Big|
\Pr[\mim^\beta \big( \view_{\mim^\beta_\secp} \dist{\mathsf{C}(u), \mim^\beta(\rho_\secp), \mathsf{R}} (1^\secp, \tagg_{i-1}, \tagg'_\alpha) \big) = 1] \\
- \Pr[\mim^\beta \big( \view_{\mim^\beta_\secp} \dist{\mathsf{C}(v), \mim^\beta(\rho_\secp), \mathsf{R}} (1^\secp, \tagg_{i-1}, \tagg'_\alpha) \big) = 1]
\Big| = \\ 
\Big|
    \Pr[\mim(\hyb_{i-1}) = 1] -
    \Pr[\mim(\hyb_i) = 1]
\Big| \ge
    \frac{1}{\poly(\secp)}
\end{align*}
which is a contradiction, as desired.
\end{proof}

This completes the proof of~\cref{clm:ind-non-abort-ta}.
\end{proof}

\begin{claim}
\label{claim:ind-abort-ta}
$$\view_{u,\bot}^{\mathsf{Abort}}(\mim) \approx_c \view_{v,\bot}^{\mathsf{Abort}}(\mim)$$
\end{claim}
\begin{proof}
We prove this via a sequence of hybrids. We use $\hyb_k$ to denote the joint distribution of the $\mim$'s view (consisting of commitment and proof transcripts along with the \mim's state) in Hybrid $k$.\\

\noindent $\hyb_1:$ In this hybrid, the challenger executes one iteration of the simulator $\zk.\zkSimAbort(1^\secp, x_\secp, \zkV^*_\secp, \sigma_\secp^{(x_\secp)})$ on $\zkV^*_\secp$, where $\zkV^*_\secp$ denotes the portion of the \mim that participates in Stage 2 of the protocol, and an instance-advice distribution $(x_\secp, \sigma_\secp^{(x_\secp)})$ defined as follows:
\begin{itemize}
    \item Set the state of $\mim_\secp$ to be $\rho_\secp$.
    \item Execute Stage $1$ of the protocol the same way as in the experiment $\mathsf{real}(\mim)$, and set $x, w, \cL$ according to \proref{fig:tag_amplification_nmcom} on behalf of $\mathcal{C}$.
    \item If an abort occurs, output the transcript and state of the \mim until the abort.
    \item Otherwise, let $\sigma_\secp^{(x_\secp)}$ denote the joint distribution of the protocol transcript and the state of the \mim at the end of Stage $1$.
\end{itemize}
If $\zk.\zkSimAbort(1^\secp, x_\secp, \zkV^*_\secp, \sigma_\secp^{(x_\secp)})$ outputs a non-aborting transcript and state, output $\bot$, otherwise return the output of $\zk.\zkSimAbort(1^\secp, x_\secp, \zkV^*_\secp, \sigma_\secp^{(x_\secp)})$.
By Claim \ref{claim:ind-abort},
$$\view_{u,\bot}^{\mathsf{Abort}}(\mim) \approx_c \hyb_1$$

\noindent \textbf{$\hyb_2$}: In this hybrid, the challenger behaves identically to $\hyb_1$, except when generating $(x_\secp, \sigma_\secp^{(x_\secp)})$, it replaces the commitment to $u$ with a commitment to $v$ in the first parallel repetition, with $\tagg_1$, of \nmcsmall (while executing all other parallel repetitions the same way as $\hyb_1$). 
If a non-aborting transcript is produced, then output $\bot$.

We prove in \cref{claim:one-hiding} that by hiding of \nmcsmall, for every $u, v \in \{0,1\}^{p(\secp)}$,
$$\hyb_1 \approx_c \hyb_2$$

\noindent \textbf{$\hyb_i$ for $i \in [3,(t/2+1)]$}: In this hybrid, the challenger behaves identically to $\hyb_{i-1}$, except when generating $(x_\secp,\sigma_\secp^{(x_\secp)})$, it replaces the commitment to $u$ with a commitment to $v$ in the $(i - 1)^{th}$ parallel repetition, with tag $\tagg_{i-1}$, of \nmcsmall (while executing all other parallel repetitions the same way as $\hyb_{i-1}$).
If a non-aborting transcript is produced, then output $\bot$.
We prove in \cref{claim:one-hiding} that by hiding of \nmcsmall, for every $u, v \in \{0,1\}^{p(\secp)}$ and every $i \in [3,t/2+1]$,
$$\hyb_{i-1} \approx_c \hyb_i$$

Finally, by claim \ref{claim:ind-abort}, we have that 
$$\hyb_{(t/2+1)} \approx_c \view_{v,\bot}(\mim)$$

Next, we state and prove \cref{claim:one-hiding}.

\begin{claim}
\label{claim:one-hiding}
For every $u, v \in \{0,1\}^{p(\secp)}$ and all $i \in [2,t/2+1]$,
$$\hyb_i \approx_c \hyb_{i-1}$$
\end{claim}

\begin{proof}
Suppose \cref{claim:one-hiding} is false. Then there exist $u, v \in \{0,1\}^{p(\secp)}$, some $i \in [2,t/2+1]$, a PPT distinguisher $D$ and a polynomial $\poly(\cdot)$ such that for infinitely many $\secp \in \mathbb{N}$,
\begin{align} 
    \abs{
    \Pr[D(\hyb_i) = 1] -
    \Pr[D(\hyb_{i-1}) = 1]
    } \ge 
    \frac{1}{\poly(\secp)}
\end{align}
We will demonstrate a receiver that contradicts the hiding property of $\nmcsmall$, \ie we will show that there exists $\eR^*$ such that for infinitely many $\secp \in \mathbb{N}$,
\begin{align} 
    \abs{
    \Pr[\eR^*(\nmcsmall\langle \eC(u), \eR^* \rangle) = 1] -
     \Pr[\eR^*(\nmcsmall\langle \eC(v), \eR^* \rangle) = 1]
    } \ge 
    \frac{1}{\poly(\secp)}
\end{align}
$\eR^*$ obtains input $u, v$, and begins an interaction with a challenger for the hiding of $\nmcsmall$.
It then emulates the role of honest committer and honest receiver in an interaction with \mim, executing \piagk. In the left session, it embeds the challenger's messages in the $(i-1)^{th}$ instance of $\nmcsmall$, and forwards the response of \mim corresponding to the $(i-1)^{th}$ instance to the challenger. It executes the remaining instances in the left session, and all instances of the right session according to the strategy in $\hyb_{i-1}$. Next, it uses the transcript and state of the \mim to define the instance-advice sample, and executes one iteration of $\zk.\zkSimAbort(1^\secp,x_\secp,\zkV^*_\secp, \sigma_\secp^{(x_\secp)})$.

If an abort occurs at some point, then $\eR^*$ runs $D(\tau, \mathsf{st})$ where $\tau$ and $\mathsf{st}$ denote the transcript and the state of the adversary until the point in the protocol where the abort occurs.
If no abort occurs throughout the protocol, then $\eR^*$ outputs $0$.

We now analyze the probability that $\eR^*$ successfully contradicts Definition \ref{def:nmc}. Here, we note that:
$$\Pr[\eR^*(\nmcsmall\langle \eC(u), \eR^* \rangle) = 1] = \Pr[D(\hyb_{i-1}) = 1] \text{ and, }$$
$$\Pr[\eR^*(\nmcsmall\langle \eC(v), \eR^* \rangle) = 1] = \Pr[D(\hyb_{i}) = 1]$$
Therefore, for infinitely many $\secp \in \mathbb{N}$,
\begin{align*}
& \Big|\Pr[\eR^*(\nmcsmall \langle \eC(u), \eR^* \rangle) = 1] -  \Pr[\eR^*(\nmcsmall \langle \eC(v), \eR^* \rangle) = 1] \Big| \\
& =\Big| \Pr[D(\hyb_{i-1}) = 1] - \Pr[D(\hyb_{i}) = 1] \Big| \geq \frac{1}{\poly(\secp)}
\end{align*}
which is a contradiction, as desired. This completes the proof of \cref{claim:one-hiding}.
\end{proof}

This completes the proof of~\cref{claim:ind-abort-ta}.
\end{proof}

Together, these claims complete the proof of Lemma \ref{lem:tagampsecurity}.
\section{Multi-Committer Extractable Commitments against Arbitrary Distinguishers}
\label{sec:fullecom}

Recall that in the setting of multi-committer extractable commitments~\ref{defn:qspec}, we only considered computational indistinguishability against any \emph{compliant} non-uniform polynomial-size quantum distinguisher $\D$. We will now demonstrate how to upgrade any multi-committer extractable commitment secure against any compliant non-uniform polynomial-size quantum distinguisher to one which is secure against against any \emph{arbitrary} non-uniform polynomial-size quantum distinguisher.
The resulting commitment admits an extractor that makes use of the Quantum Rewinding lemma \ref{lemma:rewinding} to successfully generate both explainable and non-explainable transcripts.

\paragraph{Construction.} Let $\PECom$ denote any multi committter extractable commitment protocol which admits an extractor $\cE$. 
Consider a modified version of the protocol $\PECom'$ which is identical to $\PECom$ except that at the very end, 
each of the committers $\{\eC_i\}_{i \in [n]}$ sends a constant-round ZK argument to the receiver attesting to the fact that the committer messages were explainable. 
If the verification check passes for all the arguments, the receiver accepts all the commitments.
Otherwise it rejects all the commitments, and the committed value is set to $\bot$. We denote this by \reject.

\paragraph{Analysis.}
Assuming $\PECom$ admits an extractor $\cE$ which satisfies the extractability property against compliant distinguishers, we will construct an extractor $\cE'$ for $\PECom'$ which satisifies the extractability property against arbitrary distinguishers. The extractor $\cE'$ consists of a randomized extractor $\cE_{\comb}$ will consist of two sub-extractors, namely $\cE_r$ and $\cE_{nr}$. The purpose of $\cE_r$ is to simulate a transcript which generates a \reject whereas the purpose of $\cE_{nr}$ is to simulate a transcript which does not generate a \reject. 

At a high level, $\cE_{\comb}$ will randomly call one of the two sub-extractors and try to produce a transcript which is indistinguishable from the real view. Looking ahead, this will result in the $\cE_{\comb}$ outputting a quantum state \output that is indistinguishable from the real verifier output conditioned on $\output \ne \fail$. Furthermore, $\output \ne \fail$ will occur with probability negligibly close to $1/2$ (due to random choice of executing either $\cE_{r}$ or $\cE_{nr}$ and the computational indistinguishability of the view generated by $\cE_{r}$ and $\cE_{nr}$). In other words, $\cE_{\comb}$ is going to succeed in extraction only with probability (negligibly close to) $1/2$. Once we have this, we can apply Watrous' quantum rewinding lemma \ref{lemma:rewinding} to amplify the success probability from $\approx 1/2$ to $\approx 1$. \\

We will now show the construction of $\cE_{nr}, \cE_{r}, \cE_{\comb}$ and finally $\cE'$. \\

$\cE_{nr}(1^\secp, 1^n, I, \eC^*_\secp, \rho)$:
\begin{enumerate}
    \item Execute the extractor $\cE(1^\secp, 1^n, I, C^*_\secp, \rho_\secp)$ on the adversary $\eC^*_\secp$ (which controls a subset $I$ of committers).
    
    \item Participate as a honest verifier in $|I|$ ZK argument sessions with $\eC^*_\secp$ where $\eC^*_\secp$ sends messages on behalf of the prover.
    
    \item Execute the verification algorithm on all $|I|$ argument transcripts. If verification check passes for all $|I|$ arguments, then accept all the commitments. Otherwise, say \reject occurs.
    
    \item If \reject occurs, discard all information saved so far and output \fail. Otherwise output $C^*_\secp$'s inner state and the extracted value.
\end{enumerate}

\vspace{1cm}

$\cE_{r}(1^\secp, 1^n, I, \eC^*_\secp, \rho)$:
\begin{enumerate}
    \item Interact with $\eC^*_\secp$ as an honest receiver of $\PECom'$.
        If the verification check fails for some ZK argument transcript, count it as a \reject.
        
    \item If \reject does not occur, discard all information saved so far and output \fail. Otherwise output $C^*_\secp$'s inner state, the transcript and $\bot$ as the extracted value.
\end{enumerate}


$\cE_{\comb}(1^\secp, 1^n, I, \eC^*_\secp, \rho)$: Sample $b \xleftarrow{\$} \{ r, nr \}$ and execute $\cE_{b}$.\\


$\cE'(1^\secp, 1^n, I, \eC^*_\secp, \rho)$: 
\begin{enumerate}
    \item Generate the circuit $\cE_{\comb, \eC^*_\secp}$ which is the circuit implementation of $\cE_{\comb}$ with hardwired input $\eC^*_\secp$, that is, the only input to $\cE_{\comb, \eC^*_\secp}$ is the quantum state $\rho$.
    
    \item Let $\R$ be the algorithm from Lemma \ref{lemma:rewinding}. The output of the extractor is $\R(\cE_{\comb, \eC^*_\secp}, \rho, \lambda)$
\end{enumerate}

The following claim is similar to the definition of multi-committer extractability stated in Section \ref{defn:qspec} but generalized to handle \emph{arbitrary} distinguishers (instead of just compliant ones).

\begin{claim}
For any \emph{arbitrary} non-uniform polynomial-size quantum distinguisher $\D' = \{\D'_\secp,\sigma_\secp\}_{\secp \in \bN}$, there exists a negligible function $\mu(\cdot)$, such that for all adversaries $\eC^* = \{\eC^*_\secp,\rho_\secp\}_{\secp \in \mathbb{N}}$ representing a subset of $n$ committers, namely, $\{\eC_i\}_{i \in I}$ for some set $I \subseteq [n]$, the following holds for all polynomial-size sequences of inputs $\{\{m_{i,\secp}\}_{i \notin I}\}_{\secp \in \mathbb{N}}$ and $\secp \in \bN$.
    \begin{align*}&\big|\Pr[\D'_\secp(\mathsf{VIEW'}^{\mathsf{msg}}_{\eC^*_\secp}(\dist{\eR,\eC^*_\secp(\rho_\secp),\{\eC_i(m_{i,\secp})\}_{i \notin I}}(1^\secp,1^n)),\sigma_\secp) = 1]\\ &- \Pr[\D'_\secp( \cE'(1^\secp,1^n,I,\eC^*_\secp,\rho_\secp),\sigma_\secp)=1]\big| \leq \mu(\secp).
    \end{align*}
\end{claim}

Here $\mathsf{VIEW'}^{\mathsf{msg}}_{\eC^*_{\secp}}(\dist{\eR,\eC^*_\secp(\rho_\secp),\{\eC_i(m_i)\}_{i \notin I}}(1^\secp,1^n))$ is defined to consist of the following:

\begin{itemize}
    \item The view of $\eC^*_\secp$ on $\PECom'$ interaction with the honest receiver $\eR$ and set $\{\eC_i(m_i)\}_{i \notin I}$ of honest parties; this view includes a set of transcripts $\{\tau_i\}_{i \in I}$ and a state $\state$.
    
    \item A set of strings $\{m_i\}_{i \in I}$, where each $m_i$ is defined relative to $\tau_i$ as follows. If there exists $m'_i,r_i$ such that $\eR(1^\secp,\tau_i,m'_i,r_i) = 1$, then $m_i = m'_i$, otherwise, $m_i = \bot$.
\end{itemize}

We will prove the above claim in several steps:
\begin{enumerate}
    \item Simulating non \reject interactions using $\cE_{nr}$
    \item Simulating \reject interations using $\cE_{r}$
    \item Applying Watrous rewinding lemma on the combined extractor $\cE_{\comb}$
\end{enumerate}

\noindent First, we introduce some notation:

\begin{itemize}
    
    

    
    
    \item Let $\cE_{r, \bot}$ be the same distribution as $\cE_{r}(1^\secp,1^n,I,\eC^*_\secp,\rho_\secp)$ except that whenever a \reject does not occur, the distribution output is $\bot$
    
    \item Let $\cE_{nr, \bot}$ be the same distribution as $\cE_{nr}(1^\secp,1^n,I,\eC^*_\secp,\rho_\secp)$ except that whenever a \reject occurs, the distribution output is $\bot$
    
    \item Let $\view'_{r, \bot}$ be the same distribution as $\view'^{\mathsf{msg}}_{\eC^*_\secp}(\dist{\eR,\eC^*_\secp(\rho_\secp),\{\eC_i(m_{i,\secp})\}_{i \notin I}}(1^\secp,1^n))$ except that whenever a \reject does not occur, the distribution output is $\bot$
    
    \item Let $\view'_{nr, \bot}$ be the same distribution as $\view'^{\mathsf{msg}}_{\eC^*_\secp}(\dist{\eR,\eC^*_\secp(\rho_\secp),\{\eC_i(m_{i,\secp})\}_{i \notin I}}(1^\secp,1^n))$ except that whenever a \reject occurs, the distribution output is $\bot$
    
\end{itemize}

We use $\tau'$ to denote the input to the distinguisher $\D'$ where $\tau'$ can either be  $\cE_{\{r, nr \}, \bot}$ or $\view'_{\{r, nr \}, \bot}$.

\begin{claim} \label{claim:no-reject}
For any \emph{arbitrary} non-uniform polynomial-size quantum distinguisher $\D' = \{\D'_\secp,\sigma_\secp\}_{\secp \in \bN}$, there exists a negligible function $\mu(\cdot)$, such that for all adversaries $\eC^* = \{\eC^*_\secp,\rho_\secp\}_{\secp \in \mathbb{N}}$ representing a subset of $n$ committers, namely, $\{\eC_i\}_{i \in I}$ for some set $I \subseteq [n]$, the following holds for all polynomial-size sequences of inputs $\{\{m_{i,\secp}\}_{i \notin I}\}_{\secp \in \mathbb{N}}$ and $\secp \in \bN$.
    \begin{align*}&\big|\Pr[\D'_\secp(\view'_{nr, \bot},\sigma_\secp) = 1] - \Pr[\D'_\secp( \cE_{nr, \bot},\sigma_\secp)=1]\big| \leq \mu(\secp).
    \end{align*}
\end{claim}

\begin{proof}
Conditioned on the event that \reject happens, both distributions ($\view'_{nr, \bot}$ and  $\cE_{nr, \bot}$) output $\bot$ by definition. Therefore, in such a case, these two distributions will be prefectly indistinguishable.\\
Conditioned on the event that \reject does not happen, we can say that $\tau'$ is explainable with overwhelming probability. This holds due to the soundness of ZK. Having said that, we now prove that $\D'$ cannot distinguish between real and simulated $\tau'$.
    
    Suppose there exists an \emph{arbitrary} non-uniform polynomial-size quantum distinguisher $\D' = \{\D'_\secp,\sigma_\secp\}_{\secp \in \bN}$, a polynomial function $\poly(\cdot)$, $\eC^* = \{\eC^*_\secp,\rho_\secp\}_{\secp \in \mathbb{N}}$ representing a subset of $n$ committers, namely, $\{\eC_i\}_{i \in I}$ for some set $I \subseteq [n]$, s.t. the following holds for inifintely many polynomial-size sequence of input $\{\{m_{i,\secp}\}_{i \notin I}\}_{\secp \in \mathbb{N}}$ and $\secp \in \bN$.
    \begin{align*}&\big|\Pr[\D'_\secp(\view'_{nr, \bot},\sigma_\secp) = 1 | \neg \reject] - \Pr[\D'_\secp( \cE_{nr, \bot},\sigma_\secp)=1 | \neg \reject]\big| \geq 1/\poly(\secp).
    \end{align*}
    
    We can use $\D'$ to build a compliant distinguisher $\D$ which contradicts the multi-committer extractability of $\PECom$ as per Definition \ref{defn:qspec}. $\D$ first obtains a $\PECom$ transcript $\tau$ (and adversary's state) as a challenge. It then interacts as an honest verifier with $C^*_\secp$ in $|I|$ ZK argument sessions where $C^*_\secp$ proves that $\tau$ is explainable. $\D$ then {\em verifies the ZK argument} and outputs $0$ if the argument rejects. Otherwise, it forwards the entire transcript along with $C^*_\secp$'s internal state and the value inside commitment to $\D'$.
    If $\D'$ returns 1, $\D$ returns 1. Otherwise $\D$ returns 0.
    
    Note that since non-explainable transcripts that are not rejected occur with negligible probability (due to soundness of ZK), the probability that $\D$ outputs 1 on receiving a non-explainable transcript is negligible. Therefore, $\D$ is a compliant distinguisher. Moreover, the following holds:
    \begin{align*}
        & \Pr[\D_\secp(\mathsf{VIEW}^{\mathsf{msg}}_{\eC^*_\secp}(\dist{\eR,\eC^*_\secp(\rho_\secp),\{\eC_i(m_{i,\secp})\}_{i \notin I}}(1^\secp,1^n)),\sigma_\secp) = 1] = \Pr[\D'_\secp(\view'_{nr, \bot},\sigma_\secp) = 1 | \neg \reject] \text{ and }\\
        & \Pr[\D_\secp( \cE(1^\secp,1^n,I,\eC^*_\secp,\rho_\secp),\sigma_\secp)=1]
        =
       \Pr[\D'_\secp( \cE_{nr, \bot},\sigma_\secp)=1 | \neg \reject]
    \end{align*}
    Therefore, 
     \begin{align*}
     & \big|
     \Pr[\D_\secp(\mathsf{VIEW}^{\mathsf{msg}}_{\eC^*_\secp}(\dist{\eR,\eC^*_\secp(\rho_\secp),\{\eC_i(m_{i,\secp})\}_{i \notin I}}(1^\secp,1^n)),\sigma_\secp) = 1]\\
     & -
     \Pr[\D_\secp( \cE(1^\secp,1^n,I,\eC^*_\secp,\rho_\secp),\sigma_\secp)=1]
     \big|
     \geq 1/\poly(\secp).
    \end{align*}
    which gives a contradiction.
\end{proof}

\begin{claim} \label{claim:reject}
For any \emph{arbitrary} non-uniform polynomial-size quantum distinguisher $\D' = \{\D'_\secp,\sigma_\secp\}_{\secp \in \bN}$, there exists a negligible function $\mu(\cdot)$, such that for all adversaries $\eC^* = \{\eC^*_\secp,\rho_\secp\}_{\secp \in \mathbb{N}}$ representing a subset of $n$ committers, namely, $\{\eC_i\}_{i \in I}$ for some set $I \subseteq [n]$, the following holds for all polynomial-size sequences of inputs $\{\{m_{i,\secp}\}_{i \notin I}\}_{\secp \in \mathbb{N}}$ and $\secp \in \bN$.
    \begin{align*}&\big|\Pr[\D'_\secp(\view'_{r, \bot},\sigma_\secp) = 1] - \Pr[\D'_\secp( \cE_{r, \bot},\sigma_\secp)=1]\big| \leq \mu(\secp).
    \end{align*}
\end{claim}

\begin{proof}
Conditioned on the event that \reject does not happen, both distributions ($\view'_{r, \bot} and  \cE_{r, \bot}$) output $\bot$ by definition. Therefore, in such a case, these two distributions will be prefectly indistinguishable.\\
Conditioned on the event that \reject happens, the distribution produced by $\cE_{r}$ is identical to the distribution produced by an honest reveiver. This holds due to the construction of $\cE_{r}$
\end{proof}

To complete our proof, we will introduce some additional notation:
\begin{enumerate}
    \item Let $\Pr^\reject_{\view'}$ be the probability that a \reject happens in $\view'$
    
    \item Let $\Pr^\reject_{\cE_r}$ be the probability that a \reject happens in $\cE_r$
    
    \item Let $\Pr^\reject_{\cE_{nr}}$ be the probability that a \reject happens in $\cE_{nr}$
\end{enumerate}

Following~\cite{BS20}, we now show that $\view' = \cE_{\comb}$ via a sequence on hybrids. In particular, we show that:

\begin{align*}
   \view' &\substack{(1) \\ \equiv \\ \ } (\view' | \reject) \mathsf{Pr}^\reject_{\view'} + (\view' | \noreject) (1- \mathsf{Pr}^\reject_{\view'})\\
      &\substack{(2) \\ \approx_c \\ \ } (\view' | \reject) \mathsf{Pr}^\reject_{\cE_r} + (\view' | \noreject) (1 - \mathsf{Pr}^\reject_{\cE_{nr}})\\
      &\substack{(3) \\ \approx_c \\ \ } (\view' | \reject) \mathsf{Pr}^\reject_{\cE_r} + (\cE_{nr} | \noreject) (1 - \mathsf{Pr}^\reject_{\cE_{nr}})\\
   &\substack{(4) \\ \approx_c \\ \ } (\cE_r | \reject) \mathsf{Pr}^\reject_{\cE_r} + (\cE_{nr} | \noreject) (1 - \mathsf{Pr}^\reject_{\cE_{nr}})\\
   &\substack{(5) \\ \approx_s \\ \ }\cE_{\comb}
\end{align*}

where 

\begin{enumerate}
    \item The equality $(1)$ follows by definition.

    \item The indistinguishability $(2)$ follows as Corollary of Claim \ref{claim:reject} and Claim \ref{claim:no-reject}. Indeed, $\view'_{r, \bot} \approx_c \cE_{r, \bot}$ in particular implies that the difference in probability that \reject happens in the real interaction versus the simulated interaction is negligible, and likewise for $\view'_{nr, \bot} \approx_c \cE_{nr, \bot}$.
    
    \item The indistinguishability $(3)$ follows as a corollary of Claim \ref{claim:no-reject}. This can be seen by considering two cases. First, if the probability that \reject happens in the real interaction is negligible, then $\view' | \noreject \approx_c \cE_{nr} | \noreject$ directly follows from Claim \ref{claim:no-reject}, and the indistinguishability follows. Otherwise, this probability is non-negligible, meaning that $\view' | \reject$ is efficiently sampleable. Thus, a reduction to \ref{claim:no-reject} can sample from the distribution $\view' | \reject$ whenever it receives $\bot$ from the challenger.
    
    \item The indistinguishability $(4)$ follows as a corollary of Claim \ref{claim:reject} via a similar analysis as the last step.
    
    \item The indistinguishability $(5)$ follows from the definition of $\cE_{\comb}$
\end{enumerate}




Also, by an analysis similar to \cite{BS20} Corollary 3.1, 3.2, we can say that the sucess probability of $\cE_{\comb}$ is negligibly close to $1/2$ and therefore the success probability of $\cE_{\comb}$ is \emph{input-oblivious}.

Now we can apply the Quantum rewinding lemma \ref{lemma:rewinding} to amplify the success probability from $\approx 1/2$ to $\approx 1$ following an analysis similar to \cite{BS20}. Consider the quantum circuit $\cE_{\comb, \eC^*}$ which is the circuit implementation of $\cE_{\comb}$ with hardwired input $\eC^*_\secp$, that is, the only input to $\cE_{\comb, \eC^*_\secp}$ is the quantum state $\rho$. By denoting the success probability for input $\rho$ by $p(\rho)$ and setting $\epsilon \coloneqq \negl(\secp) + 2^{-\secp \cdot \frac{3}{4}},p_0 \coloneqq 1/4,$ and $q \coloneqq 1/2$, we can satisfy all the conditions for Quantum Rewinding Lemma \ref{lemma:rewinding}.

This implies that trace distace between $\R(\cE_{\comb, C^*}, \rho, \lambda)$ and $\cE_{\comb}$ is bounded by a negligible function. Therefore, our final extractor $\cE'(1^\secp, 1^n, I, \eC^*_\secp, \rho) = \R(\cE_{\comb, \eC^*_\secp}, \rho, \lambda)$ completes the extraction successfully with probability negligibly close to 1.

\end{document}